\documentclass[prodmode,acmtocl]{acmsmall}

\usepackage{latexsym, amsmath, amssymb}
\usepackage{times}
\usepackage{xspace}
\usepackage{url}
\usepackage[only bigsqcap]{stmaryrd} 

\usepackage{tikz}
\usetikzlibrary{shapes,snakes,arrows,positioning,calc,decorations.markings}

\newcommand{\DL}{\textsl{DL-Lite}}

\newcommand{\bDLa}{\ensuremath{\DL_{\alpha}}}
\newcommand{\bDLb}{\ensuremath{\DL_\textit{bool}}}
\newcommand{\bDLk}{\ensuremath{\DL_\textit{krom}}}
\newcommand{\bDLc}{\ensuremath{\DL_\textit{core}}}

\newcommand{\DLan}{\ensuremath{\smash{\DL_{\alpha}^\mathcal{N}}}}
\newcommand{\DLbn}{\ensuremath{\smash{\bDLb^\mathcal{N}}}}
\newcommand{\DLkn}{\ensuremath{\smash{\bDLk^\mathcal{N}}}}
\newcommand{\DLcn}{\ensuremath{\smash{\bDLc^\mathcal{N}}}}

\newcommand{\DLbrn}{{\ensuremath{\bDLb^{\mathcal{H}\mathcal{N}}}}}

\newcommand{\rDLaHN}{\ensuremath{\bDLa^{(\mathcal{HN})}}}
\newcommand{\rDLbHN}{\ensuremath{\bDLb^{(\mathcal{HN})}}}
\newcommand{\rDLkHN}{\ensuremath{\bDLk^{(\mathcal{HN})}}}
\newcommand{\rDLcHN}{\ensuremath{\bDLc^{(\mathcal{HN})}}}

\newcommand{\PTL}{\ensuremath{\mathcal{PTL}}}
\newcommand{\coreLTL}{\ensuremath{\PTL_{\textit{core}}}}
\newcommand{\kromLTL}{\ensuremath{\PTL_{\textit{krom}}}}
\newcommand{\boolLTL}{\ensuremath{\PTL_{\textit{bool}}}}
\newcommand{\hornLTL}{\ensuremath{\PTL_{\textit{horn}}}}

\newcommand{\TuDLbn}{\ensuremath{\smash{\textsl{T}_{\mathcal{US}}\DLbn}}}%

\newcommand{\TdDLan}{\ensuremath{\smash{\textsl{T}_{F\!P}\DLan}}}%
\newcommand{\TdDLbn}{\ensuremath{\smash{\textsl{T}_{F\!P}\DLbn}}}%
\newcommand{\TdDLkn}{\ensuremath{\smash{\textsl{T}_{F\!P}\DLkn}}}%
\newcommand{\TdDLcn}{\ensuremath{\smash{\textsl{T}_{F\!P}\DLcn}}}%

\newcommand{\TdxDLan}{\ensuremath{\smash{\textsl{T}_{F\!P\!X}\DLan}}}
\newcommand{\TdxDLbn}{\ensuremath{\smash{\textsl{T}_{F\!P\!X}\DLbn}}}%
\newcommand{\TdxDLkn}{\ensuremath{\smash{\textsl{T}_{F\!P\!X}\DLkn}}}%
\newcommand{\TdxDLcn}{\ensuremath{\smash{\textsl{T}_{F\!P\!X}\DLcn}}}%

\newcommand{\zSVDLan}{\ensuremath{\smash{\textsl{T}_{U}\DLan}}}%
\newcommand{\zSVDLb}{\ensuremath{\smash{\textsl{T}_{U}\DLbn}}}%
\newcommand{\zSVDLk}{\ensuremath{\smash{\textsl{T}_{U}\DLkn}}}%
\newcommand{\zSVDLc}{\ensuremath{\smash{\textsl{T}_{U}\DLcn}}}%

\newcommand{\TxDLbn}{\ensuremath{\smash{\textsl{T}^{\,\ast}_{X}\DLbn}}}%
\newcommand{\TdrDLbn}{\ensuremath{\smash{\textsl{T}^{\,\ast}_{F\!P}\DLbn}}}%
\newcommand{\TurDLbn}{\ensuremath{\smash{\textsl{T}^{\,\ast}_{U}\DLbn}}}%

\newcommand{\TdxrDLcn}{\ensuremath{\smash{\textsl{T}^{\,\ast}_{F\!P\!X}\DLcn}}}%
\newcommand{\TdxrDLkn}{\ensuremath{\smash{\textsl{T}^{\,\ast}_{F\!P\!X}\DLkn}}}%

\newcommand{\ALC}{\ensuremath{\mathcal{ALC}}}
\newcommand{\ALCQI}{\ensuremath{\mathcal{ALCQI}}}
\newcommand{\DLR}{\ensuremath{\mathcal{DLR}}}

\newcommand{\QTL}{\mathcal{QTL}}
\newcommand{\QTLi}{{\ensuremath{\QTL^1}}}

\newcommand{\A}{\ensuremath{\mathcal{A}}}
\newcommand{\K}{\ensuremath{\mathcal{K}}}
\newcommand{\I}{\ensuremath{\mathcal{I}}}
\newcommand{\T}{\ensuremath{\mathcal{T}}}

\newcommand{\QT}{Q_\T}
\newcommand{\QA}{Q_\A}

\newcommand{\LogSpace}{\textsc{LogSpace}}
\newcommand{\NLogSpace}{\textsc{NLogSpace}}
\newcommand{\PTime}{\textsc{PTime}}
\newcommand{\NP}{\textsc{NP}}
\newcommand{\PSpace}{\textsc{PSpace}}
\newcommand{\ExpTime}{\textsc{ExpTime}}
\newcommand{\NExpTime}{\textsc{NExpTime}}
\newcommand{\ExpSpace}{\textsc{ExpSpace}}

\newcommand{\ex}[1]{\textit{#1}}
\newcommand{\nm}[1]{\mbox{\small\textsf{#1}}}
\newcommand{\nmd}[1]{\mbox{\scriptsize\textsf{#1}}}
\newcommand{\nms}[1]{\mbox{\tiny\textsf{#1}}}

\newcommand{\nxt}{{\ensuremath\raisebox{0.25ex}{\text{\scriptsize$\bigcirc$}}}}
\newcommand{\U}{\ensuremath{\mathbin{\mathcal{U}}}}
\renewcommand{\S}{\ensuremath{\mathbin{\mathcal{S}}}}

\newcommand{\Rdiamond}{\Diamond_{\!\scriptscriptstyle F}}
\newcommand{\Rbox}{\Box_{\!\scriptscriptstyle F}}
\newcommand{\Rnext}{\nxt_{\!\scriptscriptstyle F}}
\newcommand{\Ldiamond}{\Diamond_{\!\scriptscriptstyle P}}
\newcommand{\Lbox}{\Box_{\!\scriptscriptstyle P}}
\newcommand{\Lnext}{\nxt_{\!\scriptscriptstyle P}}

\newcommand{\SVdiamond}{\mathop{\ooalign{$\Diamond$ \cr \kern0.5ex
    \raisebox{0.35ex}{\scalebox{0.7}{$*$}}} \kern-0.9ex}}
\newcommand{\SVbox}{\mathop{\ooalign{$\Box$ \cr \kern0.42ex
    \raisebox{0.35ex}{\scalebox{0.7}{$*$}}}\rule{0pt}{1.5ex} \kern-0.7ex}}

\newcommand{\type}{\boldsymbol{t}}
\newcommand{\cp}{\textit{cp}}

\newcommand{\role}[1][\K]{\mathsf{role}_{#1}} 
\newcommand{\ob}[1][\A]{\mathsf{ob}_{#1}} 
\newcommand{\inv}[1]{\mathsf{inv}(#1)}

\newcommand{\Qs}{\smash{\widehat{Q}}}
\newcommand{\extA}{\mathcal{E}} 
\newcommand{\Nr}[4]{|#1^{#2,#3}_{#4}|}
\newcommand{\Nb}[3]{|#1^{\Box #2}_{#3}|}
\newcommand{\Nd}[3]{|#1^{\Diamond #2}_{#3}|}
\newcommand{\Ir}[4]{#1^{#2,#3}_{#4}}
\newcommand{\Irb}[3]{#1^{\Box #2}_{#3}}
\newcommand{\Ird}[3]{#1^{\Diamond #2}_{#3}}
\newcommand{\Rr}[3]{\varrho^{\smash{#1,#2}}_{#3}}
\newcommand{\Rb}[2]{\varrho^{\Box\hspace*{-0.1em} #1}_{#2}}
\newcommand{\Rd}[2]{\varrho^{\hspace*{-0.1em}\Diamond\hspace*{-0.1em} #1}_{#2}}
\newcommand{\Tr}[4]{\tau^{#1,#2}_{#3,#4}}
\newcommand{\Tb}[3]{\tau^{\Box\hspace*{-0.1em} #1}_{#2,#3}}
\newcommand{\Td}[3]{\tau^{\Diamond\hspace*{-0.1em} #1}_{#2,#3}}
\newcommand{\Dr}[4]{\delta^{#1,#2}_{#3,#4}}
\newcommand{\Db}[3]{\delta^{\Box\hspace*{-0.1em} #1}_{#2,#3}}
\newcommand{\Dd}[3]{\delta^{\Diamond\hspace*{-0.1em} #1}_{#2,#3}}

\newcommand{\LambdaB}[2]{\Lambda_{\Box\hspace*{-0.1em}#1}^{#2}}
\newcommand{\LambdaD}[2]{\Lambda_{\Diamond\hspace*{-0.1em}#1}^{#2}}

\newcommand{\Z}{\mathbb{Z}}
\newcommand{\N}{\mathbb{N}}
\newcommand{\TT}{\mathfrak{T}}
\newcommand{\Mmf}{\mathfrak{M}}

\newenvironment{akrzlist}{\begin{list}{}{\itemsep=0pt\leftmargin=2.5em\labelwidth=1.5em}}{\end{list}}
\newenvironment{akrzitemize}{\begin{list}{--}{\itemsep=0pt\leftmargin=1.5em\labelwidth=1.5em\topsep=4pt}}{\end{list}}

\begin{document}

\markboth{A.~Artale, R.~Kontchakov, V.~Ryzhikov, M.~Zakharyaschev}{A
  Cookbook for Temporal Conceptual Data Modelling with Description
  Logics}

\title{A Cookbook for Temporal Conceptual Data Modelling with Description Logics}

\author{ALESSANDRO ARTALE and VLADISLAV RYZHIKOV
\affil{KRDB Research Centre, Free University of Bozen-Bolzano, Italy}
ROMAN KONTCHAKOV and MICHAEL ZAKHARYASCHEV
\affil{Department of Computer Science and Information Systems,
Birkbeck, University of London, U.K.}
}


\begin{abstract}
  We design temporal description logics suitable for reasoning about
  temporal conceptual data models and investigate their computational
  complexity. Our formalisms are based on \DL{} logics with three
  types of concept inclusions (ranging from atomic concept inclusions and 
  disjointness to the full Booleans), as well as 
  cardinality constraints and role inclusions. The logics are interpreted over the
  Cartesian products of object domains and the flow of time $(\Z,<)$, satisfying the constant domain assumption. Concept and role inclusions of the TBox 
  hold at all moments of time (globally) and data assertions of the ABox hold at specified moments of time. 
  To express temporal constraints of conceptual data models, the languages are equipped 
  with flexible and rigid roles, standard future and past temporal operators on
  concepts and operators `always' and `sometime' 
  on roles.  The
  most expressive of our temporal description logics (which can
  capture lifespan cardinalities and either qualitative or
  quantitative evolution constraints) turns out to be
  undecidable. However, by omitting some of the temporal operators on
  concepts/roles or by restricting the form of concept inclusions we
  construct logics whose complexity ranges between 
  \NLogSpace{} and \PSpace{}. These positive results are obtained by reduction to
  various clausal fragments of propositional temporal logic, which
  opens a way to employ propositional or first-order temporal provers
  for reasoning about temporal data models.
\end{abstract}




\acmformat{Artale, A., Kontchakov, R., Ryzhikov, V., Zakharyaschev,
  M. 2014. A Cookbook for Temporal Conceptual Data Modelling with Description Logics.}

\begin{bottomstuff}
This work was partially supported by the U.K.\ EPSRC grant EP/H05099X/1.
\end{bottomstuff}

\maketitle


\section{Introduction}
\label{sec:intro}

The aim of this article is twofold. On the one hand, we investigate the complexity of reasoning about temporal conceptual data models depending on the available modelling constructs. On the other hand, we achieve this by encoding temporal conceptual data models in carefully crafted temporal description logics (TDLs, for short). As a result, we obtain a new family of TDLs and a clear understanding of how their constructs affect the complexity of reasoning. Most of the constructed TDLs feature an unexpectedly low complexity---compared to other known TDLs---such as \NLogSpace, \PTime, \NP{} and \PSpace, which is good news for automated temporal conceptual modelling. However, some combinations of the constructs (which involve temporal operators on relationships) result in undecidability, giving a new type of undecidable fragments of first-order temporal logic.

Conceptual data modelling formalisms, such as the Extended
Entity-Relationship model (EER) and Unified Modelling Language
(UML), provide visual means to describe
application domains in a declarative and reusable way, and are  regarded as standard tools in database design and
software engineering. One of the main tasks in conceptual modelling is to ensure that conceptual schemas satisfy various `quality properties': for instance, one may wish to check whether a given schema is consistent, whether its entities and relationships can be populated, whether a certain individual is an instance of a certain class, etc. That was where conceptual modelling met description logics (DLs), a family of knowledge representation formalisms specifically designed to efficiently reason about structured knowledge~\cite{BCMNP03}. Since 2007, DLs have been recognised as the backbone of the Semantic Web, underlying the standard Web Ontology Languages OWL and OWL~2\hbox to 0pt{.}\footnote{\url{www.w3.org/2007/OWL}, \ \url{www.w3.org/TR/owl2-overview}}

Connections between conceptual data models (CMs, for short) and DLs have
been investigated since the 1990s (see,
e.g.,~\cite{calvanese-et-al:jair-99,BoBr03,BeCD05-AIJ-2005,ACKRZ:er07}
and references therein), which resulted in a classification
of CMs according to the computational complexity of checking schema consistency depending on the available modelling constructs. The standard EER/UML
constructs include generalisation (inheritance) for entities  (classes), 
relationships and attributes  with disjointness and covering constraints on them, cardinality constraints for relationships and their
refinements, multiplicity constraints for attributes and key
constraints for entities.  
Reasoning over CMs equipped with the full set of constructs is \ExpTime-complete, which was shown by mapping CMs into the DLs \DLR{} and
\ALCQI~\cite{calvanese-et-al:jair-99,BeCD05-AIJ-2005}.
With the invention of the \DL{}
family~\cite{CDLLR05,CDLLR07,ACKZ:aaai07,ACKZ:jair09}, it became clear that reasoning
over CMs can often be done using DLs much weaker than \DLR{}
and \ALCQI. For example, the
\NP-complete \smash{\rDLbHN{}} was shown to be adequate for
representing a large class of CMs with generalisation and both disjointness and covering constraints, but no upper cardinality bounds on specialised relationships; see ~\cite{ACKRZ:er07} and Section~\ref{sec:dl-lite} for details. If we are also prepared
to sacrifice covering constraints, then the \NLogSpace-complete fragment
\smash{\rDLcHN{}} can do the job. (Note that \smash{\rDLcHN{}}
contains the OWL~2~QL profile\footnote{\url{www.w3.org/TR/owl2-profiles}} of OWL~2 and the DL fragment of RDF Schema, RDFS.\!\footnote{\url{www.w3.org/TR/rdf-schema}})

Temporal conceptual data models (TCMs)
extend CMs with means to represent constraints over temporal database instances. 
Temporal constraints can be grouped into three categories:
\emph{timestamping}, \emph{evolution} and \emph{temporal cardinality}
constraints. Timestamping constraints discriminate between those
classes, relationships and attributes that change over time and those
that are
time-invariant (or, \emph{rigid})~\cite{theodoulidis:et:al:is-91,gregersen:jensen:tkde-99,%
  finger:mcbrien:2000,Ar:Fr:er-99,mads-book:06}.
Evolution constraints control how the domain elements evolve
over time by migrating from one class to
another~\cite{gupta:hall:icde-91,%
  Mendelzon:94,Su97,mads-book:06,APS:amai07}. We distinguish between
qualitative evolution constraints describing generic temporal
behaviour, and quantitative ones specifying the exact time of
migration.
Temporal cardinality constraints restrict the number of times
an instance of a class can participate in a relationship: snapshot
cardinality constraints do it at each moment of time, while
lifespan cardinality constraints impose restrictions over the
entire existence of the instance as a member of the
class~\cite{tauzovich:er-91,mcbrien:et:al:cismod-92,artale:franconi:john09}.

Temporal extensions of DLs have been constructed and investigated since Schmiedel~\citeyear{schmiedel:90} and Schild's~\citeyear{Schild93} seminal papers (see~\cite{GKWZ03,ArFr01,AF05,LuWoZa-TIME-08} for detailed surveys), with reasoning over TCMs being one of the main objectives. 
The first attempts to
represent TCMs by means of TDLs resulted in fragments of
$\DLR_\mathcal{US}$ and $\ALCQI_\mathcal{US}$ whose
complexity ranged from \ExpTime{} and \ExpSpace{} up to
undecidability~\cite{Ar:Fr:er-99,AFWZ:02,AFM:lead03}. 
A general conclusion one could draw from the obtained results is that---as far as there is a nontrivial interaction between the temporal and DL components---TDLs based on full-fledged DLs such as \ALC{} turn out to be too complex for effective practical reasoning  (in more detail, this will be discussed in Section~\ref{sec:summary}).

The possibility to capture CMs using logics of the 
\DL{} family gave a glimpse of hope that automated reasoning over TCMs can
finally be made practical. 
The first temporal extension of $\smash{\rDLbHN{}}$ was constructed
by~\citeN{AKLWZ:time07}. It featured rigid roles,  with temporal and Boolean operators applicable not only to concepts but also to TBox axioms and ABox assertions. The resulting logic was shown to be \ExpSpace-complete.  (To compare: the same temporalisation of \ALC{} is trivially undecidable~\cite{AFWZ:02,GKWZ03}.)  
This encouraging result prompted a systematic investigation of TDLs
suitable for reasoning about TCMs.

Our aim in this article is to design \DL-based TDLs that are capable of
representing various sets of TCM constructs and have as low computational
complexity as possible.
Let us first formulate our minimal requirements for such
TDLs. At the model-theoretic level, we are interested in temporal
interpretations that are Cartesian products of object domains and the
flow of time $(\Z,<)$. At each moment of time, we interpret the
DL constructs over the same domain (thus complying with the constant
domain assumption adopted in temporal
databases~\cite{ChomickiTB:01}). We want to be able to specify, using
temporal ABoxes, that a finite number of concept and role membership
assertions hold at specific moments of time. We regard timestamping
constraints as indispensable; this means, in particular, that we
should be able to declare that certain roles and concepts are rigid (time-invariant) in
the sense that their interpretations do not change over time.  Other
temporal and static (atemporal) modelling constraints are expressed by means of
TBox axioms (concept and role inclusions). In fact, we observe that to
represent TCM constraints, we only require concept and role inclusions
that hold globally, at every time instant; thus, temporal and Boolean
operators on TBox axioms~\cite{AKLWZ:time07,BaaGhiLu-KR08,BaaderGL12} are not needed for our aims
(but may be useful to impose constraints on schema evolution).
Finally, in order to represent cardinality constraints (both snapshot and lifespan), we require number restrictions; thus, we assume this construct to be available in all of our formalisms. 

The remaining options include the choice of (\emph{i}) the underlying dialect of \DL{}
for disjointness and covering constraints; (\emph{ii}) the temporal operators on
concepts for different types of evolution constraints,
and (\emph{iii}) the temporal operators on roles for lifespan
cardinality constraints. For (\emph{i}), we consider three DLs: \smash{\rDLbHN}
and its sub-Boolean fragments \smash{\rDLkHN} and \smash{\rDLcHN}. For
(\emph{ii}), we take various subsets of the standard future and past
temporal operators (since and until, next and previous
time, sometime and always in the future/past, or simply sometime and
always). Finally, for (\emph{iii}), we only use the undirected
temporal operators `always' and `sometime' (referring to all time
instants); roles in the scope of such operators are called
temporalised.

Our most expressive TDL, based on \rDLbHN, captures all the standard types of temporal constraints:  timestamping, evolution and temporal cardinality.
Unfortunately, and to our surprise, this TDL turns out to be undecidable. As follows from the proof of Theorem~\ref{thm:undec}, it is a subtle interaction of functionality constraints on temporalised roles with the temporal operators and full Booleans on concepts that causes undecidability. 
On a more positive note, we show that even small restrictions of this interaction result in TDLs with better computational properties.

First, keeping \rDLbHN as the base DL but limiting the temporal operators on concepts to `always' and `sometime\hbox to 0pt{,}' we obtain an \NP-complete logic, which can express timestamping and life\-span cardinalities. To appreciate this result, recall that a similar logic based on \ALC\ is 2\ExpTime-complete~\cite{ALT:ijcai07}. 
Second, by giving up temporalised roles but retaining temporal operators on concepts, we obtain \PSpace- or \NP-complete logics depending on the available temporal operators, which matches the complexity of the underlying propositional temporal logic.
These TDLs have sufficient expressivity to capture timestamping
and evolution constraints, but cannot represent temporal cardinality
constraints (see Section~\ref{sec:tex}). We prove these
upper complexity bounds by a reduction to the propositional
temporal logic \PTL, which opens a way to employ the existing temporal
provers for checking quality properties of TCMs.
Again, we note that a similar logic based on \ALC{} is undecidable~\cite{WoZa99b,AFWZ:02,GKWZ03}.

We can reduce the complexity even further by restricting \rDLbHN{} to
its sub-Boolean fragments \smash{\rDLkHN{}} and
\smash{\rDLcHN}, which are unable to capture covering constraints. 
This results in logics within \NP{} and \PTime. And
if the temporal operators on concepts are limited to `always' and
`sometime' then the two sub-Boolean fragments are \NLogSpace-complete.
To obtain these results we consider sub-Boolean
fragments of \PTL{} by imposing restrictions on both the type of
clauses in Separated Normal Form~\cite{DBLP:conf/ijcai/Fisher91} and
the available temporal operators. We give a complete classification of
such fragments according to their complexity (see
Table~\ref{tab:PTL-fragments}).

The rest of the article is organised as follows. Section~\ref{sec:dl+modelling} introduces, using a simple example, conceptual data modelling languages and illustrates how they can be captured by various dialects of \DL, which are formally defined in Section~\ref{sec:dl-lite}. Section~\ref{sec:tex} introduces temporal conceptual modelling  constraints using a temporal extension of our example. In Section~\ref{sec:tdl}, we design \DL{} based TDLs that can represent those constraints. Section~\ref{sec:summary} gives a detailed overview of the results obtained in this article together with a discussion of related work. Section~\ref{sec:tdl-rigid:Z:bool-until-complex} gives the reduction of TDLs to \PTL{} mentioned above. In Section~\ref{sec:2ltl}, we establish the complexity results for the clausal fragments of propositional temporal logic. Section~\ref{sec:tdl-temporalised-roles} studies the complexity of TDLs with temporalised roles. 
We discuss the obtained results, open problems and future directions in Section~\ref{sec:concl}.

\section{Conceptual Modelling and Description Logic}
\label{sec:dl+modelling}

Description logics (DLs; see, e.g.,~\cite{BCMNP03}) were designed in the 1980s as logic-based formalisms for knowledge representation and
reasoning; their major application areas include ontologies in life
sciences and the Semantic Web. Conceptual modelling
languages~\cite{Chen76} are a decade older, and were developed
for abstract data representation in database design. Despite apparent
notational differences, both families of languages are built around
concepts (or entities) and relationships using a number of `natural'
constructs; a close correspondence between them was discovered and
investigated in~\cite{calvanese-et-al:jair-99,BoBr03,%
  BeCD05-AIJ-2005,ACKRZ:er07}.

The \DL{} description
logics~\cite{CDLLR05,CDLLR07,PLCD*08,ACKZ:aaai07,ACKZ:jair09} and the
\DL-based profile OWL~2~QL of OWL~2 have grown from the idea of
linking relational databases and ontologies in the framework of
ontology-based data access~\cite{DFKM*08,HMAM*08,PLCD*08}. The
chief aims that determined the shape of the \DL{} logics are:
(\emph{i}) the ability to represent basic constraints used in
conceptual modelling, and (\emph{ii}) the ability to support
query answering using standard relational database systems. In this article, we concentrate on \DL{} as a modelling language and briefly return to the issue of ontology-based data access (OBDA) in Section~\ref{sec:concl}.

In this section, we give an intuitive example illustrating the main
constructs of conceptual data models and their \DL{}
representations. In the example, we use the Extended
Entity-Relationship (EER) language~\cite{ElNa07}; however, one can
easily employ other conceptual modelling formalisms such as UML class
diagrams (\url{www.uml.org}). Then we formally define the syntax and
semantics of the \DL{} logics to be used later on in this article.

\subsection{A Motivating Example}
\label{subsec:mot-ex}

Let us consider the EER diagram in Fig.~\ref{atemp-uml} representing  (part of) a company information system.%
\begin{figure}[t]
\centering
\begin{tikzpicture}[>=latex,class/.style={rectangle,draw=black,thick,inner xsep=6pt},
conn/.style={draw,thick}, 
double-conn/.style={->,>=stealth,draw,thick,double distance=1pt, 
          decoration={markings,mark=at position 1 with {\arrow[scale=0.9]{>}}}, postaction={decorate}},
single-conn/.style={->,>=stealth,draw,thick,
          decoration={markings,mark=at position 1 with {\arrow[scale=1.5]{>}}}, postaction={decorate}},
temp/.style={draw,thick,densely dashed,>=open triangle 60,
	   decoration={markings,mark=at position 1 with {\arrow[scale=0.8]{>}}}, postaction={decorate}, shorten >= 5pt},
attribute/.style={circle,draw,thick,minimum size=1.5mm,inner sep=0pt},
relation/.style={draw,thick,diamond,aspect=4,inner ysep=0pt,inner xsep=6pt},yscale=0.9]\footnotesize
\node[class] (department) at (0,0) {\nmd{Department}};
\node[class] (interest group) at (2.3,0) {\nmd{Interest Group}};
\node[circle,draw,thick,minimum size=3mm,inner sep=0pt] (disj) at (1.15,0.6) {\tiny d};
\draw[conn] (department) -- (disj);
\draw[conn] (interest group) -- (disj);
\node[class] (organisational unit) at (1.15,1.5) {\nmd{Organisational Unit}};
\draw[double-conn] (disj) -- (organisational unit);
\node[relation] (member) at (1.15,2.9) {\nmd{Member}};
\draw[conn] (organisational unit) -- (member) node [near start, left] {\scriptsize $_{(1,\infty)}$} node[near end, right] {\scriptsize $\mathsf{org}$};
\node[class] (area manager) at (4.6,0) {\nmd{Area Manager}};
\node[class] (top manager) at (6.8,0) {\nmd{Top Manager}};
\node[circle,draw,thick,minimum size=3mm,inner sep=0pt] (cover) at (5.7,0.6) {};
\draw[conn] (area manager) -- (cover);
\draw[conn] (top manager) -- (cover);
\node[class] (manager) at (5.7,1.5) {\nmd{Manager}};
\draw[double-conn] (cover) -- (manager);
\node[class] (employee) at (5.7,2.9) {\nmd{Employee}};
\draw[single-conn] (manager) -- (employee);
\draw[conn] (employee) -- (member) node[very near end, below] {\scriptsize $\mathsf{mbr}$}; 
\node[attribute,label=above:{\nms{Salary(Integer)}}] (salary) at (7.7,3.3) {};
\draw[conn] (employee.north east) -- (salary);
\node[attribute,label=above:{\nms{Name(String)}}] (name) at (3.8,3.3) {};
\draw[conn] (employee.north west) -- (name);
\node[attribute,label=above:{\nms{\underline{Payroll Number(Integer)}}}] (payslip) at (5.7,3.5) {};
\draw[conn] (employee) -- (payslip);
\node[relation] (manages) at (10.5,0) {\nmd{Manages}};
\draw[conn] (top manager) -- (manages) node[very near end, above] {\scriptsize $\mathsf{man}$}  node [near start, below] {\scriptsize $_{(1,1)}$}; 
\node[relation] (works on) at (10.5,2.9) {\nmd{Works On}};
\draw[conn] (employee) -- (works on) node[very near end, below] {\scriptsize $\mathsf{emp}$}; 
\node[class] (project) at (10.5,1.5) {Project};
\node[attribute,label=above:{\nms{\underline{Project Code(String)}}}] (p code) at (11.8,1.8) {};
\draw[conn] (project.east) -- (p code);
\draw[conn] (project) -- (manages) node[near end, left] {\scriptsize $\mathsf{prj}$} node [near start, right] {\scriptsize $_{(1,1)}$}; 
\draw[conn] (project) -- (works on) node[near end, left] {\scriptsize $\mathsf{act}$} node [near start, left] {\scriptsize $_{(3,\infty)}$}; 
\end{tikzpicture}
\caption{A conceptual data model of a company information system.}
\label{atemp-uml}
\end{figure}
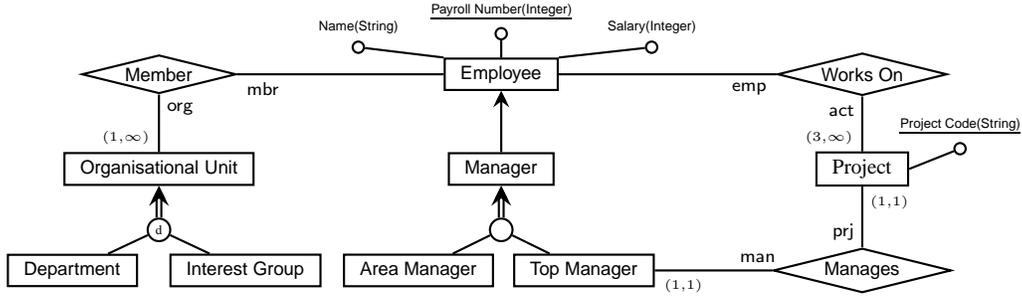
%
The arrow from the entity \nm{Manager} to the entity \nm{Employee} stands
for the statement `all managers are employees.' The double arrow with
a circle below \nm{Manager} means that the set of managers is the union
of the set of area managers and the set of top managers. These
statements can be represented in the language of description logic as
inclusions between concepts:
\begin{align*}
  \ex{Manager} &~\sqsubseteq~ \ex{Employee}, &
  \ex{AreaManager} &~\sqsubseteq~ \ex{Manager},\\
\ex{Manager} &~\sqsubseteq~ \ex{AreaManager} \sqcup \ex{TopManager}, &
  \ex{TopManager} & ~\sqsubseteq~ \ex{Manager}.
\end{align*}
Here $\ex{Manager}$, $\ex{Employee}$, \ex{AreaManager}, 
\ex{TopManager} are \emph{concept names} (or unary predicates) and the
symbols $\sqsubseteq$ and $\sqcup$ denote the usual
set-theoretic inclusion and union, respectively.
In a similar way we read and represent the part of the EER diagram
located below \nm{Organisational Unit}; the only new ingredient here is
the circled $\mathsf{d}$, indicating that the union is
\emph{disjoint}:
\begin{align*}
  \ex{Department} &~\sqsubseteq~ \ex{OrganisationalUnit}, &\quad
  & \ex{OrganisationalUnit} ~\sqsubseteq~ \ex{Department} \sqcup \ex{InterestGroup}, \\
  \ex{InterestGroup} &~\sqsubseteq~ \ex{OrganisationalUnit}, &
  & \ex{Department} \sqcap \ex{InterestGroup} ~\sqsubseteq~ \bot.
\end{align*}
Here $\bot$ denotes the empty set and $\sqcap$ the set-theoretic intersection.

The entity \nm{Employee} in Fig.~\ref{atemp-uml} has three
\emph{attributes}: \nm{Name}, which is a string, and \nm{Payroll Number} and
\nm{Salary}, both of which are integers. The attribute
\nm{Payroll Number} (underlined) is a \emph{key} for the entity
\nm{Employee}. In description logic, we can encode attributes by means of
\emph{roles} (binary predicates). For example, to say that every
employee has a salary, which is an integer number,
we can represent the attribute \nm{Salary} by a role, \ex{salary},
together with the concept inclusions
\begin{align*}
  \ex{Employee} &~\sqsubseteq~ \exists \ex{salary}, &
    \exists \ex{salary}^- &~\sqsubseteq~ \ex{Integer},
\end{align*}
where $\exists \ex{salary}$ denotes the domain of \ex{salary}, and
$\ex{salary}^-$ is the inverse of \ex{salary}, so that $\exists
\ex{salary}^-$ is the range of \ex{salary}. Then the fact that each
individual has a unique \ex{salary} attribute value can be expressed
by the concept inclusion
\begin{align*}
    \mathop{\geq 2} \ex{salary} &~\sqsubseteq~ \bot,
\end{align*}
where $\mathop{\geq 2} \ex{salary}$ stands for the set of all domain
elements with at least two values of \ex{salary} attached to them
(which must be empty according to this inclusion, i.e., \ex{salary} is a \emph{functional} role).  
The attributes \nm{Payroll Number} and
\nm{Name} are represented in a similar manner. The fact that
\nm{Payroll Number} is a key for \nm{Employee} can be encoded by the inclusion
\begin{align*}
  \mathop{\geq 2} \ex{payrollNumber}^- &~\sqsubseteq~ \bot.
\end{align*}

\emph{Relationships} are used to describe connections among objects
from (possibly) different entities.  \nm{Works On}, \nm{Member} and
\nm{Manages} in Fig.~\ref{atemp-uml} are binary
relationships. The argument \nm{emp} of
\nm{Works On} is of type \nm{Employee} in the sense that its values always belong to the entity \nm{Employee}
(in other words, \nm{Employee} participates in \nm{Works On} as
\nm{emp}). Likewise, the argument \nm{act} of \nm{Works On} 
is of type \nm{Project}. In description logic, a binary
relationship such as \nm{Works On} can be represented by a role,
say, \ex{worksOn}. If we agree that the first argument of \ex{worksOn} corresponds to 
\nm{emp} and the second to \nm{act}, then the
domain of \ex{worksOn} belongs to \ex{Employee} and its
range to \ex{Project}:
\begin{align*}
  \exists \ex{worksOn} &~\sqsubseteq~ \ex{Employee}, & 
  \exists \ex{worksOn}^- &~\sqsubseteq~ \ex{Project}.
\end{align*}
The expression $(3,\infty)$ labelling the argument \nm{act} of
\nm{Works On} is a \emph{cardinality constraint} meaning that every element of the set \nm{Project} participates
in at least three distinct tuples in the relationship \nm{Works On} (each
project involves at least three employees). This can be represented by the
inclusion
\begin{align}\label{eq:card}
  \ex{Project} &~\sqsubseteq~ \mathop{\geq 3} \ex{worksOn}^-.
\end{align}
The expression $(1,1)$ labelling the argument \nm{prj} of the
relationship \nm{Manages} means that each element of \nm{Project}
participates in at least one and at most one (that is, exactly one) tuple in \nm{Manages}, which
is represented by two inclusions:
\begin{align*}
  \ex{Project} &~\sqsubseteq~ \exists \ex{manages}^-, &
  \ex{Project} &~\sqsubseteq~ \mathop{\leq 1} \ex{manages}^-.
\end{align*}

Relationships of arity greater than 2 are encoded by using
\emph{reification}~\cite{CDLN-handbook-AR-2001} (binary
relationships can also be reified).  For instance, to reify the binary
relationship \nm{Works On}, we introduce a new concept name, say
\ex{C-WorksOn}, and two functional roles, \ex{emp} and \ex{act},
satisfying the following concept inclusions:
\begin{align}\label{eq:reifi:1}
\ex{C-WorksOn} & \sqsubseteq \exists \ex{emp}, & \mathop{\geq 2} \ex{emp} & \sqsubseteq \bot, & \exists \ex{emp} & \sqsubseteq \ex{C-WorksOn}, & \exists \ex{emp}^- & \sqsubseteq \ex{Employee},\\
\label{eq:reifi:2}
\ex{C-WorksOn} & \sqsubseteq \exists \ex{act}, & \mathop{\geq 2} \ex{act} & \sqsubseteq \bot, & \exists \ex{act} & \sqsubseteq \ex{C-WorksOn}, & \exists \ex{act}^- & \sqsubseteq \ex{Project}.%
\end{align}
Thus, each element of \ex{C-WorksOn} is related, via the roles
\ex{emp} and \ex{act}, to a unique pair of elements of \ex{Employee}
and \ex{Project}. Cardinality constraints are still representable for
reified relations, e.g., the cardinality expressed by the
formula~\eqref{eq:card} becomes
\begin{align}\label{eq:reifi:card}
  \ex{Project} &~\sqsubseteq~ \mathop{\geq 3} \ex{act}^-.  
\end{align}
Of the data modelling constructs not used in
Fig.~\ref{atemp-uml}, we mention here \emph{relationship
  generalisation}, i.e., a possibility to state that one
relationship is a subset of another relationship. For example,
we can state that everyone managing a project must also work on the
project. In other words: \nm{Manages} is a sub-relationship of \nm{Works On},
which can be represented in description logic as the role inclusion
\begin{align*}
  \ex{manages} &~\sqsubseteq~ \ex{worksOn}
\end{align*}
if both relationships are binary and not reified.  On the other hand,
if both relationships are reified then we need
a concept inclusion between the respective reifying concepts
as well as role inclusions between the functional roles for their arguments:
\begin{align*}
\ex{C-Manages} ~\sqsubseteq~\ex{C-WorksOn},\qquad
\ex{prj}~\sqsubseteq \ex{act},\qquad
\ex{man}~\sqsubseteq\ex{emp}.
\end{align*}

To represent database instances of a conceptual model, we use assertions such as $\ex{Manager}(\ex{bob})$ for `Bob is a manager' and $\ex{manages}(\ex{bob},\ex{cronos})$ for `Bob manages
Cronos.'

As conceptual data models can be large and contain non-trivial
implicit knowledge, it is important to make sure that the constructed
conceptual model satisfies certain \emph{quality properties}. For
example, one may want to know whether it is consistent, whether all or
some of its entities and relationships are not necessarily empty or
whether one entity or relationship is (not) subsumed by another. 
To automatically check such quality properties, it is essential to
provide an effective reasoning support during the construction phase
of a conceptual model. 

We now define the reasoning problems formally, by giving the
syntax and semantics of description logics containing the constructs
discussed above.

\subsection{\DL{} Logics}
\label{sec:dl-lite}

We start with the logic called \DLbn{} in the nomenclature 
of~\citeN{ACKZ:jair09}.
The language of \DLbn{} contains \emph{object names} $a_0,a_1,\dots$,
\emph{concept names} $A_0,A_1,\dots$, and \emph{role names}
$P_0,P_1,\dots$. \emph{Roles} $R$, \emph{basic concepts} $B$
and \emph{concepts} $C$ of this language are defined by the grammar:
\begin{align*}\label{eq:def:concept}
R \quad &::=\quad P_k \quad\mid\quad P_k^-, \\ %
B \quad &::=\quad \bot 
\quad\mid\quad A_k \quad\mid\quad
\mathop{\geq q} R, \\%
C \quad &::=\quad B \quad\mid\quad  \neg C \quad\mid\quad C_1
\sqcap C_2,
\end{align*}
where $q$ is a positive integer represented in binary.
A \DLbn{} \emph{TBox}, $\T$, is a finite set of \emph{concept inclusion axioms} of the form
\begin{equation*}
  C_1 ~\sqsubseteq~ C_2.
\end{equation*}
An \emph{ABox}, $\A$, is a finite set of assertions of
the form
\begin{equation*}
  A_k(a_i), \qquad \neg A_k(a_i), \qquad P_k(a_i,a_j), \qquad \neg P_k(a_i,a_j).
\end{equation*}
Taken together, $\T$ and $\A$ constitute the 
\emph{knowledge base} (KB, for short) $\K=(\T,\A)$.

An \emph{interpretation} $\I=(\Delta^\I,
\cdot^\I)$ of this and other \DL\ languages consists of a
\emph{domain} $\Delta^\I\ne\emptyset$ and an interpretation
function $\cdot^\I$ that assigns to each object name $a_i$ an
element $a_i^{\I}\in \Delta^\I$, to each concept
name $A_k$ a subset $A_k^{\I}\subseteq \Delta^\I$,
and to each role name $P_k$ a binary relation
$P_k^{\I}\subseteq\Delta^\I\times\Delta^\I$. As in databases, we adopt the \emph{unique name
  assumption} (UNA): $a_i^{\I}\neq
a_j^{\I}$ for all $i\neq j$ (note, however, that OWL does not use the UNA).
The role and concept constructs are interpreted in $\I$ as follows:
\begin{align*}
  (P_k^-)^{\I} & ~=~ \{ (y,x) \in \Delta^\I \times
  \Delta^\I \mid (x,y)\in P_k^{\I} \}, &
 \bot^{\I} & ~=~\emptyset,\\
  (\mathop{\geq\! q} R)^{\I} & ~=~
  \big\{x\in\Delta^\I \mid\sharp\{y\in\Delta^\I
   \mid (x,y)\in R^{\I}\}\geq q\big\}, &
  (\neg C)^{\I} & ~=~ \Delta^\I\setminus C^{\I},\\
&&   (C_1\sqcap C_2)^{\I} & ~=~ C_1^{\I}\cap C_2^{\I},
\end{align*}
where $\sharp X$ denotes the cardinality of $X$. We use the standard
abbreviations: 
\begin{equation*}
  C_1 \sqcup C_2 = \neg(\neg C_1 \sqcap \neg C_2),\qquad
  \top = \neg\bot,\qquad
  \exists R = (\mathop{\geq 1} R),\qquad
  \mathop{\leq q} R = \neg(\mathop{\geq q + 1} R).
\end{equation*}
Concepts of the form $\mathop{\leq q} R$ and $\mathop{\geq q} R$ are
called \emph{number restrictions}, and those of the form $\exists R$
are called \emph{existential concepts}.

The \emph{satisfaction relation} $\models$ is defined as expected:
\begin{align*}
  \I\models C_1\sqsubseteq C_2 & \quad\text{iff}\quad
  C_1^{\I}\subseteq C_2^{\I}, \\
  \I\models A_k(a_i) & \quad\text{iff}\quad
  a_i^{\I}\in A_k^{\I}, & %
  \I\models P_k(a_i,a_j) & \quad\text{iff}\quad
  (a_i^{\I},a_j^{\I})\in P_k^{\I},  \\
  \I\models \neg A_k(a_i) & \quad\text{iff}\quad
  a_i^{\I}\notin A_k^{\I},  &
  \I\models \neg P_k(a_i,a_j) & \quad\text{iff}\quad
  (a_i^{\I},a_j^{\I})\notin P_k^{\I}.
\end{align*}
A knowledge base $\K=(\T,\A)$ is said to be \emph{satisfiable} (or
\emph{consistent}) if there is an interpretation $\I$ satisfying all
the members of $\T$ and $\A$. In this case we write $\I\models\K$ (as
well as $\I\models\T$ and $\I\models\A$) and say that $\I$ is a
\emph{model of} $\K$ (and of $\T$ and $\A$). The \emph{satisfiability
  problem}---given a KB $\K$, decide whether $\K$ is satisfiable---is
the main reasoning problem we consider in this article. 
\emph{Subsumption} (given an inclusion $C_1 \sqsubseteq C_2$ and a TBox $\T$, decide whether $\I \models C_1^\I \subseteq C_2^\I$ for every model $\I$ of $\T$; or $\T \models C_1 \sqsubseteq C_2$  in symbols) and \emph{concept
satisfiability} (given a concept $C$ and a TBox $\T$, decide whether there is a model $\I$ of $\T$ such that $C^\I\neq \emptyset$; or $\T \not\models C \sqsubseteq \bot$) are reducible to satisfiability. For
example, to check whether $\T \models C_1 \sqsubseteq C_2$
we can construct a new KB $\K = ( \T \cup \{A \sqsubseteq C_1, A
\sqsubseteq \neg C_2 \}, \{A(a)\})$ with a fresh concept name $A$, and
check whether $\K$ is \emph{not} satisfiable.

The two sub-languages of \DLbn{} we deal with in this article are obtained by restricting the Boolean operators on concepts. In
\DLkn{} TBoxes,\!\footnote{The Krom fragment of first-order logic consists of  formulas in prenex normal form whose quantifier-free part is a conjunction of binary clauses.}
concept inclusions are of the form
\begin{equation*}\tag{\textit{krom}}
  B_1 ~\sqsubseteq~ B_2, \qquad B_1 ~\sqsubseteq~ \neg B_2
  \qquad\text{or}\qquad \neg B_1 ~\sqsubseteq~ B_2.
\end{equation*}
(Here and below $B_1$, $B_2$ are basic concepts.)
\DLcn{} only uses concept inclusions of the form
\begin{equation*}\tag{\textit{core}}
  B_1 ~\sqsubseteq~ B_2 \qquad\text{or}\qquad
  B_1 \sqcap B_2 ~\sqsubseteq~ \bot.
\end{equation*}
As $B_1 \sqsubseteq \neg B_2$ is equivalent to $B_1 \sqcap B_2
\sqsubseteq \bot$, \DLcn{} is a sub-language of \DLkn. Although the Krom fragment does not seem to be more useful for
conceptual modelling than \DLcn, we shall see in
Remark~\ref{krom} that temporal extensions of \DLkn{} can capture some important temporal modelling constructs that are not representable by the corresponding extensions of \DLcn.

Most of the constraints in the company conceptual model from Section~\ref{subsec:mot-ex} were represented by means of \DLcn{} concept inclusions. The only exceptions were the covering constraints $\ex{Manager} \sqsubseteq \ex{AreaManager} \sqcup \ex{TopManager}$ and $\ex{OrganisationalUnit} \sqsubseteq \ex{Department} \sqcup \ex{InterestGroup}$, which belong to the language \DLbn, and  the role inclusion $\ex{manages} \sqsubseteq \ex{worksOn}$.
The extra expressive power, gained from the addition of covering constraints to \DLcn, comes at a price~\cite{ACKZ:aaai07}: the satisfiability problem is \NLogSpace-complete for \DLcn{} and \DLkn{} KBs and \NP-complete for \DLbn{} KBs.

The straightforward extension of \DLcn{} with role inclusions of the form
\begin{align*}
R_1 ~\sqsubseteq~ R_2 \qquad (\text{with \quad $\I\models R_1\sqsubseteq R_2 \quad\text{iff}\quad
  R_1^{\I}\subseteq R_2^{\I}$})
\end{align*}
leads to an even higher complexity: satisfiability becomes \ExpTime-complete~\cite{ACKZ:jair09}. The reason for this is the interaction of functionality constraints and role inclusions such as
\begin{equation*}
R_1 \sqsubseteq R_2\quad\text{and}\quad \mathop{\geq 2} R_2 \sqsubseteq \bot.
\end{equation*}
Note that inclusions of this sort are required when we use relationship generalisation with reification (see Section~\ref{subsec:mot-ex}).
If we restrict this interaction in  TBoxes $\T$ by requiring that no role $R$ can occur in $\T$ in both a role inclusion of the form $R' \sqsubseteq R$ and a number restriction $\mathop{\geq q} R$ or $\mathop{\geq q} R^-$ with $q\ge 2$, then the complexity of satisfiability checking with such TBoxes matches that of the language  without role inclusions. The extension of $\DLan$, where $\alpha \in \{{\it core}, {\it krom}, {\it bool} \}$, with role inclusions satisfying the condition above is denoted by $\smash{\rDLaHN}$; without this condition, the extension is denoted by $\DL_\alpha^\mathcal{HN}$.
Table~\ref{table:DL-Lite:languages} summarises the complexity of the KB satisfiability problem for \DL{} logics (for details, consult~\cite{ACKZ:jair09}).

\begin{table}[t]
\tbl{Complexity of the \DL{} logics.\label{table:DL-Lite:languages}}{
\centering\renewcommand{\arraystretch}{1.3}
\begin{tabular}[c]{|c|c|c|c|}
  \hline
  \raisebox{-2pt}{concept}        & \multicolumn{3}{c|}{role inclusions}\\ \cline{2-4}
  \raisebox{4pt}{inclusions} & $\DLan$ &  \rule[-4pt]{0pt}{16pt}$\rDLaHN$  & $\DL_\alpha^\mathcal{HN}$\\\hline\hline
Bool &  \NP & \NP & \ExpTime \\\hline 
  Krom & \NLogSpace & \NLogSpace &  \ExpTime \\\hline 
  core & \NLogSpace & \NLogSpace &  \ExpTime \\\hline
\end{tabular}
}
\end{table}

Thus, already in the atemporal case, a conceptual data model engineer has to search for a suitable compromise between the expressive power of the modelling language and efficiency of reasoning. In the temporal case, the trade-off between expressiveness and efficiency becomes even more dramatic.

In the next section, we extend the atemporal conceptual data model considered above with a number of temporal constructs and use them  to design a family of temporal description logics that are suitable for temporal conceptual modelling.


\section{Temporal Conceptual Modelling and Temporal Description Logic}
\label{sec:tex}

Temporal conceptual data models extend standard conceptual schemas
with means to visually represent temporal constraints imposed on
temporal database
instances~\cite{theodoulidis:et:al:is-91,tauzovich:er-91,%
jensen:snodgrass:tkde-99,AFM:lead03,mads-book:06,combi:et:al:er-08}.

When introducing a temporal dimension into conceptual data models,
time is usually modelled by a linearly ordered set of time instants,
so that at each moment of time we can refer to its past and future. In
this article, we assume that the flow of time is isomorphic to the
strictly linearly ordered set $(\Z,<)$ of integer
numbers. (For a survey of other options, including 
interval-based and branching models of time, consult,
e.g.,~\cite{Gabbayetal94,GabbayFingerReynolds2000a,GKWZ03}.)

We will now introduce the most important temporal conceptual modelling
constructs by extending the company information system example from
Section~\ref{subsec:mot-ex}.


\subsection{The Motivating Example Temporalised}
\label{subsec:t-mot-ex}

A basic assumption in temporal conceptual models is that
entities, relationships and attributes may freely change over
time as long as they satisfy the constraints of the schema at
\emph{each} time instant.
Temporal constructs are used to impose constraints on the temporal
behaviour of various components of conceptual schemas. We group these
constructs into three categories---\emph{timestamping},
\emph{evolution} and \emph{temporal cardinality
  constraints}---and illustrate them by the model in
Fig.~\ref{uml}.
%
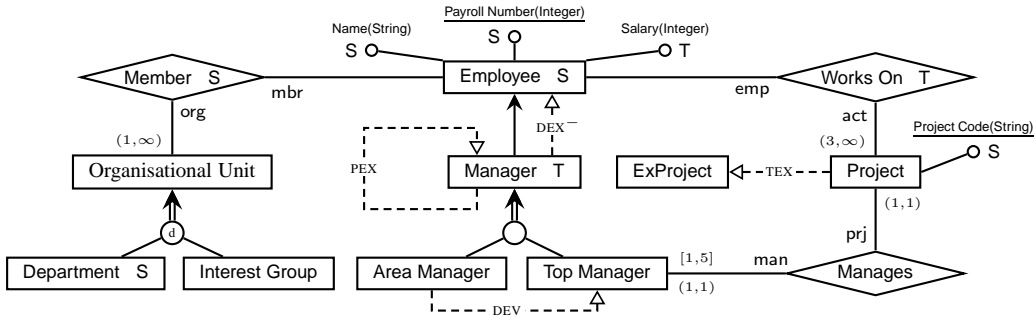
\begin{figure}[t]
\centering
\begin{tikzpicture}[>=latex,class/.style={rectangle,draw=black,thick,inner xsep=6pt},
conn/.style={draw,thick}, 
double-conn/.style={->,>=stealth,draw,thick,double distance=1pt, 
          decoration={markings,mark=at position 1 with {\arrow[scale=0.9]{>}}}, postaction={decorate}},
single-conn/.style={->,>=stealth,draw,thick,
          decoration={markings,mark=at position 1 with {\arrow[scale=1.5]{>}}}, postaction={decorate}},
temp/.style={draw,thick,densely dashed,>=open triangle 60,
	   decoration={markings,mark=at position 1 with {\arrow[scale=0.8]{>}}}, postaction={decorate}, shorten >= 5pt},
attribute/.style={circle,draw,thick,minimum size=1.5mm,inner sep=0pt},
relation/.style={draw,thick,diamond,aspect=4,inner ysep=0pt,inner xsep=6pt},yscale=0.9]\footnotesize
\node[class] (department) at (0,0) {\nmd{Department} \ \ {\scriptsize\textsf{S}}};
\node[class] (interest group) at (2.3,0) {\nmd{Interest Group}};
\node[circle,draw,thick,minimum size=3mm,inner sep=0pt] (disj) at (1.15,0.6) {\tiny d};
\draw[conn] (department) -- (disj);
\draw[conn] (interest group) -- (disj);
\node[class] (organisational unit) at (1.15,1.5) {Organisational Unit};
\draw[double-conn] (disj) -- (organisational unit);
\node[relation] (member) at (1.15,2.9) {\nmd{Member}  \ \ {\scriptsize\textsf{S}}};
\draw[conn] (organisational unit) -- (member) node [near start, left] {\scriptsize $_{(1,\infty)}$} node[near end, right] {\scriptsize $\mathsf{org}$};
\node[class] (area manager) at (4.6,0) {\nmd{Area Manager}};
\node[class] (top manager) at (6.8,0) {\nmd{Top Manager}};
\draw[temp] ($(area manager.south)+(0,0)$) -- ++(0,-0.3) -- ++(2,0) node[midway,fill=white,inner sep=1pt] {\tiny\textsc{DEV}} -|  ($(top manager.south)+(-0,0)$) ;
\node[circle,draw,thick,minimum size=3mm,inner sep=0pt] (cover) at (5.7,0.6) {};
\draw[conn] (area manager) -- (cover);
\draw[conn] (top manager) -- (cover);
\node[class] (manager) at (5.7,1.5) {\nmd{Manager}  \ \ {\scriptsize\textsf{T}}};
\draw[temp] ($(manager.south)+(-0.5,0)$) -- ++(0,-0.3) -- ++(-1.5,0) -- ++(0,1.2) node[midway,fill=white,inner sep=1pt] {\tiny\textsc{PEX}} -- ++(1.5,0) -- ($(manager.north)+(-0.5,0)$);
\draw[double-conn] (cover) -- (manager);
\node[class] (employee) at (5.7,2.9) {\nmd{Employee}  \ \ {\scriptsize\textsf{S}}};
\draw[temp] ($(manager.north)+(0.5,0)$) -- ($(employee.south)+(0.5,0)$) node[midway,fill=white,inner sep=1pt] {\tiny\hspace*{1em}\textsc{DEX$^-$}};
\draw[single-conn] (manager) -- (employee);
\draw[conn] (employee) -- (member) node[very near end, below] {\scriptsize $\mathsf{mbr}$}; 
\node[attribute,label=above:{\nms{Salary(Integer)}},label=right:{{\scriptsize\textsf{T}}}] (salary) at (7.7,3.3) {};
\draw[conn] (employee.north east) -- (salary);
\node[attribute,label=above:{\nms{Name(String)}},label=left:{{\scriptsize\textsf{S}}}] (name) at (3.8,3.3) {};
\draw[conn] (employee.north west) -- (name);
\node[attribute,label=above:{\nms{\underline{Payroll Number(Integer)}}},label=left:{{\scriptsize\textsf{S}}}] (payslip) at (5.7,3.5) {};
\draw[conn] (employee) -- (payslip);
\node[relation] (manages) at (10.5,0) {\nmd{Manages}};
\draw[conn] (top manager) -- (manages) node[very near end, above] {\scriptsize $\mathsf{man}$} node [near start, above] {\scriptsize $_{[1,5]}$} node [near start, below] {\scriptsize $_{(1,1)}$}; 
\node[relation] (works on) at (10.5,2.9) {\nmd{Works On}   \ \ {\scriptsize\textsf{T}}};
\draw[conn] (employee) -- (works on) node[very near end, below] {\scriptsize $\mathsf{emp}$}; 
\node[class] (project) at (10.5,1.5) {\nmd{Project}};
\node[class] (exproject) at (7.8,1.5) {\nmd{ExProject}};
\draw[temp] (project) -- (exproject) node[midway,fill=white,inner sep=1pt] {\tiny\textsc{TEX}};
\node[attribute,label=above:{\nms{\underline{Project Code(String)}}},label=right:{{\scriptsize\textsf{S}}}] (p code) at (11.8,1.8) {};
\draw[conn] (project.east) -- (p code);
\draw[conn] (project) -- (manages) node[near end, left] {\scriptsize $\mathsf{prj}$} node [near start, right] {\scriptsize $_{(1,1)}$}; 
\draw[conn] (project) -- (works on) node[near end, left] {\scriptsize $\mathsf{act}$} node [near start, left] {\scriptsize $_{(3,\infty)}$}; 
\end{tikzpicture}
\caption{A temporal conceptual model of a company information system.}
\label{uml}
\end{figure}

\emph{Timestamping constraints}~\cite{theodoulidis:et:al:is-91,%
  gregersen:jensen:tr98,gregersen:jensen:tkde-99,finger:mcbrien:2000,%
  Ar:Fr:er-99,mads-book:06} distinguish between entities,
relationships and attributes that are
\begin{akrzitemize}
\item \emph{temporary} in the sense that no element belongs to them at all moments of time,

\item \emph{snapshot}, or time-invariant, in the sense that their interpretation does not change with time,

\item \emph{unconstrained} (all others).
\end{akrzitemize}
In temporal entity-relationship diagrams, the 
temporary entities, relationships and attributes are marked with {\sf T} and the snapshot ones with {\sf S}.
In Fig.~\ref{uml}, \nm{Employee} and \nm{Department} are snapshot entities,
\nm{Name}, \nm{Payroll Number} and \nm{Project Code} are snapshot attributes and
\nm{Member} a snapshot relationship.  On the other hand, \nm{Manager} is a
temporary entity, \nm{Salary} a temporary attribute, and \nm{Works On} a
temporary relationship.

There are (at least) two ways of representing timestamping constraints in temporal description logics. One of them is to introduce special names for temporary and snapshot concepts and roles, and interpret them accordingly. Another way is to employ a \emph{temporal operator} $\SVbox$, which is read as `always' or `at all---past, present and future---time instants\hbox to 0pt{.}' Intuitively, for a concept $C$, $\SVbox C$ contains those elements that belong to $C$ at all time instants. Using this operator, the constraints `\nm{Employee} is a snapshot entity' and `\nm{Manager} is a temporary entity' can be represented as follows:
\begin{align*}
  \ex{Employee} &~\sqsubseteq~\SVbox \ex{Employee}, &
  \SVbox\ex{Manager} &~\sqsubseteq~\bot.
\end{align*}
The first inclusion says that, at any moment of time, every element of \nm{Employee} has always been and will always be an element of
\nm{Employee}. The second one states that no element can belong to
\nm{Manager} at all time instants.
Note that both of these concept inclusions are meant to hold
\emph{globally}, that is, at all moments of time.

The same temporal operator $\SVbox$ together with \emph{rigid} roles
(i.e., roles that do not change over time) can be used to capture
timestamping of reified relationships.  If the relationship \nm{Member}
is reified by the concept \ex{C-Member} with two functional roles
\ex{org} and \ex{mbr}, satisfying the concept inclusions similar
to~\eqref{eq:reifi:1} and~\eqref{eq:reifi:2}, then
the requirement that both roles \ex{org} and \ex{mbr}
are rigid ensure that \nm{Member} is a snapshot relationship.  On the
other hand, for the reified temporary relationship \nm{Works On} we
require the concept inclusion
\begin{align*}
  \SVbox \ex{C-WorksOn} &~\sqsubseteq~\bot
\end{align*}
and two \emph{flexible} roles \ex{emp} and \ex{act}, which can change
arbitrarily.
Rigid roles are also used to represent both snapshot attributes and
snapshot binary relationships. Temporary attributes can be captured by
flexible roles or by using temporalised roles:
\begin{equation*}
\exists \SVbox \ex{salary} \sqsubseteq \bot,
\end{equation*}
where $\SVbox \ex{salary}$ denotes the intersection of the relations
$\ex{salary}$ at all time instants.

\smallskip

\emph{Evolution constraints} control how the domain elements evolve
over time by `migrating' from one entity to
another~\cite{gupta:hall:icde-91,Mendelzon:94,Su97,APS:amai07}. We distinguish between \emph{qualitative} evolution
constraints that describe generic temporal behaviour but do not
specify the moment of migration, and \emph{quantitative} evolution (or
transition) constraints that specify the exact moment of migration.
The dashed arrow marked with {\sc tex} ({\sc t}ransition {\sc
  ex}tension\footnote{We refer to~\cite{AKRZ:ER10} for a detailed
  explanation of the various evolution constraints and their naming
  convention.}) in Fig.~\ref{uml} is an example of a quantitative
evolution constraint meaning that each project expires in exactly
one time unit (one year) and becomes an instance of \nm{ExProject}.  The dashed arrow marked with
{\sc dev} ({\sc d}ynamic {\sc ev}olution) is a qualitative evolution
constraint meaning that every area manager will eventually become a
top manager. The {\sc dex}$^-$ ({\sc d}ynamic {\sc ex}tension) dashed
arrow says that every manager was once an employee, while the \textsc{pex} ({\sc p}ersistent {\sc ex}tension) dashed arrow means that a
manager will always be a manager and cannot be demoted.

In temporal description logic, these evolution constraints are
represented using temporal operators such as `at the next moment of
time' $\Rnext$, `sometime in the future' $\Rdiamond$, `sometime in
the past' $\Ldiamond$ and `always in the future' $\Rbox$:
\begin{align*}
\ex{Project} &~\sqsubseteq~\Rnext\ex{ExProject}, &
\ex{AreaManager} &~\sqsubseteq~\Rdiamond\ex{TopManager},\\
\ex{Manager} &~\sqsubseteq~\Ldiamond\ex{Employee}, &
\ex{Manager} &~\sqsubseteq~\Rbox\ex{Manager}.
\end{align*}
Again, these concept inclusions must hold globally. In the following,
the evolution constraints that involve $\Rdiamond$ and $\Ldiamond$
will be called \emph{migration constraints}.

\smallskip

\emph{Temporal cardinality constraints}~\cite{tauzovich:er-91,%
  mcbrien:et:al:cismod-92,gregersen:jensen:tr98} restrict the number
of times an instance of an entity participates in a
relationship. \emph{Snapshot} cardinality constraints do that at each
moment of time, while \emph{lifespan} cardinality constraints impose
restrictions over the entire existence of the instance as a member of
the entity.  In Fig.~\ref{uml}, we use $(k,l)$ to specify the
snapshot cardinalities and $[k,l]$ the lifespan cardinalities: for
example, at any moment, every top manager manages exactly one
project, but not more than five different projects over the whole
career. If the relationship \nm{Manages} is not reified and
represented by a role in temporal description logic then these two
constraints can be expressed by the following concept inclusions:
\begin{align*}
\ex{TopManager} &~\sqsubseteq~\exists \ex{manages} \sqcap  \mathop{\leq 1}\ex{manages}, &
\ex{TopManager} &~\sqsubseteq~\mathop{\leq 5}\SVdiamond \ex{manages},
\end{align*}
where $\SVdiamond$ means `sometime' (in the past, present or future),
and so $\SVdiamond \ex{manages}$ is the union of the relations 
$\ex{manages}$ over \emph{all} time instants.  
Snapshot and
lifespan cardinalities can also be expressed in a similar
way even for reified relationships (see, e.g.,~\eqref{eq:reifi:card} which captures snapshot
cardinalities). Observe that the above inclusions imply, in particular, that no one can remain a top manager for longer than five years (indeed, each top manager manages at least one project a year, each project expires in a year, and no top manager can manage more than five projects throughout the lifetime). However, this is inconsistent with `every manager always remains a manager'\!, and so the entity \nm{Manager} cannot be populated by instances, which, in turn, means that \nm{Project} must also be empty (since each project is managed by a top manager). One can make these entities consistent by, for example, dropping the persistence constraint on \nm{Manager} or the upper lifespan cardinality bound on the number of projects a top manager can manage throughout the lifetime. In large schemas, situations like this can easily remain undetected if the quality check is performed manually.

To represent temporal database instances, we use assertions like
$\Lnext\ex{Manager}(\ex{bob})$ for `Bob was a manager last year' and
$\Rnext \ex{manages}(\ex{bob},\ex{cronos})$ for `Bob will manage
Cronos next year.'

\subsection{Temporal \DL{} Logics}
\label{sec:tdl}

It is known from temporal logic~\cite{Gabbayetal94} that all the
temporal operators used in the previous section can be expressed in
terms of the binary operators $\S$ `since' and $\U$ `until' (details
will be given below). So we formulate our `base' temporal extension
\TuDLbn{} of the description logic \DLbn{} using only these two
operators.
The language of \TuDLbn{} contains \emph{object names} $a_0, a_1,
\dots $, \emph{concept names} $A_0, A_1, \dots$, \emph{flexible role
  names} $P_0, P_1, \dots$, and \emph{rigid role names} $G_0, G_1,
\dots$. \emph{Role names} $S$ and \emph{roles} $R$ are defined by
taking
\begin{align*}
  S \ \ ::= \ \ & P_i \ \ \mid \ \  G_i\qquad\text{and}\qquad   R \ \ ::= \ \ S \ \ \mid\ \  S^-.
\end{align*}
We say $R$ is a \emph{rigid role} if it is of the form $G_i$ or $G_i^-$, for a rigid role name $G_i$. 
\emph{Basic concepts} $B$, \emph{concepts} $C$ and \emph{temporal concepts} $D$ are given by the following grammar: 
\begin{align*}
  B \ \ ::= \ \  & \bot \ \ \mid\ \ 
  A_i\ \ \mid\ \  \mathop{\geq q} R, 
  \\
  C \ \ ::= \ \ & B \ \ \mid \ \ D  \ \ \mid
  \ \ \neg C \ \ \mid\ \ C_1\sqcap C_2,\\
  D \ \ ::= \ \ & C \ \ \mid \ \ C_1\U C_2\mid\ \ C_1\S C_2,
\end{align*}
where, as before, $q$ is a positive integer given in binary. (We use two separate rules for $C$ and $D$ here because in the definitions of the fragments of \TuDLbn{} below these rules will be restricted to the corresponding sub-Boolean and temporal fragments.)
A \TuDLbn{} \emph{TBox}, $\T$, is a finite set of
\emph{concept inclusions} of the form $C_1 \sqsubseteq C_2$. An
\emph{ABox}, $\A$, consists of assertions of the form
\begin{equation*}
\nxt^n
A_k(a_i),\qquad \nxt^n \neg A_k(a_i),\qquad \nxt^n S(a_i,a_j)\quad\text{and}\quad \nxt^n \neg S(a_i,a_j),
\end{equation*}
where $A_k$ is a concept name, $S$ a (flexible
or rigid) role name, $a_i$, $a_j$ object names and, for $n \in \Z$,
\begin{equation*}
\nxt^n = \underbrace{\Rnext\cdots\Rnext}_{n \text{ times}}, \   \text{ if } n \geq 0,\qquad\text{and}\quad
\nxt^n = \underbrace{\Lnext\cdots\Lnext}_{-n \text{ times}},  \ \text{ if } n < 0.
\end{equation*}
Note that we use $\nxt^n$ as an abbreviation and take the size of $\nxt^n$ to be $n$ (in other words, the numbers $n$ in ABox assertions are given in \emph{unary}).
Taken together, the TBox $\T$ and ABox $\A$ form the
\emph{knowledge base} (KB) $\K=(\T,\A)$.

A \emph{temporal interpretation} is a pair $\I = (\Delta^\I,\cdot^{\I(n)})$, where $\Delta^\I$ is a non-empty interpretation domain and $\I(n)$ gives a standard DL interpretation 
for each time instant $n \in \Z$:
\begin{equation*}
  \I(n)~=~\bigl(\Delta^\I, a_0^\I,\dots ,A_0^{\I(n)},
  \dots ,P_0^{\I(n)},\dots ,G_0^{\I}, \dots\bigr).
\end{equation*}
We assume, however, that the domain $\Delta^\I$ and the interpretations $a_i^\I\in \Delta^\I$ of object names and $G_i^{\I} \subseteq
\Delta^\I\times\Delta^\I$ of rigid role names are the same for all $n\in \Z$. (For a discussion of the constant domain assumption, consult~\cite{GKWZ03}. Recall also that we adopt the UNA.) The interpretations $A_i^{\I(n)}\subseteq\Delta^\I$ of concept names and $P_i^{\I(n)} \subseteq \Delta^\I\times\Delta^\I$ of flexible role names can vary.
The atemporal constructs are interpreted in $\I(n)$ as before; we write $C^{\smash{\I(n)}}$ for the extension of $C$ in $\I(n)$.
The interpretation of the temporal operators is as follows:
\begin{align*}
  (C_1\U C_2)^{\I(n)} & ~=~
  \bigcup_{k>n}\bigl(C_2^{\I(k)}
  \cap \bigcap_{n < m < k}\hspace*{-0.5em} C_1^{\I(m)}\bigr),\\
  (C_1\S C_2)^{\I(n)} & ~=~ 
  \bigcup_{k<n}\bigl(C_2^{\I(k)} \cap \bigcap_{n > m > k}\hspace*{-0.5em}
  C_1^{\I(m)}\bigr).
\end{align*}
Thus, for example, $x \in (C_1\U C_2)^{\smash{\I(n)}}$ iff there is a moment $k>n$ such that $x \in C_2^{\smash{\I(k)}}$ and $x \in C_1^{\smash{\I(m)}}$, for all moments $m$ between $n$ and $k$. Note that the operators $\S$ and $\U$ (as well as the $\Box$ and $\Diamond$ operators to be defined below) are `strict' in the sense that their semantics does not include the current moment of time. The non-strict operators, which include the current moment, are obviously definable in terms of the strict ones. 

As noted above, for the aims of TCM it is enough to interpret concept inclusions in $\I$ \emph{globally}:
\begin{equation*}
\I \models C_1 \sqsubseteq C_2 \quad \text{iff} \quad C_1^{\I(n)} \subseteq C_2^{\I(n)} \quad \text{for \emph{all} $n \in \Z$}.
\end{equation*}
ABox assertions are interpreted relatively to the \emph{initial moment}, 0. Thus, we set:
\begin{align*}
\I \models \nxt^n A_k(a_i) \quad &\text{iff} \quad a_i^\I \in A_k^{\I(n)}, &
\I \models \nxt^n S(a_i,a_j) \quad &\text{iff} \quad (a_i^\I, a_j^\I) \in S^{\I(n)}, \\ %
\I \models \nxt^n \neg A_k(a_i) \quad &\text{iff} \quad a_i^\I \notin A_k^{\I(n)}, &
\I \models \nxt^n \neg S(a_i,a_j) \quad &\text{iff} \quad (a_i^\I, a_j^\I) \notin S^{\I(n)}.
\end{align*}
We call $\I$ a \emph{model} of a KB $\K$ and write
$\I\models \K$ if $\I$ satisfies all
elements of $\K$. $\K$ is \emph{satisfiable} if it  has a model. A concept $C$ (role $R$) is
\emph{satisfiable} with respect to $\K$ if there are a model
$\I$ of \K{} and $n\in \Z$ such that
$C^{\I(n)}\neq \emptyset$ (respectively,
$R^{\I(n)}\neq \emptyset)$. It is readily seen that the concept and role satisfiability problems are equivalent to KB satisfiability.

We now define a few fragments and extensions of the base language \TuDLbn. Recall that to say that $C$ is a snapshot concept, we need the `always' operator $\SVbox$ with the following meaning:
\begin{equation*}
(\SVbox C)^{\I(n)}  ~=~ \bigcap_{k \in \Z} C^{\I(k)}.
\end{equation*}
The dual operator `sometime' is defined as usual: $\SVdiamond C = \neg\SVbox \neg C$. 
In terms of $\S$ and $\U$, it can be represented as $\SVdiamond C = \top \U (\top \S C)$.  
Define
\zSVDLb{} to be the sublanguage of \TuDLbn{} the temporal concepts $D$
in which are of the form
\begin{equation*}\tag{U}
D \ \ ::=\ \ C \ \ \mid\ \ \SVbox C.
\end{equation*}
Thus, in \zSVDLb, we can express timestamping constraints (see Section~\ref{subsec:t-mot-ex}).

The temporal operators $\Rdiamond$ (`sometime in the future') and
$\Ldiamond$ (`sometime in the past') that are required for qualitative evolution constraints with the standard temporal logic semantics
\begin{equation*}
  (\Rdiamond C)^{\I(n)} ~=~ \bigcup_{k>n}  C^{\I(k)}\quad\text{and}\quad
  (\Ldiamond C)^{\I(n)} ~=~ \bigcup_{k<n}  C^{\I(k)}
\end{equation*}
can be expressed via $\U$ and $\S$ as $\Rdiamond C = \top \U C$ and $\Ldiamond C = \top \S C$; the operators $\Rbox$ (`always in the future') and $\Lbox$ (`always in the past') are defined as dual to $\Rdiamond$ and $\Ldiamond$: $\Rbox C = \neg \Rdiamond \neg C$ and $\Lbox C = \neg \Ldiamond \neg C$.
We define the fragment \TdDLbn{} of \TuDLbn{} by restricting the temporal concepts $D$ to the form
\begin{equation*}\tag{FP}
D \ \ ::= \ \ C \ \ \mid \ \  \Rbox C\ \ \mid \ \ \Lbox C.
\end{equation*}
Clearly, we have the following equivalences:
\begin{equation*}
\SVbox C = \Rbox\Lbox C\qquad \text{and}\qquad \SVdiamond C = \Rdiamond\Ldiamond C.
\end{equation*}
In what follows, these equivalences will be regarded as definitions for $\SVbox$ and $\SVdiamond$ in those languages where they are not explicitly present.
Thus, \TdDLbn{} is capable of expressing both time\-stamping and qualitative (but not quantitative) evolution constraints.

The temporal operators $\Rnext$ (`next time') and $\Lnext$ (`previous time'), used in quantitative evolution constraints, can be defined as $\Rnext C = \bot \U C$ and $\Lnext C = \bot \S C$, so that we have
\begin{equation*}
  (\Rnext C)^{\I(n)} = C^{\I(n+1)}\quad\text{and}\quad(\Lnext C)^{\I(n)} = C^{\I(n-1)}.
\end{equation*}
The fragment of \TuDLbn{} with temporal concepts of the form
\begin{equation*}
D \ \ ::=\ \  C \ \ \mid\ \ \Rbox C\ \  
\mid \ \ \Lbox C \ \ \mid \ \ \Rnext C \ \ \mid \ \  \Lnext C
\tag{FPX}
\end{equation*}
will be denoted by \TdxDLbn. In this fragment, we can express timestamping, qualitative and
quantitative evolution constraints.

Thus, we have the following inclusions between the languages introduced above:
\begin{equation*}
\zSVDLb \quad\subseteq\quad \TdDLbn \quad\subseteq\quad \TdxDLbn \quad\subseteq\quad \TuDLbn.
\end{equation*}

Similarly to the atemporal case, we can identify sub-Boolean
fragments of the above languages. A temporal TBox is called
a \emph{Krom} or a \emph{core} TBox if it contains only concept
inclusions of the form
\begin{gather*}\tag{\textit{krom}}
D_1 \sqsubseteq D_2, \qquad D_1 \sqsubseteq \neg D_2, \qquad \neg D_1 \sqsubseteq D_2,\\
\tag{\textit{core}}
D_1 \sqsubseteq D_2, \qquad D_1  \sqcap D_2\sqsubseteq \bot,
\end{gather*}
respectively, where the $D_i$ are temporal concepts defined
by~{(FPX)}, {(FP)} or {(U)} with
\begin{equation*}
C \ \ ::= \ \ B \ \ \mid \ \ D.
\end{equation*}
Note that no Boolean operators are allowed in the $D_i$.  This gives us 6
fragments: \TdxDLan{}, \TdDLan{} and \zSVDLan{}, for
$\alpha\in\{\textit{core},\textit{krom}\}$. 
\begin{remark}\label{subUS}
  We do not consider the core and Krom fragments of the full language
  with since ($\S$) and until ($\U$) because, as we shall see in
  Section~\ref{sec:compex-rigid-roles} (Theorem~\ref{thm:pspace:2}),
  these operators allow one to go beyond the language of binary
  clauses of the core and Krom fragments, and the resulting languages
  would have the same complexity as \TuDLbn{} (but less expressive).
\end{remark}
\begin{remark}\label{krom}
The introduced fragments of the full language $\TuDLbn$ do not contain $\Rdiamond$ and $\Ldiamond$. Both operators, however,  can be defined in the Krom and Bool fragments. For example, the concept inclusion $\Ldiamond B_1 \sqsubseteq \Rdiamond B_2$ can be represented by means of the inclusions
\begin{equation*}
\Rbox A_2 \sqsubseteq \Lbox A_1\qquad\text{ and }\qquad  A_i \sqsubseteq \neg B_i,\quad \neg B_i \sqsubseteq A_i, \quad\text{ for } i=1,2.
\end{equation*}
In the core fragments, where we do not have negation in the left-hand side, this trick does not work. Therefore, evolution constraints involving $\Ldiamond$ or $\Rdiamond$ (such as $\ex{Manager} \sqsubseteq
\Ldiamond\ex{Employee}$) are not expressible in the core fragments (but timestamping remains expressible). 
\end{remark}

As we have seen in our running example, in order to express lifespan
cardinality constraints, temporal operators on roles are required. For a role $R$ of the form
\begin{equation*}
R \ \ ::=\ \  S \ \ \mid\ \  S^- \ \ \mid \ \ \SVdiamond R \ \ \mid \ \  \SVbox R,
\end{equation*}
we define the extensions of $\SVdiamond R$ and $\SVbox R$ in an interpretation $\I$ by taking
\begin{equation*}
  (\SVdiamond R)^{\I(n)}=\bigcup_{k\in\Z}  R^{\I(k)}\quad\text{and}\quad
  (\SVbox R)^{\I(n)}=\bigcap_{k\in\Z}  R^{\I(k)}.
\end{equation*}
In this article we consider three extensions of $\DLbn$ with such temporalised roles, which are denoted by
${\smash{\textsl{T}^{\,\ast}_{\beta}\DLbn}}$, for $\beta\in\{\textit{X}, \textit{FP}, \textit{U} \}$, where
${\smash{\textsl{T}^{\,\ast}_{\textit{X}}\DLbn}}$ allows only
$\Lnext,\Rnext$ as the temporal operators on concepts.

We can also extend our languages with role inclusions, which are interpreted globally (in the same way as concept inclusions):
\begin{equation*}
\I\models R_1 \sqsubseteq R_2 \quad\text{iff}\quad R_1^{\I(n)} \subseteq R_2^{\I(n)}, \quad \text{ for \emph{all} } n\in\Z.
\end{equation*}
These extensions are denoted by \ensuremath{\smash{\textsl{T}_{\mathcal{US}}\DLbrn}}, \ensuremath{\smash{\textsl{T}_{\textit{FP}}\DL_\textit{bool}^{(\mathcal{HN})}}}, etc. 

In the remaining part of the article, we investigate the computational complexity of the satisfiability problem for the temporal extensions of the \DL{} logics designed above. But before that we briefly summarise the obtained results in the more general context of temporal description logics.

\subsection{Summary of the Complexity Results and Related Work}
\label{sec:summary}

The temporal \DL{} logics we analyse here are collected in
Table~\ref{table:TDL-Lite:languages} together with the obtained and
known complexity results. (Note that the complexity bounds in
Table~\ref{table:TDL-Lite:languages} are all tight except the case of $\TdDLcn$, where we only have an upper bound.) To avoid clutter, we omitted
from the table the logics of the form
\ensuremath{\smash{\textsl{T}_{\beta}\DL_\alpha^{(\mathcal{HN})}}},
whose complexity is the same as the complexity of the respective
\ensuremath{\smash{\textsl{T}_{\beta}\DL_\alpha^{\mathcal{N}}}}.

\begin{table}[t]
\tbl{Complexity of the temporal \DL{} logics.
\label{table:TDL-Lite:languages}}{%
\newcommand{\cplx}[3][4.5em]{\rule[-8pt]{0pt}{20pt}\hspace*{1em}\parbox{#1}{\centering #2\\[-1pt]#3}}
\begin{minipage}[t]{\linewidth}
\centering%
\renewcommand{\arraystretch}{1.3}
\begin{tabular}[c]{|c|c|c|c|}
  \hline
   \raisebox{-2pt}{concept}        & \multicolumn{3}{c|}{temporal constructs}\\ \cline{2-4}
  \raisebox{2pt}{inclusions} & $\U/\S$, $\raisebox{2pt}{\tiny$\bigcirc$}\hspace{-0.15em}_{\scriptscriptstyle F}/\raisebox{2pt}{\tiny$\bigcirc$}\hspace{-0.15em}_{\scriptscriptstyle P}$, $\Rbox/\Lbox$ $^{a}$ &  $\Rbox/\Lbox$  & $\SVbox$\\\hline
Bool &  \raisebox{-0.6ex}{\rule[-10pt]{0pt}{28pt}\parbox{7.5em}{\centering \TuDLbn{}\\[2pt]\TdxDLbn{}}}\hspace*{-0.5em} \cplx{\PSpace}{Thm.~\ref{thm:pspace}} & \TdDLbn \cplx[4em]{\NP}{Thm.~\ref{thm:np1}} & \zSVDLb \cplx{\NP}{Thm.~\ref{thm:np1}} \\\hline 
Krom & \TdxDLkn \cplx{\NP}{Thm.~\ref{thm:krom-core}} & \TdDLkn \cplx[4em]{\NP}{Thm.~\ref{thm:krom-core}} &  \zSVDLk \cplx[5.4em]{\NLogSpace}{Thm.~\ref{thm:nlogspace}} \\\hline 
core & \TdxDLcn \cplx{\NP}{Thm.~\ref{thm:krom-core}} & \TdDLcn \cplx[4em]{$\leq$ \PTime}{Thm.~\ref{thm:core}} &  \zSVDLc \cplx[5.4em]{\NLogSpace}{} \\\hline\hline 
 \parbox{5.2em}{\centering temporalised\\[-2pt]roles} &   \TxDLbn  \cplx{undec.}{Thm.~\ref{thm:undec}}    &  \TdrDLbn \cplx{undec.}{Thm.~\ref{thm:undec}} & \TurDLbn \cplx{\NP}{Thm.~\ref{thm:R:NP}} \\\hline\hline
\rule[-14pt]{0pt}{32pt}\parbox{5.2em}{\centering unrestricted\\[-2pt]role\\[-2pt]inclusions} &   \ensuremath{\smash{\textsl{T}_{\mathcal{US}}\DLbrn}}\hspace*{-1em} \cplx[4.5em]{undec.}{\cite{GKWZ03}}    &  \ensuremath{\smash{\textsl{T}_{F\!P}\DLbrn}}\hspace*{-1em} \cplx[4.5em]{undec.}{\cite{GKWZ03}} &  \ensuremath{\smash{\textsl{T}^{\,\ast}_{U}\DLbrn}}\hspace*{-1em} \cplx[5.5em]{2\ExpTime}{\cite{ALT:ijcai07}} \\\hline
\end{tabular}\\[4pt]
\hspace*{0.3em}\rule{10em}{0.1pt}\hfill\mbox{}\\
\hspace*{0.5em}$^{a}$ Sub-Boolean fragments of the language with $\U/\S$ are not defined (see Remark~\ref{subUS}).\hfill\mbox{}
\end{minipage}%
}
\end{table}


The analysis of the constructs required for temporal conceptual
modelling in Sections~\ref{subsec:mot-ex} and~\ref{subsec:t-mot-ex}
has led us to temporalisations of \DL{} logics, interpreted over the Cartesian products of object domains and the
flow of time $(\Z,<)$, in which (1) the future and past
temporal operators can be applied to concepts; (2) roles can be
declared flexible or rigid; (3) the `undirected' temporal operators
`always' and `sometime' can be applied to roles; (4) the concept
and role inclusions are global, and the database (ABox) assertions are specified
to hold at particular moments of time.

The minimal logic required to capture all of the temporal and static
conceptual modelling constraints is
\ensuremath{\smash{\textsl{T}^{\,\ast}_{\textit{FPX}}\DL_\textit{bool}^{\mathcal{HN}}}};
alas, it is undecidable. In fact, even the logic \TxDLbn, capturing
only the quantitative evolution constraints, lifespan
cardinalities and covering, 
is undecidable. Replacing `quantitative' with `qualitative'---i.e.,
considering \TdrDLbn---does not beat undecidability in the presence of
lifespan cardinalities.  Both these
undecidability results will still hold if we replace arbitrary cardinality constraints ($\mathcal{N}$) with role functionality. 
To regain decidability in the presence of temporalised roles, we have
to limit the temporal operators on concepts to the undirected
operators $\SVdiamond/\SVbox$---thus restricting the language to only
timestamping and lifespan cardinalities. We show that the logic
\TurDLbn{} is \NP-complete using the quasimodel technique.

Logics in the last row have arbitrary role inclusions, which together with functionality constraints are expressive enough to model all \ALC{} constructors~\cite{ACKRZ:er07,ACKZ:jair09}, and so the resulting TDLs are as complex as the corresponding temporal extensions of \ALC.

On a positive note, logics with restricted role inclusions and no temporal operators on roles exhibit much better computational properties. Our smallest logic,
$\smash{\textsl{T}_{\textit{U}}\DL_\textit{core}^{(\mathcal{HN})}}$,
 is \NLogSpace-complete. In the temporal dimension, it can only
express timestamping constraints. It can also capture all the static
constraints that are different from covering and do not involve any
interaction between role inclusions and number restrictions. Extending the language with covering leads to the loss of tractability in $\smash{\textsl{T}_{\textit{U}}\DL_\textit{bool}^{(\mathcal{HN})}}$.
When covering is not needed and we are interested in temporal
constraints different from lifespan cardinalities, we can regain
tractability if we restrict the language to timestamping and evolution
constraints that only capture persistence
($\smash{\textsl{T}_{\textit{FP}}\DL_\textit{core}^{(\mathcal{HN})}}$).
If we also require migration constraints (that involve 
$\Ldiamond$ and $\Rdiamond$; see Remark~\ref{krom}) then we can
use the Krom language $\smash{\textsl{T}_{\textit{FP}}\DL_\textit{krom}^{(\mathcal{HN})}}$,
which is again \NP-complete.  Surprisingly, the addition of the full set of
evolution constraints  makes reasoning \NP-complete even in
$\smash{\textsl{T}_{\textit{FPX}}\DL_\textit{core}^{(\mathcal{HN})}}$.

To better appreciate the formalisms designed in this article, we
consider them in a more general context of temporal description logics (for more detailed surveys,
consult~\cite{ArFr01,AF05,GKWZ03,LuWoZa-TIME-08}). 
Historically, the first temporal extensions of DLs were \emph{interval-based}~\cite{schmiedel:90}. \citeN{Bett97} considered interval-based temporal extensions of \ALC{} in the style of \citeN{HaSh91} and established their undecidability. \citeN{ArFr98} gave a subclass of decidable interval-based
temporal DLs.

Numerous \emph{point-based} temporal DLs have been constructed and
investigated since Schild's seminal paper~\citeyear{Schild93}. 
One of the lessons of the 20-year history of the discipline
is that logics interpreted over two- (or more) dimensional structures
are very complex and sensitive to subtle interactions between
constructs operating in different dimensions. 
The first TDLs suggested for representing TCMs were based on the expressive DLs \DLR{} and \ALCQI~\cite{Ar:Fr:er-99}. However, it turned out that already a single rigid role and the operator $\Rdiamond$ (or $\Rnext$) on \ALC-concepts led to undecidability~\cite{WoZa99b}. In fact, to construct an undecidable TDL, one only needs a rigid role and three concept constructs: $\sqcap$, $\exists R. C$ and $\Rdiamond$, that is, a temporalised $\mathcal{EL}$~\cite{AKLWZ:time07}. There have been several attempts to tame the bad computational behaviour of TDLs by imposing various restrictions on the DL and temporal components as well as their interaction.

One approach was to disallow rigid roles and temporal operators on roles, which resulted in \ExpSpace-complete temporalisations of \ALC~\cite{AFWZ:02,GKWZ03}. Such temporalisations reside in the monodic fragment\footnote{A temporal formula is \emph{monodic} if all of its sub-formulas beginning with a temporal operator have at most one free variable.} of first-order temporal logic~\cite{hodk:wolter:mz:dec-tl:99}, for which 
tableau~\cite{LutzSWZ02,KontchakovLWZ04} and
resolution~\cite{DegtyarevFK06} reasoning algorithms have been developed
and implemented~\cite{HustadtKRV04,Guensel05,LudwigH10}.
Another idea was to weaken the whole temporal component to the `undirected' temporal operators $\SVbox$ and $\SVdiamond$ (which
cannot discriminate between past, present and future) on concepts and roles, resulting in a 2\ExpTime-complete extension of \ALC~\cite{ALT:ijcai07}.
The third approach was to allow arbitrary temporal operators on \ALC{} axioms only (but not on concepts or roles)~\cite{BaaGhiLu-KR08,BaaderGL12}, which gave an \ExpTime-complete logic. The addition of rigid concepts to this logic increases the complexity to \NExpTime, while rigid concepts and roles make it 2\ExpTime-complete.
Finally, the fourth approach, which dates back to Schild~\citeyear{Schild93}, was to use only global axioms. In this case, \ALC{} with temporal operators on concepts is \ExpTime-complete, which matches the complexity of \ALC{} itself (in contrast, temporal operators on axioms and concepts make the less expressive $\DL_\textit{bool}$ \ExpSpace{}-complete~\cite{AKLWZ:time07}).

As argued above, global axioms are precisely what we need in TCM. On the other hand, to capture timestamping and evolution constraints  we need the full set of temporal operators on concepts, while to capture lifespan cardinalities and timestamping on
relations we need temporalised or rigid roles. To achieve decidability in the case with rigid roles, we also weaken \ALC{} to \DL, which, as we have seen above, perfectly suits the purpose of conceptual modelling. We thus start from the first promising results of~\citeN{AKRZ:09}, which demonstrated that even with rigid roles temporal extensions of $\DLbn$ could be decidable, and extend them to various combinations of temporal operators and  different sub-Boolean fragments of $\DLbn$, proving encouraging complexity results and showing how these logics can represent TCM schemas.

The results in the first three rows of Table~\ref{table:TDL-Lite:languages} are established by using embeddings into the propositional temporal logic \PTL. To cope with the sub-Boolean core and Krom logics, we introduce, in Section~\ref{sec:2ltl}, a number of new fragments of \PTL{} by restricting the type of clauses in Separated Normal Form~\cite{DBLP:conf/ijcai/Fisher91} and the available temporal operators. The obtained complexity classification in Table~\ref{tab:PTL-fragments} helps understand the results in the first three rows of Table~\ref{table:TDL-Lite:languages}.


\section{Reducing Temporal \DL{} to Propositional Temporal Logic}
\label{sec:tdl-rigid:Z:bool-until-complex}

In this section we reduce the satisfiability problem for \TuDLbn{} KBs to the satisfiability problem for propositional temporal logic. This will be achieved in two steps. First, we reduce \TuDLbn{} to
the one-variable first-order temporal logic \QTLi{} \cite{GKWZ03}. And then we show that satisfiability of the resulting \QTLi-formulas can be further reduced to satisfiability of \QTLi-formulas without
positive occurrences of existential quantifiers, which are essentially propositional temporal formulas. 
To simplify presentation, we consider here the logic \TuDLbn{}
(without role inclusions). The full
\ensuremath{\smash{\textsl{T}_{\mathcal{US}}\DL_\textit{bool}^{(\mathcal{HN})}}}
requires a bit more elaborate (yet absolutely routine) reduction that
is similar to the one given by~\citeN{AKRZ:09} for the atemporal case.


\subsection{First-Order Temporal Logic}

The language of \emph{first-order temporal logic} $\QTL$ contains \emph{predicates} $P_0, P_1, \dots$ (each with its arity), \emph{variables} $x_0,x_1,\dots$ and \emph{constants} $a_0,a_1,\dots$. \emph{Formulas} $\varphi$ of $\QTL$ are defined by the grammar:
\begin{equation*}
\varphi \ \ ::= \ \ P_i(t_1,\dots,t_{k_i}) \ \ \mid \ \ \bot \ \ \mid \ \ \forall x\,\varphi \ \ \mid \ \ \neg \varphi \ \ \mid \ \ \varphi_1 \land \varphi_2 \ \ \mid \varphi_1 \U \varphi_2 \ \ \mid \ \ \varphi_1 \S \varphi_2,
\end{equation*}
where $k_i$ is the arity of $P_i$ and the $t_j$ are \emph{terms}---i.e., variables or constants. These formulas are interpreted in \emph{first-order temporal models} $\Mmf$, which, for every $n\in \Z$, give a first-order structure 
\begin{equation*}
\Mmf(n) = (\Delta^\Mmf, a_0^{\Mmf}, a_1^{\Mmf}, \dots, P_0^{\Mmf,n},P_1^{\Mmf,n}\dots)
\end{equation*}
with the same domain $\Delta^\Mmf$, the same $a_i^\Mmf\in\Delta$, for each constant $a_i$, and where  $P_i^{\Mmf,n}$ is a $k_i$-ary relation on $\Delta^{\Mmf}$, for each predicate $P_i$ of arity $k_i$ and each $n\in\Z$. An \emph{assignment} in $\Mmf$ is a function, $\mathfrak{a}$, that maps variables to elements of $\Delta^\Mmf$. For a term $t$, we write $t^{\mathfrak{a},\Mmf}$ for $\mathfrak{a}(x)$ if $t = x$, and for $a^\Mmf$ if $t = a$. The semantics of $\QTL$ is standard (see, e.g.,~\cite{GKWZ03}):
\begin{align*}
&\Mmf,n\models^\mathfrak{a} P(t_1,\dots,t_k) \quad\text{iff}\quad (t_1^{\mathfrak{a},\Mmf},\dots,t_k^{\mathfrak{a},\Mmf})\in P^{\Mmf,n},\\
&\Mmf,n \not\models^\mathfrak{a} \bot,\\
&\Mmf,n\models^\mathfrak{a} \forall x\,\varphi \quad\text{iff}\quad \Mmf,n\models^{\smash{\mathfrak{a}'}} \varphi, \text{ for all assignments $\smash{\mathfrak{a}'}$ that differ from $\mathfrak{a}$ on $x$ only},\\
&\Mmf,n \models^\mathfrak{a} \neg \varphi \quad\text{iff}\quad \Mmf,n \not \models^\mathfrak{a} \varphi,\\
&\Mmf,n \models^\mathfrak{a} \varphi_1 \land \varphi_2 \quad\text{iff}\quad \Mmf,n \models^\mathfrak{a} \varphi_1 \text{ and } \Mmf,n \models^\mathfrak{a} \varphi_2,\\
&\Mmf,n\models^\mathfrak{a} \varphi_1\U\varphi_2 \quad\text{iff}\quad \text{there is } k > n \text{ with } \Mmf,k\models^\mathfrak{a}\varphi_2
\text{ and } \Mmf,m\models^\mathfrak{a}\varphi_1, \text{ for } n < m < k,\\
&\Mmf,n\models^\mathfrak{a} \varphi_1\S\varphi_2 \quad\text{iff}\quad \text{there is } k < n \text{ with } \Mmf,k\models^\mathfrak{a}\varphi_2 
\text{ and } \Mmf,m\models^\mathfrak{a}\varphi_1, \text{ for } k < m < n.
\end{align*}
We use the standard abbreviations such as 
\begin{multline*}
\top = \neg \bot,\qquad\varphi_1\lor\varphi_2 = \neg(\neg\varphi_1\land\neg\varphi_2),\qquad
\exists x \, \varphi = \neg \forall x\, \neg \varphi, \\ \Rdiamond \varphi = \top\U\varphi, \qquad \Rbox \varphi = \neg\Rdiamond\neg\varphi,\qquad \Rnext \varphi = \bot\U\varphi
\end{multline*}
as well as the past counterparts for $\Ldiamond$, $\Lbox$ and $\Lnext$; we also write $\SVbox \varphi$ for $\Rbox\Lbox\varphi$.

If a formula $\varphi$ contains no free variables (i.e., $\varphi$ is a sentence), then we omit the valuation $\mathfrak a$ in $\Mmf,n\models^{\mathfrak a}\varphi$ and write
$\Mmf,n\models\varphi$. If $\varphi$ has a single
free variable $x$, then we write $\Mmf,n\models\varphi[a]$
in place of $\Mmf,n\models^\mathfrak{a}\varphi$ with
$\mathfrak{a}(x) = a$.

A $\QTLi$-\emph{formula} is a $\QTL$-formula which is constructed using at most one variable. Satisfiability of $\QTLi$-formulas is known to be \ExpSpace-complete~\cite{GKWZ03}. 
In the \emph{propositional temporal logic}, $\PTL$, only 0-ary predicates (that is, propositional variables) are allowed. The satisfiability problem for $\PTL$-formulas is \PSpace-complete~\cite{SistlaClarke82}.


\subsection{Reduction to $\QTLi$}\label{sec:tld2tfol-red}

Given a \TuDLbn{} KB $\K=(\T,\A)$, let $\ob$
be the set of all object names occurring in $\A$ and
$\role$ the set of rigid and flexible role names occurring in
$\K$ and their
inverses.

In our reduction, objects $a \in \ob$ are mapped
to constants $a$,  concept names $A$ to
unary predicates $A(x)$, and number restrictions $\mathop{\geq q} R$
to unary predicates $E_qR(x)$. Intuitively, for a role name $S$, the predicates $E_qS(x)$ and $E_qS^-(x)$ represent, at each moment of time, the sets of elements with \emph{at least
$q$} distinct $S$-successors and \emph{at least $q$} distinct
$S$-predecessors; in particular, $E_1S(x)$ can be thought of as the domain of $S$ and $E_1S^-(x)$ as the range of $S$. By induction on the
construction of a \TuDLbn{} concept $C$, we define the \QTLi{}-formula $C^*(x)$:
\begin{align*}%
  A^* & = A(x), & \bot^* & =\bot,& (\mathop{\geq q} R)^* & = E_qR(x), \\ %
  (C_1\U C_2)^* & = C_1^* \U C_2^*, &
  (C_1\S C_2)^* & =C_1^* \S C_2^*,& (C_1 \sqcap C_2)^* &= C_1^* \land  C_2^*, &
  (\neg C)^* &= \neg C^*.
\end{align*}
It can be easily seen that the map $\cdot^*$ commutes with all the Boolean and temporal operators: e.g., $(\Rdiamond C)^* = \Rdiamond C^*$. For a TBox $\T$, we consider the following sentence, saying that the concept inclusions in $\T$ hold globally:
\begin{equation*}
\T^\dagger \quad=\ \ \bigwedge_{C_1 \sqsubseteq C_2\in\T}
\hspace*{-1em} \SVbox\forall x\,\bigl(C_1^*(x) \to C_2^*(x)\bigr).
\end{equation*}
In the translation above, we replaced binary predicates (i.e., roles) by collections of unary predicates, the $E_qR(x)$. Clearly, we have to ensure that these predicates behave similarly to the respective number restrictions. In particular, the following three properties trivially hold in \TuDLbn{} interpretations, for all roles $R$ at all moments of time:
\begin{akrzitemize}
\item every point with at least $q'$ $R$-successors has at least $q$ $R$-successors, for each $q < q'$;

\item if $R$ is a rigid role, then every point with at least $q$ $R$-successors at some moment has at least $q$ $R$-successors at all moments of time;

\item if the domain of a role is not empty, then its range is not empty either.
\end{akrzitemize}
These conditions can be encoded by the following \QTLi-sentences:
\begin{align}
  \label{eq:role:saturation}&\bigwedge_{R\in\role}\hspace*{1em}
  \bigwedge_{q,q'\in \QT \text{ with } q' > q}\hspace*{-1em}\SVbox\ \forall
  x\,\bigl((\mathop{\geq q'}R)^*(x) \to (\mathop{\geq q}R)^*(x)\bigr),\\
  \label{eq:role:global:relation}&
  \bigwedge_{R \in\role \text{ is rigid}} \hspace*{1em}\bigwedge_{q\in
    \QT}\hspace*{-0.4em} \SVbox\ \forall x\,\bigl((\mathop{\geq q}R)^*(x)
  \rightarrow \SVbox\,(\mathop{\geq q}R)^*(x)\bigr),\\%
  \label{eq:role:existence:2} & \bigwedge_{R\in\role} \SVbox\,
  \bigl(\exists x\, (\exists R)^*(x) \ \rightarrow \ \exists
  x\,(\exists \inv{R})^*(x)\bigr),
\end{align}%
where $\QT$\label{def:QK} is the set containing 1 and all the numbers $q$ such that
$\mathop{\geq q} R$ occurs in $\T$, and $\inv{R}$ is the inverse of $R$, i.e., $\inv{S} = S^-$ and $\inv{S^-} = S$, for a
role name $S$. As we shall see later, these three properties are enough
to ensure that the real binary relations for roles $S$ in \TuDLbn{}
can be reconstructed from the collections of unary predicates $E_qS(x)$ and $E_qS^-(x)$ satisfying~\eqref{eq:role:saturation}--\eqref{eq:role:existence:2}.

It is easy to extend the above reduction to ABox concept assertions: take $\nxt^n A(a)$ for each $\nxt^n A(a)\in\A$, and
$\nxt^n \neg A(a)$ for each $\nxt^n \neg A(a)\in\A$. However, ABox role assertions need a more elaborate treatment. For every $a \in \ob$, if $a$ has $q$ $R$-successors in
$\A$ at moment $n$---i.e., $\nxt^n R(a, b_1),\dots, \nxt^n
R(a,b_q) \in \A$, for distinct $b_1,\dots,b_q$---then we include $(\nxt^n\mathop{\geq
  q}R)^*(a)$ in the translation of $\A$. When counting the number
of successors, one has to remember the following property
of rigid roles $S$: if an ABox contains $\nxt^m S(a,b)$ then $\nxt^n
S(a,b)$ holds for all $n\in\Z$, and so $\nxt^m S(a,b)$ contributes to the
number of $S$-successors of $a$ and $S^-$-successors of $b$ at
\emph{each} moment.

In what follows, we assume that $\A$ contains $\nxt^n
S^-(b,a)$ whenever it contains $\nxt^n S(a,b)$.\label{def:invABox}  For each $n\in \Z$ and each role $R$, we define the \emph{temporal slice}
$\A_n^R$ of $\A$ by taking
\begin{equation*}
  \A_n^R ~=~ \begin{cases}\bigl\{ R(a,b) \mid \nxt^m R(a,b)\in \A \text{ for some
  }m \in \Z \bigr\}, & R \text{ is a rigid role},\\
  \bigl\{ R(a,b) \mid \nxt^n R(a,b) \in \A \bigr\}, & R \text{ is a flexible role}.\end{cases}
\end{equation*}
The translation $\A^\dagger$ of the \TuDLbn{} ABox $\A$ is
defined now by taking
\begin{equation*}
  \A^\dagger = \hspace*{-0.5em}\bigwedge_{{\scriptscriptstyle\bigcirc}^n A(a)\in \A}\hspace*{-1.5em}
  \nxt^n A(a)\hspace*{0.7em} \land\hspace*{-0.2em}
  \bigwedge_{{\scriptscriptstyle\bigcirc}^n \neg A(a)\in \A}\hspace*{-1.5em} \nxt^n
  \neg A(a)\hspace*{1em} \land\hspace*{-0.2em}
  \bigwedge_{{\scriptscriptstyle\bigcirc}^n R(a,b)\in \A}\hspace*{-1.5em} \nxt^n
  (\mathop{\geq q^{R,n}_{\A(a)}} R)^*(a)\hspace*{1em}\land
  \bigwedge_{\begin{subarray}{c}{\scriptscriptstyle\bigcirc}^n \neg
      S(a,b)\in \A\\S(a,b)\in \A_n^S \end{subarray}}\hspace*{-1.5em} \bot,
\end{equation*}
where $\smash{q^{R,n}_{\A(a)}}$ is the number of distinct $R$-successors of $a$
in $\A$ at moment $n$:
\begin{equation*}
  q^{R,n}_{\A(a)} = \max \bigl\{ q \in \QT
  \mid R(a,b_1),\dots,R(a,b_q)\in\A_n^R,  \text { for distinct } b_1,\dots,b_q \bigr\}.
\end{equation*}
We note that $\A^\dagger$ can be effectively computed for any given $\K$ because we need  temporal slices $\A_n^R$ only for those $n$ that are explicitly mentioned in $\A$, i.e.,
those $n$ with $\nxt^n R(a,b)\in\A$.

Finally, we define the \QTLi-\emph{translation} $\K^\dagger$ of $\K=(\T,\A)$ as the conjunction
of $\T^\dagger$, $\A^\dagger$ and
formulas~\eqref{eq:role:saturation}--\eqref{eq:role:existence:2}. The size of $\T^\dagger$ and
$\A^\dagger$ does not exceed the size of $\T$ and
$\A$, respectively. The size 
of~\eqref{eq:role:global:relation} and~\eqref{eq:role:existence:2} is
linear in the size of $\T$, while the size 
of~\eqref{eq:role:saturation} is cubic in the size of
$\T$ (though it can be made linear by taking account of only  those $q$ that occur in $\mathop{\geq q} R$, for a fixed $R$, and replacing $q' > q$ in
the conjunction index with a more restrictive condition `$q' > q$ and
there is no $q''\in \QT$ with $q' > q'' > q$'; for details, see~\cite{ACKZ:jair09}).

The main technical result of this section is that 
$\K$ and $\K^\dagger$ are equisatisfiable; the  proof (based on the unravelling construction) is given in Appendix~\ref{app:proof:qtli-equisat}.
\begin{theorem}\label{lem:qtli-equisat}
A \TuDLbn{} KB $\K$ is satisfiable iff the \QTLi-sentence
  $\K^\dagger$ is satisfiable.
\end{theorem}

Meanwhile, we proceed to the second step of our reduction.

\subsection{Reduction to $\PTL$}

Our next aim is to construct a $\PTL$-formula that is equisatisfiable with $\K^\dagger$. First, we observe that $\K^\dagger$ can be represented in the form $\K^{\dagger_0} \land \bigwedge_{R\in\role} \vartheta_R$, where 
\begin{equation*}
\K^{\dagger_0} \ \ = \ \ \SVbox\forall x\,\varphi(x) \  \land \ \psi \qquad \text{ and }\qquad \vartheta_R =  
\SVbox\, \forall x\, \bigl( (\exists R)^*(x) \rightarrow \exists x\,(\exists \inv{R})^*(x)\bigr),
\end{equation*}
for a quantifier-free first-order temporal formula $\varphi(x)$
with a single variable $x$ and unary predicates only and a variable-free formula $\psi$. We show now that one can replace $\vartheta_R$ by a formula without existential quantifiers. To this end we require the following lemma:

\begin{lemma}\label{lem:roles}
  For every \TuDLbn{} KB $\K$, if there is a model $\Mmf$ satisfying
  $\K^{\dagger_0}$ such that $\Mmf,n_0\models (\exists R)^*[d]$, for
  some $n_0\in\Z$ and $d\in\Delta^{\Mmf}$, then there is a model
  $\Mmf'$ extending $\Mmf$ with new elements and satisfying
  $\K^{\dagger_0}$ such that, for each $n\in\Z$, there is
  $d_n\in\Delta^{\Mmf'}$ with $\Mmf',n\models (\exists R)^*[d_n]$.
\end{lemma}
\begin{proof}
Consider a new model $\Mmf'$ with domain $\Delta^{\Mmf} \cup
  \bigl(\{ d \} \times \Z\bigr)$ by setting
\begin{equation*}
B^{\Mmf',n} = B^{\Mmf,n} \cup \bigl\{ (d, k) \mid d\in B^{\Mmf,n - k}, k\in\Z \bigr\},\quad \text{ for each $B$  of $\K^{\dagger_0}$ and $n \in \Z$}.
\end{equation*}
In other words, $\Mmf'$ is the disjoint union of $\Mmf$ and copies of $d$ `shifted' along the timeline by $k$ steps, for each $k\in\Z$. It follows that, at each moment $n\in\Z$, the element $(d,n-n_0)$ belongs to $(\exists R)^*$, thus making $(\exists R)^*$ non-empty at all moments of time. Moreover, $\Mmf',0\models \K^{\dagger_0}$ because $\varphi(x)$ expresses a property of a single domain element and holds at \emph{each} moment of time, $\psi$ depends only on the part of the model that corresponds to constants (and which are interpreted as in $\Mmf$).
\end{proof}

Next, for each  $R\in\role$, we take a fresh constant $d_R$ and a fresh propositional variable $p_R$ (recall that $\inv{R}$ is also in $\role$),  and consider the following \QTLi{}-formula:
\begin{equation*}
\textstyle \K^{\ddagger} = \K^{\dagger_0}\land \bigwedge_{R\in\role} \vartheta_R',\text{ where }
\vartheta_R'  \ = \ 
\SVbox\forall x\,\bigl((\exists R)^*(x)
 \to \SVbox p_R\bigr) 
 \land
   \bigl(p_{\inv{R}} \to (\exists R)^*(d_{R}) \bigr)
\end{equation*}
($p_{\inv{R}}$ and $p_R$ indicate that $\inv{R}$ and $R$ are non-empty whereas $d_R$ and $d_{\inv{R}}$ witness that  at $0$).

\begin{lemma}\label{lemma:Kddagger}
  A \TuDLbn{} KB \K{} is satisfiable iff the \QTLi{}-sentence $\K^{\ddagger}$ is
  satisfiable.
\end{lemma}
\begin{proof}
  ($\Rightarrow$) If $\K$ is satisfiable then, by
  Theorem~\ref{lem:qtli-equisat} and repeated application of
  Lemma~\ref{lem:roles}, $\K^{\dagger_0}$ is satisfied in a model
  $\Mmf$ such that, for each $R\in\role$, the predicates $(\exists
  R)^*$ and $(\exists \inv{R})^*$ are either both empty at all moments
  of time or both non-empty at all moments of time. To satisfy the 
  $\vartheta_R'$, for $R\in\role$, we extend $\Mmf$ to $\Mmf'$ as follows: if $(\exists R)^*$ and $(\exists \inv{R})^*$ are non-empty, we set $p_R^{\smash{\Mmf'},n}$ to be true at all $n\in\Z$, and take $d_R$ to be an element in $((\exists R)^*)^{\Mmf,0}$; 
  otherwise, we set $p_R^{\smash{\Mmf'},0}$ to be false and take some domain element as $d_R$. It follows that $\Mmf',0\models\K^\ddagger$.

\smallskip

($\Leftarrow$) Conversely, suppose $\Mmf,0\models\K^\ddagger$. By
repeated application of Lemma~\ref{lem:roles}, $\K^{\dagger_0}$ is
satisfied in a model $\Mmf'$ such that, for
each $R\in\role$, the predicates $(\exists R)^*$ and $(\exists
\inv{R})^*$ are either both empty at all moments of time or both
non-empty at all moments of time. It follows that $\Mmf',0\models
\vartheta_R$ for all $R\in\role$, and so $\Mmf',0\models\K^\dagger$.
\end{proof} 

Finally, as $\K^{\ddagger}$ contains no existential quantifiers, it
can be regarded as a propositional temporal formula because all the
universally quantified variables can be instantiated by all the
constants in the formula (which only results in a polynomial
blow-up). Observe also that the translation $\cdot^\ddagger$ can be
done in logarithmic space in the size of $\K$.  This is almost trivial
for all conjuncts of $\K^\ddagger$ apart from $(\mathop{\geq
 q^{ \smash{R,n}}_{\A(a)}} R)^*$ in $\A^\dagger$, where the numbers can
be computed using a \LogSpace{}-transducer as follows: initially 
set $q = 0$; then enumerate all object names $b_i$ in $\A$
incrementing $q$ for each $R(a, b_i) \in \smash{\A_R^n}$ and stop if
$q = \max \QT$ or the end of the object names list is reached; let
$\smash{q^{R,n}_{\A(a)}}$ be the maximum number in $\QT$ not exceeding
$q$. Note that in the case of
$T_{\mathcal{US}}\bDLb^{(\mathcal{HN})}$, the translation is feasible
only in  \NLogSpace{} (rather than \LogSpace{}) because we have to
take account of role inclusions (and graph reachability is
\NLogSpace{}-complete).

\subsection{Complexity of \TuDLbn{} and its Fragments}
\label{sec:compex-rigid-roles}

We now use the translation $\cdot^\ddagger$ to obtain
the complexity results announced in Section~\ref{sec:summary}.

\begin{theorem}\label{thm:pspace} The satisfiability problem for \TuDLbn{} and \TdxDLbn{} KBs is \PSpace-complete.
\end{theorem}
\begin{proof}
  The upper bound follows from the reduction $\cdot^\ddagger$ above
   and the fact that $\PTL$ is \PSpace-complete over
  $(\Z,<)$~\cite{Rabi:10,Rey:10,SistlaClarke82}.  The lower
  bound is an immediate consequence of the observation that \TdxDLbn{}
  can encode formulas of the form $\theta \land
  \SVbox\bigwedge_i(\varphi_i\to\Rnext \psi_i)$, where $\theta$, the
  $\varphi_i$ and $\psi_i$ are conjunctions of propositional
  variables: satisfiability of such formulas is known to be
  \PSpace-hard (see e.g.,~\cite{Gabbayetal94}).
\end{proof}

In fact, using the $\U$-operator, we can establish the following:

\begin{theorem}\label{thm:pspace:2}
The satisfiability problem for the core fragment of \TuDLbn{} KBs is \PSpace-hard.
\end{theorem}
\begin{proof}
The proof is by reduction of the non-halting problem for deterministic Turing
 machines with a polynomial tape.
  Let $s(n)$ be a polynomial and $M$ a deterministic Turing machine
  that requires $s(n)$ cells of the tape given an input of length
  $n$. Without loss of generality, we assume that $M$ 
  never runs outside the first $s(n)$ cells.
  Let $M=\langle Q, \Gamma, \#, \Sigma, \delta, q_0, q_f
  \rangle$, where $Q$ is a finite set of states, $\Gamma$ a tape alphabet, $\# \in \Gamma$ the blank symbol, $\Sigma \subseteq
  \Gamma$ a set of input symbols, $\delta\colon
  (Q\setminus\{q_f\}) \times \Gamma \to Q \times \Gamma \times
  \{L,R\}$ a transition function, and $q_0,q_f \in Q$ the
  initial and accepting states, respectively. 
  Let $\vec{a}=a_1\dots a_n$ be an input for $M$. We construct a KB
   that is unsatisfiable iff $M$ accepts
  $\vec{a}$. This will prove  
  \PSpace-hardness.
  The KB uses a single object name
  $d$ and the following concepts, for $a \in \Gamma$, $q \in Q$
  and $1 \leq i \leq s(n)$, representing configurations of $M$:
  \begin{akrzitemize}
  \item $H_{iq}$, which contains $d$ if the head points to cell $i$ and the current
    state is $q$;
  \item $S_{ia}$, which contains $d$ if tape cell $i$ contains symbol $a$ in the current configuration;
  \item $D_i$, which contains $d$ if the head pointed to 
    cell $i$ in the \emph{previous} configuration.
\end{akrzitemize}
Let $\T_M$ contain the
following concept inclusions, for $a,a'\in \Gamma$, $q,q'\in Q$,
and $1 \leq i,j \leq s(n)$: 
\begin{align}
\label{eq:tm:change:head:R} H_{iq} & \sqsubseteq \bot \U
  H_{{(i+1)}q'},\quad H_{iq} \sqsubseteq \bot \U S_{ia'}, &&
  \text{if } \delta(q,a) = (q',a',R) \text{ and } i < s(n),\\
  \label{eq:tm:change:head:L} H_{iq} & \sqsubseteq \bot \U
  H_{{(i-1)}q'},\quad H_{iq} \sqsubseteq \bot \U S_{ia'}, && \text{if } \delta(q,a) = (q',a',L) \text{ and }
  i > 1,\\
  \label{eq:tm:E} H_{iq} & \sqsubseteq \bot \U D_i,\\
 \label{eq:tm:neq:D} D_i \sqcap  D_j & \sqsubseteq \bot,&& \text{if } i \neq j,\\
  \label{eq:tm:preserve} S_{ia} & \sqsubseteq S_{ia} \U  D_i,\\
  \label{eq:tm:non-term} H_{iq_f} & \sqsubseteq \bot,
\end{align}
and let $\A_{\vec{a}}$ consist of the following ABox assertions:
\begin{equation*}
  H_{1q_0}(d),
  \qquad S_{ia_i}(d), \ \text{ for } 1 \leq i \leq n,\qquad
  S_{i\#}(d), \ \text{ for } n < i \leq s(n).
\end{equation*}
Note that the concept inclusions in $\T_M$ are of the
form $B_1 \sqcap B_2 \sqsubseteq \bot$ or $B_1 \sqsubseteq B_2 \U
B_3$, where each $B_i$ is either a concept name or $\bot$, and that $\U$, in essence, encodes the `next-time' operator.
The proof that $(\T_M,\A_{\vec{a}})$ is unsatisfiable
iff $M$ accepts $\vec{a}$ is given in Appendix~\ref{app:PSpace}. 
\end{proof}

On the other hand, if we do not have the $\U/\S$ or $\Rnext/\Lnext$
operators in our languages, then the complexity drops to \NP, which matches the complexity of the $\Rbox/\Lbox$-fragment of propositional temporal logic~\cite{OnoNakamura80}:
\begin{theorem}\label{thm:np1}
Satisfiability of \TdDLbn{} and \zSVDLb{} KBs is \NP-complete.
\end{theorem}
\begin{proof}
  The lower bound is immediate from the complexity of \DLbn{}.
  The upper bound for \TdDLbn{} can be shown using a slight
  modification of the reduction $\cdot^\ddagger$ and the
  result of~\citeN{OnoNakamura80} mentioned above. We need to modify
  $\cdot^\ddagger$ in such a way that the target language does not
  contain the $\nxt^n$ operators of the ABox. We take a fresh predicate
  $H_C^n(x)$ for each ground atom $\nxt^n C(a)$ occurring in
  $\A^\dagger$ and use the following formulas instead of $\nxt^n C(a)$
  in $\A^\dagger$:
\begin{align*}
&\bigl(\Rdiamond^n H^n_C(a) \land \neg\Rdiamond^{n+1} H^n_C(a)\bigr) \land \SVbox\,\bigl(H^n_C(a) \rightarrow C(a)\bigr), & \text{if } n \geq 0,\\
&\bigl(\Ldiamond^{-n} H^n_C(a) \land \neg\Ldiamond^{-n+1} H^n_C(a)\bigr) \land \SVbox\,\bigl(H^n_C(a) \rightarrow C(a)\bigr), & \text{if } n < 0,
\end{align*}
where $\Rdiamond^n$ and $\Ldiamond^n$ denote $n$ applications of
$\Rdiamond$ and $\Ldiamond$, respectively.
Note that $H^n_C(a)$ holds at $m$ iff $m=n$. Thus, we use these predicates to `mark' a small number of moments of time in models.

The \NP{} upper bound trivially holds for \zSVDLb, a sublanguage
of \TdDLbn.
\end{proof}

Our next theorem also uses the reduction $\cdot^\ddagger$
and follows from the complexity results for the fragments $\kromLTL(\SVbox,\Rnext/\Lnext,\Rbox/\Lbox)$ and $\coreLTL(\SVbox,\Rnext/\Lnext)$ of \PTL{}, obtained by restricting the form of clauses in the Separated Normal Form (SNF)~\cite{DBLP:conf/ijcai/Fisher91} and proved in Section~\ref{sec:2ltl}.
\begin{theorem}\label{thm:krom-core}
  Satisfiability of \TdxDLkn{}, \TdDLkn{} and \TdxDLcn{} KBs KBs is \NP-complete.
\end{theorem}
\begin{proof}
  The \NP{} upper bound follows from the fact that the $\cdot^\ddagger$   translation of a KB in any of the three languages is a $\kromLTL(\SVbox,\Rnext/\Lnext,\Rbox/\Lbox)$-formula. By Theorem~\ref{lem:bin-ltl:core-diamond}, satisfiability of such formulas is in \NP.
The matching lower bound for \TdDLkn{} (and \TdxDLkn{}) follows from the proof of \NP-hardness of $\kromLTL(\SVbox,\Lbox/\Rbox)$, which will be presented in Theorem~\ref{krom-low-NP}: one can take concept names instead of propositional variables and encode, in an obvious way, the formulas of the proof of Theorem~\ref{krom-low-NP} in a KB;  similarly, the lower bound for \TdxDLcn{} follows from \NP-hardness of $\coreLTL(\SVbox,\Rnext/\Lnext)$; see Theorem~\ref{lem:bin-ltl:core-diamond}.
\end{proof}

\begin{theorem}\label{thm:core}
 Satisfiability of \TdDLcn{} KBs is in \PTime.
\end{theorem}
\begin{proof}
The result follows from the observation that  the $\cdot^\ddagger$   translation of a \TdDLcn{}  KB is of the form $\varphi'\land\varphi''$, where $\varphi'$ is a  $\coreLTL(\SVbox,\Rbox/\Lbox)$-formula representing the TBox and $\varphi''$ is a conjunction of formulas of the form $\nxt^n p$, for propositional variables $p$. A modification of the proof of Theorem~\ref{thm:hornLTL} (explained in  Remark~\ref{rem:core:extension}) shows that satisfiability of such formulas is in \PTime{}.

We note in passing that the matching lower bound for $\coreLTL(\SVbox,\Rbox/\Lbox)$, to be proved in Theorem~\ref{core-low-P}, does not imply \PTime{}-hardness of \TdDLcn{} as the formulas in the proof require an implication to hold at the initial moment of time, which is not expressible in \TdDLcn{}.\end{proof}

Finally, we show that the Krom and core fragments of \zSVDLb{} can  be simulated by 2CNFs (see, e.g.,~\cite{BoGG97}), whose satisfiability is \NLogSpace-complete.
\begin{theorem}\label{thm:nlogspace}
The satisfiability problem for \zSVDLc{} and \zSVDLk{} KBs is
  \NLogSpace-complete.
\end{theorem}
\begin{proof}
The lower bound is trivial from \NLogSpace-hardness of \DLcn.
We show the matching upper bound. Given a \zSVDLk{} KB $\K=(\T,\A)$, we consider the $\QTLi$-formula $\K^\ddagger$, which, by Lemma~\ref{lemma:Kddagger}, is satisfiable iff $\K$ is satisfiable. 
Now, we transform $\K^\ddagger$ into a two-sorted first-order formula $\K^{\ddagger_2}$ by representing the time dimension explicitly as a predicate argument. Recall that $\K^\ddagger$ is built from the propositional variables $p_R$, for $R\in\role$, and unary predicates $B^*$, for concepts $B$ of the form $A$ and $\mathop{\geq q} R$. Without loss of generality, we assume that there is at most one $\SVbox$ in front of each $B^*$ in $\K^\ddagger$. We replace each $B^*(x)$ in $\K^\ddagger$ that is not prefixed by $\SVbox$ with the binary predicate $B^*(x,t)$, and each $\SVbox B^*(x)$ with a fresh unary predicate $U_B(x)$; the outermost $\SVbox$ is replaced by $\forall t$.  To preserve the semantics of the $\SVbox B^*$, we also append to the resulting formula the conjuncts $\forall x\, \bigl(U_B(x) \leftrightarrow  \forall t\, B^*(x,t)\bigr)$, which are equivalent to
\begin{equation*}
\forall t\,\forall x\,\bigl(U_B(x) \to B^*(x,t)\bigr)
  \ \ \land\ \ 
\forall x\,\exists t\,\bigl(B^*(x,t) \to U_B(x) \bigr). 
\end{equation*}
The propositional variables $p_R$ of $\K^\ddagger$ remain propositional variables in $\K^{\ddagger_2}$, and the second conjunct of $\K^\ddagger$ is replaced by the following formula:
\begin{equation*}
 \bigl(p_R \to (\exists \inv{R})^*(d_{\inv{R}}, 0) \bigr)
 \ \ \land\ \ 
   \forall t\,\forall x\,\bigl((\exists R)^*(x,t)
 \to p_R\bigr)
\end{equation*}
with constant $0$.  Finally,
the ground atoms $\nxt^n B^*(a)$ in $\A^\dag$ are replaced by $B^*(a,n)$ with constants $n$. Thus, 
$\K^{\ddagger_2}$ is a conjunction of (at most)
binary clauses without quantifiers or with prefixes of the form $\forall
t\,\forall x$ and $\forall x\,\exists t$. Since the first argument of
the predicates, $x$, is always universally quantified, $\K^{\ddagger_2}$ is
equisatisfiable with the conjunction $\K^{\ddagger_3}$ of the
formulas obtained by replacing $x$ in $\K^{\ddagger_2}$ with  the constants in the set $\ob\cup\{d_R\mid R\in\role\}$. But then $\K^{\ddagger_3}$ is equivalent to a first-order Krom formula in prenex form with the quantifier prefix $\exists^*\forall$, satisfiability of which can be checked in \NLogSpace{} (see e.g.,~\cite[Theorem 8.3.6]{BoGG97}).
\end{proof}

\section{Clausal Fragments of Propositional Temporal Logic}\label{sec:2ltl}

Our aim in this section is to introduce and investigate a number of
new fragments of the propositional temporal logic \PTL. One reason for this is to obtain the complexity results required for the proof of Theorems~\ref{thm:krom-core} and~\ref{thm:core}. We believe, however, that these fragments are of sufficient interest on their own, independently of temporal conceptual modelling and reasoning. 

\citeN{SistlaClarke82} showed  that full \PTL{} is \PSpace-complete; see also~\cite{HalpernR81,DBLP:conf/lop/LichtensteinPZ85,Rabi:10,Rey:10}. \citeN{OnoNakamura80} proved that for formulas with only $\Rbox$ and $\Rdiamond$ the satisfiability problem becomes \NP-complete. Since then a number of fragments of \PTL{} with lower computational complexity have been identified and studied.
\citeN{ChenLin93} observed that the complexity of \PTL{} does not change even if we restrict attention to temporal Horn formulas. \citeN{DBLP:journals/iandc/DemriS02} determined the
complexity of fragments that depend on three parameters: the available
temporal operators, the number of nested temporal operators, and the
number of propositional variables in formulas. \citeN{Markey04}
analysed fragments defined by the allowed set of temporal operators,
their nesting and the use of negation.  \citeN{DBLP:conf/ijcai/DixonFK07} introduced a XOR fragment of
\PTL{} and showed its tractability.
\citeN{DBLP:journals/corr/abs-0812-4848}
systematically investigated the complexity of fragments given by both
temporal operators and Boolean connectives (using Post's lattice of
sets of Boolean functions).

In this section, we classify temporal formulas according to their clausal normal form. We remind the reader that 
any \PTL-formula can be transformed to an equisatisfiable formula in  \emph{Separated Normal Form}
(SNF)~\cite{DBLP:conf/ijcai/Fisher91}. A formula in SNF is a
conjunction of \emph{initial clauses} (that define `initial
conditions' at moment 0), \emph{step clauses} (that define `transitions' between consecutive states), and \emph{eventuality clauses} (ensuring that
certain states are eventually reached). More precisely, for the time flow $\Z$, a formula in SNF is a conjunction of clauses of the form
\begin{align*}
& L_1 \lor \dots \lor L_k,\\
& \SVbox \bigl( (L_1 \land \dots \land L_k) \to \nxt (L'_1 \lor \dots \lor L'_m)\bigr),\\
& \SVbox \bigl( (L_1 \land \dots \land L_k) \to \Rdiamond L\bigr),\\
& \SVbox \bigl( (L_1 \land \dots \land L_k) \to \Ldiamond L\bigr),
\end{align*}
where $L,L_1,\dots,L_k,L'_1,\dots,L'_m$ are \emph{literals}---i.e.,
propositional variables or their negations---and $\nxt$ is a
short-hand for $\Rnext$ (we will use this abbreviation throughout this
section). By definition, we assume the empty disjunction to be
$\bot$ and the empty conjunction to be $\top$. For example, the second clause with $m =
0$ reads $\SVbox (L_1\land \dots
\land L_k \to \bot)$.

The transformation to SNF is achieved by fixed-point unfolding and renaming~\cite{DBLP:journals/tocl/FisherDP01,Plaisted86}. Recall that an occurrence of a subformula is said to be \emph{positive} if it is in the scope of an even number of negations. Now, as $p \U q$ is equivalent to $\nxt q \lor \bigl(\nxt p \land \nxt (p \U q)\bigr)$, every \emph{positive occurrence} of $p \U q$ in a given formula $\varphi$ can be replaced by a fresh propositional variable $r$, with the following three clauses added as conjuncts to $\varphi$:
\begin{equation*}
\SVbox\bigl(r \to \nxt (q \lor p)\bigr), \qquad \SVbox\bigl(r \to \nxt (q \lor r)\bigr) \quad\text{and}\quad \SVbox\bigl(r \to \Rdiamond q\bigr).
\end{equation*}
The result is equisatisfiable with $\varphi$ but does not contain positive occurrences of $p \U q$. Similarly, we can get rid of other temporal operators and transform the formula to SNF~\cite{DBLP:journals/tocl/FisherDP01}.

We now define four types of fragments of \PTL, which are called $\coreLTL(\mathcal{X})$, $\kromLTL(\mathcal{X})$, $\hornLTL(\mathcal{X})$ and $\boolLTL(\mathcal{X})$, where $\mathcal{X}$ has one of the following four forms: $\SVbox,\Rnext/\Lnext,\Rbox/\Lbox$, or $\SVbox,\Rnext/\Lnext$, or $\SVbox,\Rbox/\Lbox$ or $\SVbox$. 
$\coreLTL(\mathcal{X})$-\emph{formulas}, $\varphi$, are constructed using the grammar:
\begin{align*}
\varphi \ \ &::= \ \ \psi \ \ \mid \ \ 
\SVbox \psi  \ \ \mid \ \ 
 \varphi_1 \land \varphi_2,\\
\psi \ \ &::= \ \ \lambda_1 \to \lambda_2 \ \ \mid \ \ \lambda_1 \land \lambda_2 \to \bot, \tag{\textit{core}}\\
\lambda \ \ &::= \ \ \bot \ \ \mid \ \ p \ \ \mid \ \  \bigstar \lambda, \quad \text{where $\bigstar$ is one of the operators in $\mathcal{X}$}.
\end{align*}
Definitions of the remaining three fragments differ only in the shape of $\psi$. In $\kromLTL(\mathcal{X})$-\emph{formulas}, $\psi$ is a binary clause:
\begin{align*}
\psi \ \ &::= \ \ \lambda_1 \to \lambda_2 \ \ \mid \ \ \lambda_1 \land \lambda_2 \to \bot \ \ \mid \ \ \lambda_1 \lor \lambda_2. \tag{\textit{krom}} 
\end{align*}
In $\hornLTL(\mathcal{X})$-\emph{formulas}, $\psi$ is a Horn clause:
\begin{align*}
\psi \ \ &::= \ \ \lambda_1 \land \dots \land \lambda_n \to \lambda, \tag{\textit{horn}} 
\end{align*}
while in $\boolLTL(\mathcal{X})$-\emph{formulas}, $\psi$ is an arbitrary  clause:
\begin{align*}
\psi \ \ &::= \ \ \lambda_1 \land \dots \land \lambda_n \to \lambda'_1 \lor \dots \lor \lambda'_k. \tag{\textit{bool}} 
\end{align*}
Note that, if $\mathcal{X}$ contains $\Box$-operators then the corresponding $\Diamond$-operators can be defined in the fragments $\kromLTL(\mathcal{X})$ and $\boolLTL(\mathcal{X})$. 

\begin{table}[t]
\tbl{Complexity of Clausal Fragments of \PTL.
\label{table:PTL:languages}}{
\centering\renewcommand{\arraystretch}{1.4}
\begin{tabular}{|c|c|c|c|c|}\hline
 & $\SVbox$, $\Rbox/\Lbox$, $\raisebox{2pt}{\tiny$\bigcirc$}\hspace{-0.15em}_{\scriptscriptstyle F}/\raisebox{2pt}{\tiny$\bigcirc$}\hspace{-0.15em}_{\scriptscriptstyle P}$ &$\SVbox$, $\raisebox{2pt}{\tiny$\bigcirc$}\hspace{-0.15em}_{\scriptscriptstyle F}/\raisebox{2pt}{\tiny$\bigcirc$}\hspace{-0.15em}_{\scriptscriptstyle P}$ & $\SVbox$, $\Rbox/\Lbox$ & $\SVbox$  \\\hline 
Bool & \PSpace{}  &\PSpace & \NP{}  & \NP \\\hline
Horn & \PSpace  &\PSpace{}  & \PTime{} {\scriptsize $[\leq$ Th.~\ref{thm:hornLTL}$]$} & \PTime \\\hline
Krom & \NP{} {\scriptsize $[\leq$ Th.~\ref{lem:bin-ltl:krom-diamond-next-np}$]$} & \NP & \NP{} $[\geq$ Th.~\ref{krom-low-NP}$]$ & \NLogSpace \\\hline
core & \NP & \NP{} {\scriptsize $[\geq$ Th.~\ref{lem:bin-ltl:core-diamond}$]$} & \PTime{} {\scriptsize $[\geq$Th.~\ref{core-low-P}$]$} & \NLogSpace \\\hline
\end{tabular}
}\label{tab:PTL-fragments}
\end{table}
Table~\ref{tab:PTL-fragments} shows how the complexity of the
satisfiability problem for \PTL-formulas depends on the type of the
underlying propositional clauses and the available temporal operators.
The \PSpace{} upper bound is
well-known; the
matching lower bound can be obtained by a standard encoding of
deterministic Turing machines with polynomial tape
(cf.~Theorem~\ref{thm:pspace}).  The \NP{} upper bound for
$\boolLTL(\SVbox,\Rbox/\Lbox)$ follows from~\cite{OnoNakamura80}.  The
\NLogSpace{} lower bound is trivial and the matching upper bound
follows from the complexity of the Krom formulas with the quantifier
prefix of the form $\exists^*\forall$~\cite{BoGG97} (a similar
argument is used in Theorem~\ref{thm:nlogspace}). In the remainder of
this section, we prove all other results in this table.  It is worth
noting how the addition of $\nxt$ or $\neg$ increases the complexity
of $\hornLTL(\SVbox,\Rbox/\Lbox)$ and $\coreLTL(\SVbox,\Rbox/\Lbox)$.

\begin{theorem}\label{lem:bin-ltl:krom-diamond-next-np}
The satisfiability problem for $\kromLTL(\SVbox,\Rnext/\Lnext,\Rbox/\Lbox)$-formulas is in \NP.
\end{theorem}
\begin{proof}
We proceed as follows. First, in Lemma~\ref{l:structure}, we give a satisfiability criterion for \PTL-formulas in terms of types---sets of propositions that occur in the given formula---and distances between them in temporal models. The number of types required is polynomial in the size of the given formula; the distances, however, are exponential, and although they can be represented in binary (in polynomial space), in general there is no polynomial algorithm that checks whether two adjacent types can be placed at a given distance (unless $\PTime = \PSpace$).  In the remainder of the proof, we show that, for formulas with \emph{binary clauses}, this condition can be verified by constructing a polynomial number of polynomial arithmetic progressions (using unary automata). This results in a non-deterministic polynomial-time algorithm: guess types and distances between them, and then verify (in polynomial time) whether the types can be placed at the required distances.  

Let $\varphi'$ be a $\kromLTL(\SVbox,\Rnext/\Lnext,\Rbox/\Lbox)$-formula. By introducing fresh propositional variables if required, we can transform $\varphi'$ (in polynomial time) to a formula 
\begin{equation}\label{eq:kromltl:input}
\varphi \ \ \ = \ \ \ \Psi\ \ \land \ \ \SVbox \Phi,
\end{equation}
where $\Psi$ contains no temporal operators and $\Phi$ contains no nested occurrences of temporal operators. Indeed, if $\varphi'$ contains a conjunct $\psi$ with a temporal $\lambda$, then we take a fresh propositional variable $\overline{\lambda}$, replace $\lambda$ in $\psi$ with $\overline{\lambda}$, and add to $\varphi'$ a new conjunct $\SVbox (\overline{\lambda} \leftrightarrow \lambda)$. In a similar way we get rid of nested occurrences of temporal operators in $\Phi$.

We will not distinguish between a set of formulas and the conjunction of its elements, and  write $\SVbox \Phi$ for $\bigwedge_{\chi\in\Phi}\SVbox \chi$. As $\SVbox (\lambda \lor \Lnext\lambda')$ is  equivalent to $\SVbox(\Rnext\lambda \lor \lambda')$, we can assume that $\Phi$ does not contain $\Lnext$ (remember that we agreed to denote $\Rnext$ by $\nxt$). 
We regard $\SVbox$ inside $\Phi$ as  defined by $\SVbox\lambda = \Rbox\Lbox\lambda$. Thus, we assume that $\Phi$ contains only $\nxt$, $\Lbox$ and $\Rbox$ (which are not nested).  

We first characterise the structure of models for formulas of the form~\eqref{eq:kromltl:input} (with $\Psi$ and $\Phi$ satisfying those conditions). It should be noted that this structure only depends on $\varphi$ being of that form (cf.~\cite{Gabbayetal94,GKWZ03} and references therein) and does not depend on whether or not $\Psi$ and $\Phi$ are sets of binary clauses.
To this end, for each $\Rbox L$ in $\varphi$, we take a fresh propositional variable, denoted $\overline{\Rbox L}$ and called the \emph{surrogate} of $\Rbox L$; likewise, for each $\Lbox L$ we take 
its surrogate $\overline{\Lbox L}$. Let $\overline{\Phi}$ be the result of replacing $\Box$-subformulas in $\Phi$ by their surrogates. 
It should be clear that $\varphi$ is equisatisfiable with
\begin{equation*}
\overline{\varphi} \ \ \ =  \ \ \ \Psi\land\SVbox\overline{\Phi}\ \ \ \ \land \ \bigwedge_{\Rbox L \text{ occurs in } \Phi} \hspace*{-1.5em}\SVbox(\overline{\Rbox L} \leftrightarrow \Rbox L) \ \ \ \ \land\ \bigwedge_{\Lbox L \text{ occurs in } \Phi} \hspace*{-1.5em}\SVbox(\overline{\Lbox L} \leftrightarrow \Lbox L).
\end{equation*}
By a \emph{type} for $\overline{\varphi}$ we mean any set of literals that contains either $p$ or $\neg p$, for each variable $p$ in $\overline{\varphi}$ (including the surrogates $\overline{\Rbox L}$ and $\overline{\Lbox L}$).

\begin{lemma}\label{l:structure}
The formula $\overline{\varphi}$ is satisfiable iff there exist $k+5$ integers
\begin{equation*}
m_0 < m_1 < \dots < m_{k+4}
\end{equation*}
\textup{(}where $k$ does not exceed the number of $\Rbox L$ and $\Lbox L$\textup{)} and a sequence $\Psi_0, \Psi_1, \dots, \Psi_{k+4}$ of types for $\overline{\varphi}$ satisfying the following conditions \textup{(}see Fig.~\ref{b6}\textup{)}\textup{:}
\begin{akrzlist}\itemsep=6pt
\item[B$_0$] $m_{i+1} - m_i < 2^{|\overline{\varphi}|}$, for $0 \leq i < k+4$\textup{;}
\item[B$_1$] there exists $
\ell_0$ such that $0 \leq \ell_0 \leq k+4$ and  $\Psi \land \Psi_{\ell_0}$ is consistent\textup{;}
\item[B$_2$] for each $i$, $0 \leq i < k+4$, and each $\Rbox L$ in $\Phi$,
\begin{equation*}
\text{if } \overline{\Rbox L}\in \Psi_i  \text{ then } L,\overline{\Rbox L}\in \Psi_{i+1},\qquad\text{ and }\qquad
\text{if }\overline{\Rbox L} \in \Psi_{i+1} \setminus \Psi_i \text{ then } L\notin \Psi_{i+1}\textup{;}
\end{equation*}
\item[B$_3$] there exists $\ell_F < k+4$ such that
  $\Psi_{\ell_F}=\Psi_{k+4}$ and, for each $\Rbox L$ in $\Phi$,
\begin{equation*}
\text{if } \overline{\Rbox L} \notin \Psi_{\ell_F} \text{ then } L \notin
    \Psi_j, \text{ for some } j \geq \ell_F\textup{;}
\end{equation*}
\item[B$_4$] for each $i$, $0 < i \leq k+4$, and each $\Lbox L$ in $\Phi$,
\begin{equation*}
\text{if } \overline{\Lbox L}\in \Psi_i \text{ then } L,\overline{\Lbox L}\in \Psi_{i-1},\qquad\text{ and }\qquad
\text{if } \overline{\Lbox L} \in \Psi_{i-1} \setminus
  \Psi_i \text{ then } L\notin \Psi_{i-1}\textup{;} 
\end{equation*}
\item[B$_5$] there exists $\ell_P > 0$ such that $\Psi_{\ell_P}=\Psi_0$ and, for each $\Lbox L$ in $\Phi$,
\begin{equation*}
\text{if } \overline{\Lbox L} \notin \Psi_{\ell_P} \text{ then } L \notin
    \Psi_j, \text{ for some } j \leq \ell_P\textup{;}
\end{equation*}
\item[B$_6$]
for all $i$, $0 \leq i < k+4$, the following formula is consistent\textup{:}
\begin{equation}\label{eq:kromltl:interval}
  \Psi_i \ \ \ \land \bigwedge_{j = 1}^{m_{i+1} - m_i - 1} \hspace*{-1.5em}\nxt^j 
  \Theta_i \ \ \ \ \land \ \ \ \nxt^{m_{i+1} - m_i} \Psi_{i+1} \ \ \ \ \land \ \ \ \ \SVbox \overline{\Phi},
\end{equation}
where $\nxt^j \Psi$ is the result of attaching $j$ operators $\nxt$ to each literal in $\Psi$ and
\begin{multline*}
\Theta_i \ \ = \ \
 \{ L, \overline{\Rbox L} \mid \overline{\Rbox L}\in \Psi_i \} \cup
\{ \neg\overline{\Rbox L} \mid \overline{\Rbox L}\notin \Psi_i \} \cup {} \\
 \{ L, \overline{\Lbox L} \mid \overline{\Lbox L}\in\Psi_{i+1} \} \cup
\{ \neg\overline{\Lbox L} \mid \overline{\Lbox L}\notin \Psi_{i+1} \}.
\end{multline*}
\end{akrzlist}
\end{lemma}
\begin{figure}[t]
\centering
\begin{tikzpicture}[>=latex,point/.style={circle,draw=black,minimum size=1mm,inner sep=0pt}]\footnotesize
\begin{scope}[xscale=0.7]
\draw[->](0,0) -- (13,0);
\node at (1.5,0) [point,fill=black,label=above:{$\neg\Rbox L$}] {};
\node at (3,0) [point,fill=black,label=above:{$\neg\Rbox L$}] {};
\node at (4,0) [label=above:{$\ldots$}] {};
\node at (5,0) [point,fill=black,label=above:{$\neg\Rbox L$}] {};
\node at (6.5,0) [point,fill=black,label=above:{$\Rbox L$},label=below:{$\neg L$}] {};
\node at (8,0) [point,fill=black,label=above:{$\Rbox L$},label=below:{$L$}] {};
\node at (9,0) [label=above:{$\ldots$}] {};
\node at (10,0) [point,fill=black,label=above:{$\Rbox L$},label=below:{$L$}] {};
\node at (11.5,0) [point,fill=black,label=above:{$\Rbox L$},label=below:{$L$}] {};
\draw[densely dotted] (0.75,-0.5) -- ++(0,1.5);
\draw[densely dotted] (2.25,-0.5) -- ++(0,1.5);
\draw[densely dotted] (5.75,-0.5) -- ++(0,1.5);
\draw[densely dotted] (7.25,-0.5) -- ++(0,1.5);
\draw[densely dotted] (10.75,-0.5) -- ++(0,1.5);
\draw[densely dotted] (12.25,-0.5) -- ++(0,1.5);
\node at (1.5,0.7) {$\Psi_{i-1}$};
\node at (6.5,0.7) {$\Psi_i$};
\node at (9,0.7) {$\Theta_i$};
\node at (11.5,0.7) {$\Psi_{i+1}$};
\end{scope}
\begin{scope}[xscale=0.7, yshift=-20mm, xshift=50mm]
\draw[->](0,0) -- (13,0);
\node at (1.5,0) [point,fill=black,label=above:{$\Lbox L$},label=below:{$L$}] {};
\node at (3,0) [point,fill=black,label=above:{$\Lbox L$},label=below:{$L$}] {};
\node at (4,0) [label=above:{$\ldots$}] {};
\node at (5,0) [point,fill=black,label=above:{$\Lbox L$},label=below:{$L$}] {};
\node at (6.5,0) [point,fill=black,label=above:{$\Lbox L$},label=below:{$\neg L$}] {};
\node at (8,0) [point,fill=black,label=above:{$\neg\Lbox L$}] {};
\node at (9,0) [label=above:{$\ldots$}] {};
\node at (10,0) [point,fill=black,label=above:{$\neg\Lbox L$}] {};
\node at (11.5,0) [point,fill=black,label=above:{$\neg\Lbox L$}] {};
\draw[densely dotted] (0.75,-0.5) -- ++(0,1.5);
\draw[densely dotted] (2.25,-0.5) -- ++(0,1.5);
\draw[densely dotted] (5.75,-0.5) -- ++(0,1.5);
\draw[densely dotted] (7.25,-0.5) -- ++(0,1.5);
\draw[densely dotted] (10.75,-0.5) -- ++(0,1.5);
\draw[densely dotted] (12.25,-0.5) -- ++(0,1.5);
\node at (1.5,0.7) {$\Psi_i$};
\node at (4,0.7) {$\Theta_i$};
\node at (6.5,0.7) {$\Psi_{i+1}$};
\node at (11.5,0.7) {$\Psi_{i+2}$};
\end{scope}
\end{tikzpicture}
\caption{Conditions \textbf{(B$_2$)}, \textbf{(B$_4$)} and \textbf{(B$_6$)} in Lemma~\ref{l:structure}.}
\label{b6}
\end{figure}
\begin{proof}
  ($\Rightarrow$) Let $\Mmf,0\models\overline{\varphi}$. Denote by
  $\Psi(m)$ the type for $\overline{\varphi}$ containing all literals that
  hold at $m$ in $\Mmf$. As the number of types is
  finite, there is $m_{\scriptscriptstyle F} > 0$ such that each type in the sequence
  $\Psi(m_{\scriptscriptstyle F}),\Psi(m_{\scriptscriptstyle F}+1),\dots$ appears infinitely
  often; similarly, there is $m_{\scriptscriptstyle P} < 0$ such that each type in the
  sequence $\Psi(m_{\scriptscriptstyle P}),\Psi(m_{\scriptscriptstyle P}-1),\dots$ appears
  infinitely often.  Then, for each subformula $\Rbox L$ of $\Phi$,
  we have one of the three options: (1) $L$ is always true in $\Mmf$, in which case we set
  $m_{\Box_{\!F} L} = 0$; (2)~there is $m_{\Box_F L}$ such that $\Mmf,m_{\Box_{\!F}
    L}\models \neg L \land \Rbox L$, in which case $m_{\scriptscriptstyle P} < m_{\Box_{\!F} L}
  < m_{\scriptscriptstyle F}$; or (3) $\Rbox L$ is always false in $\Mmf$, in which
  case $L$ is false infinitely often after $m_{\scriptscriptstyle F}$, and so
  there is $m_{\Box_{\!F} L}\geq m_{\scriptscriptstyle F}$ such that $\Mmf,m_{\Box_{\!F}
    L}\models \neg L$. Symmetrically, for each subformula $\Lbox L$ of
  $\Phi$, (1) $L$ is always true in $\Mmf$, in which case we set
  $m_{\Box_{\!P} L} = 0$; or (2) there is an $m_{\Box_{\!P} L}$ such that $m_{\scriptscriptstyle P} <
  m_{\Box_{\!P} L} < m_{\scriptscriptstyle F}$ and $\Mmf,m_{\Box_{\!P} L}\models \neg L
  \land \Lbox L$; or (3) $\Lbox L$ is always false in $\Mmf$, in
  which case there is $m_{\Box_{\!P} L} \leq m_{\scriptscriptstyle P}$ such that
  $\Mmf,m_{\Box_{\!P} L}\models \neg L$. Let
  $m_1<m_2<\dots<m_{k+3}$ be an enumeration of the set
\begin{equation*}
\{0,m_{\scriptscriptstyle P},m_{\scriptscriptstyle F}\} \ \ \cup\ \ \{ m_{\Box_{\!F} L} \mid \Rbox L \text{ occurs in } \Phi \} \ \ \cup\ \ \{ m_{\Box_{\!P} L} \mid \Lbox L \text{ occurs in } \Phi\}.
\end{equation*}
Let $m_{k+4} > m_{k+3}$ be such that $\Psi(m_{k+4}) =
\Psi(m_{\scriptscriptstyle F})$ and let $m_0 < m_1$ be such that $\Psi(m_0) =
\Psi(m_{\scriptscriptstyle P})$. We then set $\Psi_i = \Psi(m_i)$, for $0
\leq i \leq k+4$.  Let $\ell_0$, $\ell_P$ and $\ell_F$ be such that
$m_{\ell_0} = 0$, $m_{\ell_P} = m_{\scriptscriptstyle P}$ and $m_{\ell_F} = m_{\scriptscriptstyle F}$. It should
be clear that~\textbf{(B$_1$)}--\textbf{(B$_6$)} hold. Finally, given
a model of $\overline{\varphi}$ with two moments $m$ and $n$ such that the types
at $m$ and $n$ coincide, we can construct a new model for $\overline{\varphi}$ by
`removing' the states $i$ with $m \leq i <
n$. Since the number of distinct types is bounded by $2^{|\overline{\varphi}|}$,
by repeated applications of this construction we can further
ensure~\textbf{(B$_0$)}.

\smallskip

($\Leftarrow$) We construct a model $\Mmf$ of $\overline{\varphi}$
by taking finite cuts of the models $\Mmf_i$ of the formulas
in~\textbf{(B$_6$)}: between the moments $m_0$ and $m_{k+4}$, the
model $\Mmf$ coincides with the models
$\Mmf_0,\dots,\Mmf_{k+3}$ so that at the moment $m_i$
in $\Mmf$ we align the moment 0 of $\Mmf_i$, and at the
moment $m_{i+1}$ we align the moment $m_{i+1}-m_i$ of
$\Mmf_i$, which coincides with the moment 0 of
$\Mmf_{i+1}$ because both are defined by $\Psi_{i+1}$; before
the moment $m_0$, the model $\Mmf$ repeats infinitely often  its own
fragment between $m_0$ and $m_{\ell_P}$, and after $m_{k+4}$ it
repeats infinitely often its fragment between $m_{\ell_F}$ and $m_{k+4}$
(both fragments contain more than one
state). It is readily seen that
$\Mmf,m_{\ell_0}\models\overline{\varphi}$.
\end{proof}

By Lemma~\ref{l:structure}, if we provide a polynomial-time algorithm
for verifying~\textbf{(B$_6$)}, we can check satisfiability of
$\overline{\varphi}$ in \NP. Indeed, it suffices to guess $k+5$ types
for $\overline{\varphi}$ and $k+4$ natural numbers $n_i = m_{i+1} -
m_i$, for $0 \leq i < k+4$, whose binary representation,
by~\textbf{(B$_0$)}, is polynomial in $|\overline{\varphi}|$. It is
easy to see that \textbf{(B$_1$)}--\textbf{(B$_5$)} can be checked in
polynomial time. We show now that~\textbf{(B$_6$)} can also be
verified in polynomial time for
$\kromLTL(\SVbox,\Rnext/\Lnext,\Rbox/\Lbox)$-formulas.

Our problem is as follows: given a number $n \geq 0$ (in binary), types $\Psi$ and $\Psi'$, a set $\Theta$ 
of literals and a set $\Phi$ of binary clauses of the form
$D_1 \lor D_2$, where the $D_i$ are \emph{temporal} literals $p$,
$\neg p$, $\nxt p$ or $\neg\nxt p$, decide whether there is a model
satisfying
\begin{equation}\label{eq:reachability}
\Psi  \ \ \ \land \ \ \bigwedge_{k = 1}^{n-1} \nxt^k \Theta \ \ \ \land \ \ \nxt^n \Psi' \ \ \ \ \land \SVbox\Phi.
\end{equation}
In what follows, we write $\psi_1 \models \psi_2$ as a shorthand for `in every $\Mmf$, if $\Mmf,0 \models \psi_1$ then $\Mmf,0 \models \psi_2$.' For $0 \leq k \leq n$, we set: 
\begin{align*}
F^k_{\Phi}(\Psi) & = \bigl\{ L' \mid L\land \SVbox \Phi \models \nxt^k L', \text{ for } L\in\Psi \bigr\},\\
P^k_{\Phi}(\Psi') & = \bigl\{ L \mid \nxt^k L'\land \SVbox \Phi \models L, \text{ for } L'\in\Psi' \bigr\}.
\end{align*}

\begin{lemma}
Formula~\eqref{eq:reachability} is satisfiable iff the following conditions hold\textup{:}
\begin{akrzlist}
\item[\textup{L$_1$}] 
$F^0_{\Phi}(\Psi)\subseteq\Psi$, $F^n_{\Phi}(\Psi)\subseteq\Psi'$ and $P^0_{\Phi}(\Psi')\subseteq\Psi'$, $P^n_{\Phi}(\Psi')\subseteq\Psi$\textup{;}
\item[\textup{L$_2$}] 
 $\neg L\notin F^k_{\Phi}(\Psi)$ and $\neg L\notin P^{n-k}_{\Phi}(\Psi')$,  for all $L\in\Theta$ and $0 < k < n$.
\end{akrzlist}
\end{lemma}
\begin{proof}
It should be clear that if~\eqref{eq:reachability} is satisfiable
then the above conditions hold.  For the converse direction,  observe that if $L'\in F^k_\Phi(\Psi)$ then, since $\Phi$ is a set of binary clauses, there
  is a sequence of $\nxt$-prefixed literals $\nxt^{k_0} L_0 \leadsto \nxt^{k_1} L_1 \leadsto \dots \leadsto
  \nxt^{k_m} L_m$ such that $k_0 = 0$, $L_0\in\Psi$, $k_m = k$, $L_m =
  L'$, each $k_i$ is between $0$ and $n$ and the $\leadsto$ relation is defined by taking $\nxt^{k_i} L_i \leadsto \nxt^{k_{i+1}} L_{i+1}$ just in one of the three cases: $k_{i+1} = k_i$ and $L_i \to L_{i+1}\in\Phi$ or $k_{i+1} = k_i + 1$ and $L_i \to \nxt L_{i+1}\in\Phi$ or $k_{i+1} = k_i - 1$ and $\nxt L_i \to L_{i+1}\in\Phi$ (we assume that, for example, $\neg q \to \neg p\in\Phi$ whenever $\Phi$ contains $p \to q$).  
So, suppose
  conditions~\textbf{(L$_1$)}--\textbf{(L$_2$)} hold. We construct an
  interpretation
  satisfying~\eqref{eq:reachability}. By~\textbf{(L$_1$)}, both
  $\Psi\land\SVbox\Phi$ and $\nxt^n \Psi'\land\SVbox\Phi$ are
  consistent. So, let $\Mmf_\Psi$ and $\Mmf_{\Psi'}$ be such that
  $\Mmf_\Psi,0\models\Psi\land\SVbox \Psi$ and $\Mmf_\Psi,n\models\Psi'\land\SVbox \Psi$, respectively.
  Let $\Mmf$ be an interpretation that coincides with $\Mmf_\Psi$ for all
  moments $k \leq 0$ and with $\Mmf_{\Psi'}$ for all $k \geq n$;
  for the remaining $k$, $0 < k < n$, it is defined as follows. First,
  for each $p\in \Theta$ , we make $p$ true at $k$ and, for each $\neg
  p\in\Theta$, we make $p$ false at $k$; such an assignment exists due
  to \textbf{(L$_2$)}. Second, we extend the assignment by making $L$
  true at $k$ if $L\in F^k_{\Phi}(\Psi)\cup
  P^{n-k}_{\Phi}(\Psi')$. Observe that we have $\{p,\neg p\}\nsubseteq
  F^k_{\Phi}(\Psi)\cup P^{n-k}_{\Phi}(\Psi')$: for otherwise $L\land
  \SVbox\Phi\models \nxt^k p$ and $\nxt^{n-k} L' \land
  \SVbox\Phi\models \neg p$, for some $L\in\Psi$ and $L'\in\Psi'$, whence $L\land \SVbox\Phi\models \nxt^n
  \neg L'$, contrary to~\textbf{(L$_1$)}. Also, by~\textbf{(L$_2$)}, any
  assignment extension at this stage does not contradict the choices
  made due to $\Theta$.  Finally, all propositional variables not
  covered in the previous two cases get their values from $\Mmf_\Psi$
  (or $\Mmf_{\Psi'}$). We note that the last choice does not depend on the
  assignment that is fixed by taking account of the consequences of
  $\SVbox\Phi$ with $\Psi$, $\Psi'$ and $\Theta$ (because if the value
  of a variable depended on those sets of literals, the respective
  literal would be among the logical consequences and would have been
  fixed before).
\end{proof}

Thus, it suffices to show that conditions
\textbf{(L$_1$)}--\textbf{(L$_2$)} can be checked in polynomial time.
First, we claim that there is a polynomial-time algorithm which, given
a set $\Phi$ of binary clauses of the form $D_1\lor D_2$, 
constructs a set $\Phi^*$ of binary clauses that is `sound and
complete' in the following sense: 
\begin{akrzlist}
\item[S$_1$] $\SVbox\Phi^*\models \SVbox\Phi$;  
\item[S$_2$] if $ \SVbox\Phi\models \SVbox (L \to \nxt^k L_k)$ then
  either $k = 0$ and $L \to L_0\in\Phi^*$, or $k \geq 1$ and there are $L_0,L_1,\dots,L_{k-1}$  with $L = L_0$ and 
  $L_i\to\nxt L_{i+1}\in\Phi^*$, for $0 \leq i < k$.
\end{akrzlist}
Intuitively, the set $\Phi^*$ makes explicit the
consequences of $\SVbox \Phi$ and can be constructed in time $(2|\Phi|)^2$
(the number of temporal literals in $\Phi^*$ is bounded by the doubled length $|\Phi|$ of $\Phi$ as each of its literals  can only be prefixed by $\nxt$). Indeed, we
start from $\Phi$ and, at each step, add $D_1\lor D_2$ to $\Phi$ if
it contains both $D_1 \lor D$ and $\neg D\lor D_2$; we also add $L_1
\lor L_2$ if $\Phi$ contains $\nxt L_1 \lor \nxt L_2$ (and \emph{vice
versa}). This procedure is sound since we only add consequences of
$\SVbox\Phi$; completeness follows from the completeness proof for
temporal resolution~\cite[Section~6.3]{DBLP:journals/tocl/FisherDP01}.

Our next step is to encode $\Phi^*$ by means of unary automata.
For literals $L$ and $L'$,  consider a  nondeterministic finite automaton $\mathfrak{A}_{L,L'}$ over $\{0\}$, whose states are the literals of~$\Phi^*$, $L$ is the initial and $L'$ is the only accepting state, and the transition relation is $\{(L_1,L_2) \mid L_1 \to \nxt L_2\in\Phi^*\}$. By~\textbf{(S$_1$)} and~\textbf{(S$_2$)},  for all $k > 0$, we have
\begin{align*}
\mathfrak{A}_{L,L'} \text{ accepts } 0^k \quad\text{iff}\quad \SVbox\Phi\models \SVbox(L \to \nxt^k L'). 
\end{align*}
Then both $F^k_\Phi(\Psi)$ and $P^k_\Phi(\Psi')$, for $k > 0$, can be defined in terms of the
language of $\mathfrak{A}_{L,L'}$:
\begin{align*}
F_\Phi^k(\Psi) & = \bigl\{ L' \mid \mathfrak{A}_{L,L'} \text{ accepts } 0^k, \text{ for } L\in\Psi\bigr\},\\[-2pt]
P_\Phi^k(\Psi') & = \bigl\{ L \mid \mathfrak{A}_{\neg L,\neg L'} \text{ accepts } 0^k, \text{ for } L'\in\Psi'\bigr\}
\end{align*}
(recall that $\nxt^k L' \to L$ is equivalent to $\neg L \to \nxt^k \neg L'$). Note that the numbers $n$ and $k$ in conditions~\textbf{(L$_1$)} and~\textbf{(L$_2$)} are in general exponential in the length of $\Phi$ and, therefore, the automata $\mathfrak{A}_{L,L'}$ do not immediately provide a polynomial-time procedure for checking these conditions: although it can be shown that if~\textbf{(L$_2$)}  does not hold then it fails for a polynomial number $k$, this is not the case for \textbf{(L$_1$)}, which requires the accepting state to be reached in a fixed (exponential) number of transitions. Instead, we use the~\emph{Chrobak normal form}~\cite{chrobak-ufa} to decompose the automata into a polynomial number of polynomial-sized arithmetic progressions (which can have an exponential common period; cf.~the proof of Theorem~\ref{lem:bin-ltl:core-diamond}).

It is known that every $N$-state unary automaton
$\mathfrak{A}$ can be converted (in polynomial time) into an equivalent automaton in Chrobak normal form (e.g., by using Martinez's algorithm~\cite{to-ufa}), which has $O(N^2)$ states and gives rise to
$M$ arithmetic progressions $a_1 + b_1\N,\dots,a_M + b_M\N$, where
$a_i + b_i\N = \{ a_i + b_i m \mid m\in\N \}$, such that
\begin{akrzitemize}
\item $M\leq O(N^2)$ and $0 \leq a_i, b_i \leq N$, for $1 \leq i \leq
  M$;
\item $\mathfrak{A}$ accepts $0^k$ iff $k \in a_i + b_i\N$,
for some $1 \leq i \leq M$. 
\end{akrzitemize}
By construction,  the number of arithmetic progressions is quadratic 
in the length of $\Phi$.

We are now in a position to give a
polynomial-time algorithm for checking~\textbf{(L$_1$)} and~\textbf{(L$_2$)}, 
which requires solving
Diophantine equations.  In~\textbf{(L$_2$)}, for example, to verify that, for each $p\in\Theta$, we have $\neg p\notin F_\Phi^k(\Psi)$, for all $0 < k < n$, we take the automata $\mathfrak{A}_{L,\neg p}$, for $L\in\Psi$, and
transform them into the Chrobak normal form to obtain arithmetic
progressions $a_i+ b_i\N$, for $1 \leq i \leq M$. Then there is $k$,
$0 < k < n$, with $\neg p\in F_\Phi^k(\Psi)$ iff one of the equations
$a_i+b_i m=k$ has an integer solution with $0 < k < n$. The latter can
be verified by taking the integer $m = \lfloor -a_i/b_i\rfloor$
and checking whether either $a_i + b_im$ or $a_i + b_i(m + 1)$
belongs to the open interval $(0,n)$, which can clearly be done in
polynomial time.
This completes the proof of Theorem~\ref{lem:bin-ltl:krom-diamond-next-np}.
\end{proof}

We can establish the matching lower bound for $\coreLTL(\SVbox,\Rnext/\Lnext)$-formulas by using a result on the complexity of deciding inequality of regular languages  over singleton alphabets~\cite{StockmeyerM73}. In the following theorem, we give a more direct reduction of the \NP-complete problem 3SAT and repeat the argument of \citeN[Theorem~6.1]{StockmeyerM73} to construct a small number of arithmetic progressions (each with a small initial term and common difference) that give rise to models of exponential size:
\begin{theorem}\label{lem:bin-ltl:core-diamond}
  The satisfiability problem for $\coreLTL(\SVbox,\Rnext/\Lnext)$-formulas is \NP-hard.
\end{theorem}
\begin{proof} 
The proof is by reduction of 3SAT~\cite{1994-papadimitriou}. 
Let $f = \bigwedge_{i = 1}^n C_i$ be a 3CNF with $m$ variables $p_1,\dots,p_m$ and $n$ clauses $C_1,\dots,C_n$. By a propositional assignment for $f$ we understand a function $\sigma \colon \{p_1, \dots, p_m\} \to \{0,1\}$. We will represent such assignments by sets of positive natural numbers. More precisely, 
let $P_1,\dots,P_m$ be the first $m$ prime numbers; it is known that $P_m$ does not exceed $O(m^2)$~\cite{Apostol76}.  We say that a natural number $k$ \emph{represents} an assignment $\sigma$ if
$k$ is equivalent to $\sigma(p_i)$ modulo $P_i$, for all $i$, $1 \leq i \leq m$. 
Not every natural number represents an assignment. Consider the following arithmetic progressions:
\begin{equation}\label{eq:NP:AP-not}
j + P_i\cdot\N,\qquad\text{ for } 1 \leq i \leq m \text{ and } 2 \leq j < P_i.
\end{equation}
Every element of $j + P_i\cdot\N$ is equivalent to $j$ modulo $P_i$, and so, since $j \geq 2$, cannot represent an assignment. Moreover, every natural number that cannot represent an assignment belongs to one of these arithmetic progressions (see Fig.~\ref{fig:stockmeyer}).

\begin{figure}[t]
\centering%
\tabcolsep=3pt\footnotesize%
\begin{tabular}{c|c|cccc|c|ccc|c|cccc|c|c|cccc|c|ccc|c|cccc|c}
& \bf 1 & 2 & 3 & 4 &  5 & \bf 6 & 7 & 8 & 9 & \bf 10 & 11 & 12 & 13 & 14 & \bf 15 & \bf 16 & 17 & 18 & 19 &  20 & \bf 21 & 22 & 23 & 24 & \bf 25 & 26 & 27 & 28 & 29 & \bf 30 \\\hline
2 & \bf 1 & 0 & 1 & 0 & 1 & \bf 0 & 1 & 0 & 1 & \bf 0 & 1 & 0 & 1 & 0 & \bf 1 & \bf 0 & 1 & 0 & 1& 0 & \bf 1 & 0 & 1 & 0 & \bf 1 & 0 & 1 & 0 & 1 & \bf 0\\
3 & \bf 1 &  & 0 &  1 &  & \bf 0 &  1 &  & 0 &\bf 1 &  & 0 & 1 &  & \bf 0 & \bf 1 &  & 0 &  1 &  & \bf 0 & 1 &  & 0 & \bf 1 &  & 0 & 1 &  & \bf 0\\
5 & \bf 1 &  &  &  &  0 & \bf 1 &  &  &  &  \bf 0 & 1 &  &  &  & \bf 0 & \bf 1 &  &  &  & 0 & \bf 1 &  &  &  & \bf 0 & 1 &  &  &  & \bf 0\\
\end{tabular}
\caption{Natural numbers encoding assignments for 3 variables $p_1,p_2,p_3$ (shown in bold face).}\label{fig:stockmeyer}
\end{figure}

Let $C_i$ be a clause in $f$, for example, $C_i = p_{i_1} \lor \neg p_{i_2}  \lor p_{i_3}$. Consider the following progression:
\begin{equation}\label{eq:NP:AP-clause}
P_{i_1}^{1} P_{i_2}^{0} P_{i_3}^{1} + P_{i_1}P_{i_2}P_{i_3}\cdot \N.
\end{equation}
Then a natural number represents an assignment making $C_i$ true iff it \emph{does not} belong to the progressions~\eqref{eq:NP:AP-not} and \eqref{eq:NP:AP-clause}. 
Thus, a natural number represents a satisfying assignment for $f$ iff does not belong to any of the progressions of the form~\eqref{eq:NP:AP-not} and~\eqref{eq:NP:AP-clause}, for clauses in $f$.

To complete the proof, we show that the defined arithmetic progressions can be encoded in $\coreLTL(\SVbox,\Rnext/\Lnext)$. We take a propositional variable $d$, which will be shared among many formulas below.
Given an arithmetic progression $a+b\N$ (with $a\geq 0$ and $b > 0$), consider the formula
\begin{equation*}
\theta_{a,b} = u_0  \land\bigwedge_{j = 1}^a \SVbox (u_{j-1} \to \nxt u_j) \land{}\\ \SVbox (u_a \to v_0) \land \bigwedge_{j = 1}^b \SVbox (v_{j-1} \to \nxt v_j) \land \SVbox(v_b \to v_0) \land \SVbox(v_0 \to d),
\end{equation*}
where $u_0,\dots,u_a$ and $v_0,\dots,v_b$ are fresh propositional variables.
It is not hard to see that,  in every model satisfying $\theta_{a,b}$ at moment $0$, $d$ is true at moment $k \geq 0$ whenever $k$ belongs to $a + b\N$. 

So, we take $\theta_{a,b}$ for each of the arithmetic progressions~\eqref{eq:NP:AP-not} and~\eqref{eq:NP:AP-clause} and add two formulas,  $p \land \SVbox(\nxt p \to p) \land \SVbox (p \to d)$ and
$\SVbox d \to \bot$, which ensure that $d$ and a fresh variable $p$ are true at all $k \leq 0$ but $d$ is not true at all moments. The size of the resulting conjunction of $\coreLTL(\SVbox,\Rnext/\Lnext)$-formulas is $O(n\cdot m^6)$.
One can check that it is satisfiable iff $f$ is satisfiable. 
\end{proof}

\begin{theorem}\label{thm:hornLTL}
The satisfiability problem for $\hornLTL(\SVbox,\Rbox/\Lbox)$-formulas is in \PTime.
\end{theorem}
\begin{proof}
Without loss of generality, we can assume
  that $\SVbox$ does not occur in the formulas of the form $\lambda$  and that $\Rbox$, $\Lbox$ are applied only to 
  variables. Now, observe that every satisfiable $\hornLTL(\SVbox,\Rbox/\Lbox)$-formula
  $\varphi$ is satisfied in a model with a short prefix (of length
  linear in $|\varphi|$) followed by a loop of
  length $1$ (cf.~Lemma~\ref{l:structure}). 
  More precisely, let
  $\Mmf,0\models \varphi$. Similarly to the
  proof of Lemma~\ref{l:structure}, for each subformula $\Rbox p_i$ of $\varphi$, we have
  only three possible choices: if $\Rbox p_i$ is always true or always
  false, we set $m_{\Box_{\!F} p_i} = 0$; otherwise, there is $m_{\Box_{\!F}
    p_i}$ with $\Mmf,m_{\Box_{\!F} p_i} \models \neg p_i
  \land \Rbox p_i$. Symmetrically, we take all moments $m_{\Box_{\!P} p_i}$
  for all $\Lbox p_i$ in $\varphi$. Consider the following set
\begin{equation*}
\{ 0 \} \ \ \cup \ \ \{ m_{\Box_{\!P} p_i} \mid \Lbox p_i \text{ occurs in } \varphi \}\ \  \cup \ \ \{ m_{\Box_{\!F} p_i} \mid \Rbox p_i \text{ occurs in } \varphi \}
\end{equation*}
and suppose it consists of the numbers
$m_{-l} < \dots < m_{-1} < m_0 < m_1 < \dots < m_k$ with $m_0 = 0$. Let $N$ be the number of $\Lbox p_i$ and $\Rbox p_i$ occurring in $\varphi$ plus $1$. We extend the sequence by taking $m_i = m_k +1$, for $k < i \le N$, and $m_{-i} = m_{-l} - 1$, for $l < i \le N$. Therefore, $\Mmf,m_N\models \Rbox p_i$ iff
$\Mmf,m_N\models p_i$ (and symmetrically for
$\Lbox p_i$ at $m_{-N}$).  Let $\Mmf'$ be defined as follows:
\begin{equation*}
\Mmf',n\models p_i \ \ \ \text{ iff } \ \ \ 
\begin{cases}\Mmf,m_{-N} \models p_i, & \text{ if } n < - N,\\ 
\Mmf,m_n \models p_i, & \text{ if } -N \leq n \leq N,\\ 
\Mmf,m_{N} \models p_i, & \text{ if } n > N.\end{cases} 
\end{equation*}
It can be seen that $\Mmf',0\models \varphi$. 

It remains to encode the existence of such a model by means of
propositional Horn formulas, as Horn-SAT is known to be
\PTime-complete. To this end, for each propositional variable $p_i$,
we take $2N + 1$ variables $p_i^m$, for $-N \leq m \leq N$. Also, for
each formula $\Rbox p_i$, we take $2N + 1$ propositional variables,
denoted $(\Rbox p_i)^m$, for $-N \leq m \leq N$, and similarly, for
each $\Lbox p_i$, we take variables $(\Lbox p_i)^m$. Then each clause
$\lambda_1\land \dots\land\lambda_n \to \lambda$ in $\varphi$ gives
rise to the propositional clause
\begin{equation*}
\lambda_1^0\land \dots\land \lambda_n^0 \to \lambda^0
\end{equation*}
and each $\SVbox (\lambda_1\land \dots\land\lambda_n \to \lambda)$ in $\varphi$ gives rise to $2N + 1$ clauses 
\begin{equation*}
\lambda_1^m\land \dots\land \lambda_n^m \to \lambda^m, \qquad \text{for } -N \leq m \leq N.
\end{equation*}
Additionally, we need clauses that describe the
semantics of $\Rbox p_i$ in $\Mmf'$ at $m$, $-N \leq m < N$:
\begin{equation*}
(\Rbox p_i)^m \to (\Rbox p_i)^{m+1},\qquad (\Rbox p_i)^m \to p_i^{m+1},\qquad (\Rbox p_i)^{m+1} \land p_i^{m+1} \to (\Rbox p_i)^m, 
\end{equation*}
and clauses that describe the semantics of $\Rbox p_i$ in $\Mmf'$ at moment $N$:
\begin{equation*}
(\Rbox p_i)^{N} \to p_i^N,\qquad p_i^{N} \to (\Rbox p_i)^N,
\end{equation*}
and symmetric clauses for each $\Lbox p_i$ in $\varphi$. It is not hard to show that every satisfying assignment for the set of the  clauses above gives rise to a model $\Mmf'$ of $\varphi$ and, conversely, every model $\Mmf'$ of $\varphi$ with the structure as described above gives rise to a satisfying assignment for this set of clauses.
\end{proof}

\begin{remark}\label{rem:core:extension}
  In order to obtain Theorem~\ref{thm:core}, one can
  extend the proof above to formulas of the form
  $\varphi'\land\varphi''$, where $\varphi'$ is a
  $\hornLTL(\SVbox,\Rbox/\Lbox)$-formula and $\varphi''$ a conjunction
  of $\nxt^n p$, for propositional variables $p$. To this end, in the
  definition of the set $M$, one has to take $0$ together with all $n$
  for which $\nxt^n p$ occurs in $\varphi''$; the number $N$ is then
  equal to the number of those moments $n$ plus the number of all
  $\Rbox p$ and $\Lbox p$ occurring in $\varphi'$. The rest of the
  construction remains the same.
\end{remark}

\begin{theorem}\label{core-low-P}
The satisfiability problem for $\coreLTL(\SVbox,\Rbox/\Lbox)$-formulas is \PTime-hard.
\end{theorem}
\begin{proof}
The proof is by reduction of satisfiability of propositional Horn formulas with \emph{at most
ternary clauses}, which is known to be \PTime-complete~\cite{1994-papadimitriou}.
Let $f=\bigwedge_{i=1}^n C_i$ be such a formula. We define $\varphi_f$ to be the conjunction of the following formulas:
\begin{align*}
& p,\ \ \text{ for all clauses }  C_i \text{ of the form } p,\\
& p\to\bot,\ \ \text{ for all clauses }  C_i \text{ of the form } \neg p,\\
& p \to q,\ \ \text{ for all clauses }  C_i \text{ of the form } p\to q,\\
&  c_i \land (p \to \Rbox c_i) \land (q \to \Lbox c_i)  \land  (\SVbox c_i \to r),\ \  \text{ for all clauses } C_i \text{ of the form } p \land q \to r,
\end{align*}
where $c_i$ is a fresh variable for each $C_i$.
It is easy to see that $f$ is satisfiable iff $\varphi_f$ is
satisfiable.
\end{proof}

\begin{theorem}\label{krom-low-NP}
The satisfiability problem for $\kromLTL(\SVbox,\Rbox/\Lbox)$-formulas is \NP-hard.
\end{theorem}
\begin{proof}
  We proceed by reduction of the 3-colourability
  problem. Given a graph $G =(V,E)$, we use variables
  $p_0,\dots, p_4$ and $v_i$, for $v_i\in V$, to define the
  following $\kromLTL(\SVbox,\Rbox/\Lbox)$-formula:
\begin{multline*}
  \varphi_G \ \ \ = \ \ \ p_0 \ \ \land \ \ \bigwedge_{0 \leq i \leq 3} \SVbox (p_i \to \Rbox p_{i+1}) \  \land 
  \\ \bigwedge_{v_i\in V} \SVbox(p_0 \land \Rbox \neg v_i \to \bot) \quad \land  \bigwedge_{v_i\in V}\SVbox (p_4 \land  v_i \to \bot) \quad \land  \bigwedge_{(v_i,v_j)\in E} \SVbox(v_i \land v_j \to \bot).
\end{multline*}%
Intuitively, the first 4 conjuncts choose, for each vertex $v_i$ of the graph, a moment $1 \le n_i \leq 3$; the last one makes sure that $n_i \ne n_j$ in case $v_i$ and $v_j$ are connected by an edge in $G$. We show that $\varphi_G$ is satisfiable iff $G$ is 3-colourable.
Suppose $c\colon V \to
\{1,2,3\}$ is a colouring function for $G$. Define $\Mmf$ by setting $\Mmf,n
\models v_i$ just in case $c(v_i)=n$, for $v_i \in V$, and 
$\Mmf,n \models p_i$ iff $n\geq i$, for $0 \leq i \leq 4$.
Clearly,
$\Mmf,0 \models \varphi_G$. Conversely, if
$\Mmf,0\models\varphi_G$ then, for each $v_i \in V$, there is
$n_i \in \{1,2,3\}$ with $\Mmf,n_i \models v_i$ and
$\Mmf,n_i\not\models v_j$ whenever $(v_i,v_j) \in E$. Thus, $c
\colon v_i \mapsto n_i$ is a colouring function.
\end{proof}


\section{\DL\ with temporalised Roles}
\label{sec:tdl-temporalised-roles}

Now we investigate the complexity of extensions of \bDLb{} with temporalised roles of the form
\begin{equation*}
R \ \ ::=\ \ S \ \ \mid \ \ S^- \ \ \mid \ \ \SVdiamond R \ \ \mid \ \
\SVbox R,
\end{equation*}
where, as before, $S$ is a flexible or rigid role name. 
Recall that the interpretation of $\SVdiamond R$ and $\SVbox R$ is defined by taking 
$ (\SVdiamond R)^{\I(n)}=\bigcup_{k\in\Z}  R^{\I(k)}$ and $(\SVbox R)^{\I(n)}=\bigcap_{k\in\Z}  R^{\I(k)}$.

\subsection{Directed Temporal Operators and Functionality: Undecidability}

Our first result is negative. It shows, in fact, that any extension of
\bDLb{} with temporalised roles, functionality constraints on roles  
and either the next-time operator $\Rnext$ or both $\Rbox$ and $\Lbox$ on
concepts is undecidable.

\begin{theorem}\label{thm:undec}
Satisfiability of \TxDLbn{} and \TdrDLbn{} KBs is undecidable.
\end{theorem}
\begin{proof}
The proof is by reduction of the $\N\times\N$-tiling problem (see, e.g.,~\cite{BoGG97}): given a finite set $\TT$ of tile types $T = (\textit{up}(T), \textit{down}(T), \textit{left}(T), \textit{right}(T))$, decide whether $\TT$ can tile the $\N\times \N$-grid, i.e., whether there is a map $\tau\colon\N\times\N \to \TT$ such that $\textit{up}(\tau(i,j)) = \textit{down}(\tau(i,j+1))$ and $\textit{right}(\tau(i,j)) = \textit{left}(\tau(i+1,j))$, for all $(i,j)\in\N\times\N$. We assume that the colours of tiles in $\TT$ are natural numbers from $1$ to $k$, for a suitable $k > 1$.

Consider first \TxDLbn{} and, given $\TT$,  construct a KB
  $\K_\TT = (\T_\TT,\A)$ such that $\K_\TT$ is satisfiable iff $\TT$ tiles the $\N\times\N$-grid. The temporal dimension will provide us
  with the horizontal axis of the grid. The vertical axis will be constructed
  using domain elements. Let $R$ be a role such that
\begin{equation}\label{eq:undec:R}
  \mathop{\geq 2} \SVdiamond R \sqsubseteq \bot\quad\text{ and }\quad \mathop{\geq 2} \SVdiamond R^- \sqsubseteq \bot.
\end{equation}
In other words, if $xRy$ at some moment then there is no other
$y'$ with $xRy'$ at any moment of time (and similarly for $R^-$).  We
generate a sequence of domain elements: first,  we ensure 
that the concept $\exists R \sqcap \Rnext \exists R$ is non-empty, which can be done by taking $\A = \{A(a)\}$ and adding 
\begin{equation}
A ~\sqsubseteq~ \exists R \sqcap \Rnext \exists R,
\end{equation} 
to the TBox $\T_{\TT}$, and second, we add the following concept inclusion to $\T_{\TT}$ to produce a sequence:  
\begin{equation}
  \exists R^- \sqcap \Rnext R^- ~\sqsubseteq~ \exists R \sqcap \Rnext\exists R.
\end{equation}
(The reason for generating the $R$-arrows at two consecutive moments
of time will become clear below.) It is to be noted that the
produced sequence may in fact be either a finite loop or an infinite
chain of distinct elements. Now, let $T$ be a fresh concept name
for each $T\in \TT$ and let the concepts representing the tile types be disjoint:
\begin{equation}\label{eq:undec:tiles}
T \sqcap T' ~\sqsubseteq~ \bot,\qquad\text{ for }\ T\ne T'.
\end{equation}
Right after the double $R$-arrows we place the first column of tiles:
\begin{align}
  \exists R^- \sqcap \Rnext R^- \ \ &\sqsubseteq \ \ \bigsqcup_{T\in\TT} \Rnext\Rnext T.
\end{align}
The second column of tiles, whose colours match the colours of the first one, is placed $k+1$ moments later; the third column is  located $k+1$ moments after the second one, etc. (see Fig.~\ref{fig:undecidable}):
\begin{align}\label{eq:undec:tile-h}
  T \ \ \ &\sqsubseteq\hspace*{-1.4em} \bigsqcup_{T'\in\TT \text{ with } \textit{right}(T) =
    \textit{left}(T')}\hspace*{-1em} \Rnext^{k+1} T', \quad\text{ for each } \ T\in\TT.
\end{align}
This gives an $\N\times\N$-grid of tiles with matching \textit{left}--\textit{right} colours. To ensure that the \textit{up}--\textit{down} colours in this grid also match, we 
use the double $R$-arrows at the beginning and place the columns of tiles $k+1$ moments apart from each other.  Consider the following concept inclusions, for $T\in
\TT$: 
\begin{align}
\label{eq:undec:notR}
  T &\sqsubseteq \neg \exists R^-,\\
\label{eq:undec:notRi}
  T &\sqsubseteq \neg \Rnext^i \exists R^-,\quad \text{ for  } 1 \leq i \leq  k \text{ with } i \ne \textit{down}(T),\\
\label{eq:undec:up}
  T &\sqsubseteq \Rnext^{\textit{up}(T)} \exists R.
\end{align}
Inclusions~\eqref{eq:undec:notR}, \eqref{eq:undec:tiles} and~\eqref{eq:undec:notRi} ensure that between any two tiles $k+1$
moments apart there may be only one incoming $R$-arrow. This means that after the initial double $R$-arrows no other two
consecutive $R$-arrows can occur. The exact
position of the incoming $R$-arrow is uniquely determined by the
$\textit{down}$-colour of the tile, which together 
with~\eqref{eq:undec:up} guarantees that this colour matches the $\textit{up}$-colour of the tile below. Fig.~\ref{fig:undecidable} illustrates the construction (the solid vertical arrows represent $R$).

Let $\T_{\TT}$ contain all the concept inclusions defined above. It is not hard to check that $(\T_{\TT}, \A)$ is satisfiable iff $\TT$ tiles the $\N \times \N$-grid.

\begin{figure}[t]
\centering%
\begin{tikzpicture}[node distance=5mm, inner sep=0.4mm, >=latex]%
  \node (p1-0) {\tiny $|$};%
  \node (p1-1) [right=10mm of  p1-0] {\tiny $|$};%
  \node (p1-2) [right=10mm of  p1-1] {\tiny $|$};%
  \node (p1-4) [right=10mm of p1-2] {};%
  \node (p1-5) [right=10mm of p1-4] {};%
  \node (p1-6) [right=10mm of p1-5] {};%
  \node (p1-7) [right=10mm of p1-6] {};%
  \node (p1-8) [right=10mm of p1-7] {};%
  \node (p1-9) [right=10mm of p1-8] {};%
  \node (p1-10) [right=10mm of p1-9] {\tiny $|$};%
  \node (p1-11) [right=10mm of p1-10] {};%
  \node (p1-12) [right=10mm of p1-11] {};%
  \node[circle, draw] (d0-0) [above of=p1-0] {};%
  \node[circle, draw] (d0-1) [above of=p1-1] {};%
  \node[circle, draw] (d0-2) [above of=p1-2] {};%
  \node (d0-4) [above of=p1-4] {...};%
  \node[circle,draw] (d0-5) [above of=p1-5] {};
  \node(d0-6) [above of=p1-6] {...};%
  \node[circle, draw] (d0-7) [above of=p1-7] {};%
  \node(d0-8) [above of=p1-8] {...};%
  \node[circle, draw] (d0-9) [above of=p1-9] {};%
  \node[circle, draw] (d0-10) [above of=p1-10] {};%
  \node (d0-11) [above of=p1-11] {...};%
  \node[circle, draw] (p2-0) [above=10mm of d0-0] {};%
  \node[circle, draw] (p2-1) [above=10mm of d0-1] {};%
  \node[circle, fill] (p2-2) [above=10mm of d0-2] {};%
  \node (p2-4) [above=10mm of d0-4] {...};%
  \node[circle,draw] (p2-5) [above=10mm of d0-5] {};
  \node(p2-6) [above=10mm of d0-6] {...};%
  \node[circle, draw] (p2-7) [above=10mm of d0-7] {};%
  \node(p2-8) [above=10mm of d0-8] {...};%
  \node[circle, draw] (p2-9) [above=10mm of d0-9] {};%
  \node[circle, fill] (p2-10) [above=10mm of d0-10] {};%
  \node (p2-11) [above=10mm of d0-11] {...};%
  \node (p3-0) [above of = p2-0] {};
  \node (p3-1) [above of = p2-1] {};
  \node (p3-2) [above of = p2-2] {};
  \node (p3-4) [above of = p2-4] {};
  \node (p3-5) [above of = p2-5] {};
  \node (p3-6) [above of = p2-6] {};
  \node (p3-7) [above of = p2-7] {};
  \node (p3-8) [above of = p2-8] {};
  \node (p3-9) [above of = p2-9] {};
  \node (p3-10) [above of = p2-10] {};
  \node (p3-11) [above of = p2-11] {};
  \node[circle, draw] (p4-0) [above of=p3-0]{};%
  \node[circle, draw] (p4-1) [above of=p3-1] {};%
  \node[circle, fill] (p4-2) [above of=p3-2] {};%
  \node (p4-4) [above of=p3-4] {...};%
  \node[circle, draw] (p4-5) [above of=p3-5] {};%
  \node (p4-6) [above of=p3-6] {...};%
  \node[circle, draw] (p4-7) [above of=p3-7] {};%
  \node (p4-8) [above of=p3-8] {...};%
  \node[circle, draw] (p4-9) [above of=p3-9] {};%
  \node[circle, fill] (p4-10) [above of=p3-10] {};%
  \node (p4-11) [above of=p3-11] {...};%
  \node (p5-0) [above of = p4-0] {};
  \node (p5-1) [above of = p4-1] {};
  \node (p5-2) [above of = p4-2] {};
  \node (p5-5) [above of = p4-5] {};
  \node (p5-7) [above of = p4-7] {};
  \node (p5-6) [above of = p4-6] {};
  \node (p5-8) [above of = p4-8] {};
  \node (p5-9) [above of = p4-9] {};
  \node (p5-10) [above of = p4-10] {};
  \node (p5-11) [above of = p4-11] {};
  \node[circle,draw] (p6-0) [above of = p5-0] {};
  \node[circle,draw] (p6-1) [above of = p5-1] {};
  \node[circle, fill] (p6-2) [above of = p5-2] {};
  \node[circle, draw] (p6-5) [above of = p5-5] {};
  \node (p6-6) [above of = p5-6] {...};
  \node[circle, draw] (p6-7) [above of = p5-7] {};
  \node (p6-8) [above of = p5-8] {...};
  \node[circle, draw] (p6-9) [above of = p5-9] {};
  \node[circle, fill] (p6-10) [above of = p5-10] {};
  \node (p6-11) [above of = p5-11] {...};
  \path (p2-2) -- +(left:3mm) node {\small $T'$};%
  \path (p4-2) -- +(left:3mm) node {\small $T$};%
  \draw[dashed, -] (p2-2) -- (p6-2);
  \draw[dashed, -] (p2-10) -- (p6-10);
  \draw[thick, ->] (d0-0) edge node[left] {\small $R\ \ $} (p2-0);
  \draw[thick, ->] (d0-1) -- (p2-1);
 \draw[thick, ->] (p2-0) edge node[left] {\small $R\ \ $} (p4-0);
  \draw[thick, ->] (p2-1) -- (p4-1);
  \draw[thick, ->] (p2-7) -- (p4-7);
  \draw[thick, ->] (p4-0) edge node[left] {\small $R\ \ $} (p6-0);
  \draw[thick, ->] (p4-1) -- (p6-1);
  \draw[thick, ->] (p4-5) -- (p6-5);
  \path (d0-0) -- +(left:4mm) node {\small $a^\I$};%
  \path (p1-0) -- +(down:2.2mm) node {\scriptsize $0$};%
  \path (p1-1) -- +(down:2.2mm) node {\scriptsize $1$};%
  \path (p1-2) -- +(down:2.2mm) node {\scriptsize $2$};%
  \path (p1-10) -- +(down:2.2mm) node {\scriptsize $k+3$};%
  \path (p1-12) -- +(up:2mm) node {\scriptsize\hspace*{-2em}time};%
  \draw[->] (left:2mm) -- (p1-12); 
  \draw[snake=brace] (p5-2) -- (p5-5)  node[midway, above=1pt] {\scriptsize  $\textit{up}(T)$};
  \draw[snake=brace] (p3-2) -- (p3-7) node[midway,above=1pt] {\scriptsize $\textit{up}(T') = \textit{down}(T)$};
\end{tikzpicture}
\caption{Proof of Theorem~\ref{thm:undec}: the structure of the $\N\times\N$ grid.}\label{fig:undecidable}
\end{figure}
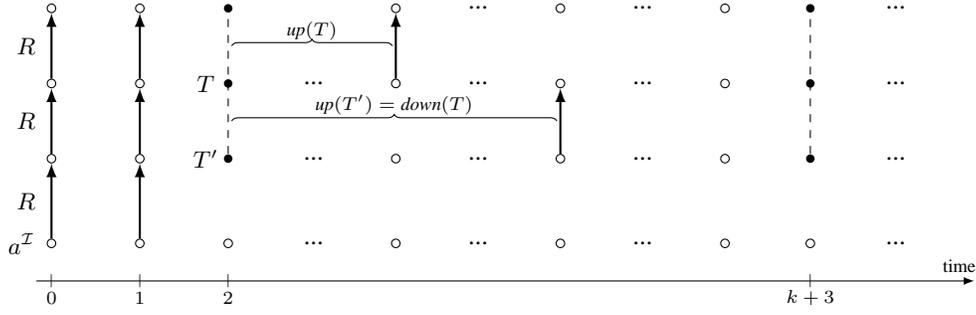

\smallskip

The proof for \TdrDLbn{} is much more involved. 
To encode the vertical axis of the $\N\times\N$-grid,  we again use the role $R$ satisfying the concept inclusions
\begin{equation}\label{eq:undec:R2}
  \mathop{\geq 2} \SVdiamond R \sqsubseteq \bot\quad\text{ and }\quad \mathop{\geq 2} \SVdiamond R^- \sqsubseteq \bot.
\end{equation}
However, as $\Rnext$ is not available in \TdrDLbn, we need a
completely different construction to ensure that the tiles match in the horizontal dimension. Indeed, in the  proof above (cf.~\eqref{eq:undec:tile-h}) we use
$\Rnext^n$ and disjunction to place a suitable tile to the right of any tile in the grid. Without the $\Rnext$ operator, we  use another role $S$ (whose $\SVdiamond S$ is also inverse-functional) and create special patterns to represent colours (as natural number from 1 to $k$) similarly to the way we paired \textit{up} and \textit{down} colours above. In order to create patterns and refer to the `next moment'\!, we use a trick  similar to the one we used in the proof of Theorem~\ref{thm:np1}: given a concept $C$ and $n \geq 0$, let 
\begin{equation*}
\Rdiamond^{=n} C =
\Rdiamond^n C \sqcap \neg \Rdiamond^{n+1} C\qquad\text{ and }\qquad \Ldiamond^{=n} C =
\Ldiamond^n C \sqcap \neg \Ldiamond^{n+1} C.
\end{equation*}
Note, however, that these $\Diamond^{=n}_{\scriptscriptstyle
F/P}C$-operators can mark a domain element with $C$ only once. So, every time we need a pattern, say of $\exists S$, of a certain length on a domain element, we create a new $S$-successor, use concepts $\textit{bit}_i$ (with various superscripts in the proof) to mark certain positions on that $S$-successor by means of the operators $\Diamond^{=i}_{\scriptscriptstyle F/P}\textit{bit}_i$  and then `transfer' the markings back to our domain element via inclusions of the form $\textit{bit}_i \sqsubseteq \exists
S^-$ and $\textit{bit}_i \sqsubseteq \neg \exists S^-$ with functional $\SVdiamond S^-$. 

The rest of the proof is organised as follows. In Step 1, we create the structure of the horizontal axis on a fixed ABox element $a$. The structure consists of repeating blocks of length $4k+4$ (to represent the four colours of the tile); each block has a certain pattern of complementary $V_0$- and $V_1$-arrows (see Fig.~\ref{fig:undec2:V0V1}), which are arranged using the same technique as we outlined for $S$ so that $a$ has a `fan' of $V_0$-successors ($y_0,y_1,\dots$) and a `fan' of $V_1$-successors ($x_0,x_1,\dots$). Then, in Step~2, we create a sequence $z_0,z_1,\dots$  of $R$-successors to represent the vertical axis (see Fig.~\ref{fig:undec2:model}) so that each of the $z_i$ repeats the structure of the horizontal axis (shifted by $k+1$ with each new $z_i$) and places tiles on a `fan' of its own $S$-successors. The particular patterns of $S$-arrows within the repeating $4k+4$ blocks will then ensure that the \textit{right}--\textit{left} colours match (within the same `fan') and, similarly, the patterns of $R$-arrows between the $z_i$ will ensure that the \textit{up}--\textit{down} colours match.

\begin{figure}[t]
\centering
\begin{tikzpicture}[yscale=0.8, point/.style={circle,draw,inner sep=0.4mm}, >=latex, looseness=0.7]
\draw[->] (-0.75,0) -- (10.5,0);
\node at (-1.1,0) {$a^\I$};
\draw[ultra thin] (-0.75,1) -- (10.5,1);
\node at (-1.1,1) {$x_0$};
\draw[ultra thin] (-0.75,2) -- (10.5,2);
\node at (-1.1,2) {$y_0$};
\node[point,fill=black,label=below:{\footnotesize $0$}] (a0) at (0,0) {};
\draw[ultra thin] (a0) -- ++(0,-0.1);
\node[point,fill=black] (a05) at (0.5,0) {};
\node[point,fill=black] (a15) at (1.5,0) {};
\node[point,fill=black,label=below:{\footnotesize $k$}] (a1) at (2,0) {};
\draw[ultra thin] (a1) -- ++(0,-0.18);
\node[point,fill=black,label=below:{\footnotesize $k\!+\!1$}] (a2) at (2.5,0) {};
\draw[ultra thin] (a2) -- ++(0,-0.18);
\node[point,fill=black] (a25) at (3,0) {};
\node[point,fill=black] (a35) at (4,0) {};
\node[point,fill=black,label=below:{\footnotesize $2k\!+\!1$\hspace*{1em}}] (a3) at (4.5,0) {};
\draw[ultra thin] (a3) -- ++(0,-0.18);
\node[point,fill=black,label=below:{\footnotesize\rule{0pt}{16pt} $2k\!+\!2$}] (a4) at (5,0) {};
\draw[ultra thin, dashed] (a4) -- ++(0,-0.5);
\node[point,fill=black,label=below:{\footnotesize \hspace*{1em}$2k\!+\!3$}] (a5) at (5.5,0) {};
\draw[ultra thin] (a5) -- ++(0,-0.18);
\node[point,fill=black] (a55) at (6,0) {};
\node[point,fill=black] (a65) at (9,0) {};
\node[point,fill=black,label=below:{\footnotesize $4k\!+\!3$}] (a6) at (9.5,0) {};
\draw[ultra thin] (a6) -- ++(0,-0.18);
\foreach \x in {-0.5,0,0.5,1,1.5,2,2.5,3,3.5,4,4.5,5,5.5,6,6.5,7,7.5,8,8.5,9,9.5,10} 
{
\node[point,fill=white]  at (\x,1) {};
\node[point,fill=white] at (\x,2) {};
};
\node[point,fill=black] (x0) at (0,1) {};
\node[point,fill=black] (x05) at (0.5,1) {};
\node[point,fill=black] (x15) at (1.5,1) {};
\node[point,fill=black] (x1) at (2,1) {};
\node[point,fill=black] (y2) at (2.5,2) {};
\node[point,fill=black] (y25) at (3,2) {};
\node[point,fill=black] (y35) at (4,2) {};
\node[point,fill=black] (y3) at (4.5,2) {};
\node[point,fill=black] (x4) at (5,1) {};
\node[point,fill=black] (y5) at (5.5,2) {};
\node[point,fill=black] (y55) at (6,2) {};
\node[point,fill=black] (y65) at (9,2) {};
\node[point,fill=black] (y6) at (9.5,2) {};
\draw[->] (a0) -- (x0);
\draw[->] (a05) -- (x05);
\draw[->] (a15) -- (x15);
\draw[->] (a1) -- (x1);
\fill[white] (0.75,-0.25) rectangle +(0.5,2.5);
\foreach \i in {0,1,2} {
	\draw[dotted] (0.75,\i) -- +(0.5,0);
}
\draw[->,bend right] (a2) to (y2);
\draw[->,bend right] (a25) to (y25);
\draw[->,bend right] (a35) to (y35);
\draw[->,bend right] (a3) to (y3);
\fill[white] (3.25,-0.25) rectangle +(0.5,2.5);
\foreach \i in {0,1,2} {
	\draw[dotted] (3.25,\i) -- +(0.5,0);
}
\draw[->] (a4) -- (x4);
\draw[->,bend right] (a5) to (y5);
\draw[->,bend right] (a55) to (y55);
\draw[->,bend right] (a65) to (y65);
\draw[->,bend right] (a6) to (y6);
\fill[white] (6.25,-0.25) rectangle +(2.5,2.5);
\foreach \i in {0,1,2} {
	\draw[dotted] (6.25,\i) -- +(2.5,0);
}
\draw[ultra thin] ($(x0) + (-0.2,0)$) -- ($(x0) + (-0.2,0.2)$) -- node[above, midway] {\small $k+1$} ($(x1) + (0.2,0.2)$)  -- ($(x1) + (0.2,0)$);
\draw[ultra thin] ($(y2) + (-0.2,0)$) -- ($(y2) + (-0.2,0.2)$) -- node[above, midway] {\small $k+1$} ($(y3) + (0.2,0.2)$)  -- ($(y3) + (0.2,0)$);
\draw[ultra thin] ($(y5) + (-0.2,0)$) -- ($(y5) + (-0.2,0.2)$) -- node[above, midway] {\small $2k+1$} ($(y6) + (0.2,0.2)$)  -- ($(y6) + (0.2,0)$);
\end{tikzpicture}
\caption{The structure of the horizontal axis: $x_0$ is a $V_1$-successor of $a^\I$ and $y_0$ is a $V_0$-successor of $a^\I$.}\label{fig:undec2:H}
\end{figure}
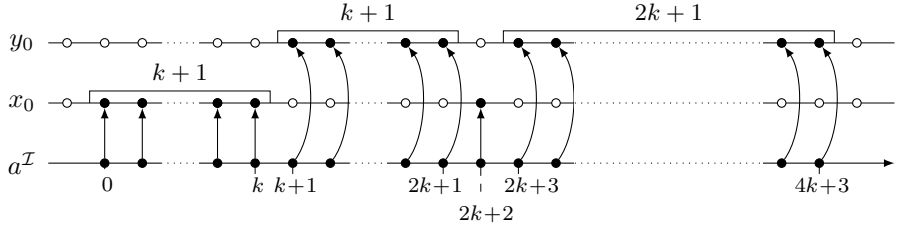

\smallskip

\noindent\textbf{Step 1.} We encode the horizontal axis using the ABox $\A = \{A(a)\}$ and a number of concept inclusions with roles $V_0$, $V_1$ and concepts ${\it bit}^{V_1}_i$, for $1 \le i \le 2k+2$, and ${\it bit}^{V_0}_i$, for $1 \le i \le 3k+2$. Consider first the following concept inclusions:
\begin{align}
\label{eq:undec2:AV1}
A &\sqsubseteq \exists V_1 \sqcap \Lbox \neg\exists V_1,\\
\label{eq:undec2:V1func}
\mathop{\geq 2} \SVdiamond V_1^- &\sqsubseteq \bot,\\
 \label{eq:undec2:V1gen}
\exists V_1^- \sqcap  \Lbox\neg \exists V_1^- & \sqsubseteq \bigsqcap_{1 \leq i \leq 2k+2} \hspace*{-1em}\Rdiamond^{=i}\, {\it bit}^{V_1}_i \ \ \sqcap \ \ \Rbox^{2k+3} \neg \exists V_1^-,\\
\label{eq:undec2:V1bits1} 
  {\it bit}^{V_1}_i &\sqsubseteq \exists V_1^-,\qquad \text{for } 1\leq i \leq k,\\
\label{eq:undec2:V1bits2} 
  {\it bit}^{V_1}_{k + i} &\sqsubseteq \neg \exists V_1^-,\qquad \text{for } 1 \leq i \leq k + 1,\\
\label{eq:undec2:V1bits3} 
  {\it bit}^{V_1}_{2k + 2} &\sqsubseteq \exists V_1^-.
\end{align}
Suppose all of them hold in an interpretation $\I$. Then, 
by~\eqref{eq:undec2:AV1}, $a^\I$ has a $V_1$-successor, say $x_0$, at moment 0 and no $V_1$-successor at any preceding moment. By~\eqref{eq:undec2:V1func}, $x_0$ does not have a $V_1$-predecessor before 0, and so, by~\eqref{eq:undec2:V1gen}--\eqref{eq:undec2:V1bits3}, $x_0$ has a $V_1$-predecessor at every moment $i$ with $0 \leq i \leq k$ and $i = 2k+2$, and no $V_1$-predecessor at any other times. By~\eqref{eq:undec2:V1func}, all these $V_1$-predecessors must coincide with $a^\I$ (Fig.~\ref{fig:undec2:H}).
We also need similar concept inclusions for the role $V_0$:
\begin{align}
\label{eq:undec2:AV0}  
  A & \sqsubseteq   \Lbox \neg\exists V_0,\\
\label{eq:undec2:V0func}
\mathop{\geq 2} \SVdiamond V_0^- &\sqsubseteq \bot,\\
 \label{eq:undec2:V0gen}
  \exists V_0^- \sqcap \Lbox \neg \exists V_0^- & \sqsubseteq \bigsqcap_{1 \leq i \leq 3k+2} \hspace*{-1em}\Rdiamond^{=i}\, {\it bit}^{V_0}_i \ \ \sqcap \ \ \Rbox^{3k+3} \neg\exists V_0^-,\\
\label{eq:undec2:V0bits1} 
  {\it bit}^{V_0}_i &\sqsubseteq \exists V_0^-,\qquad \text{for } 1\leq i \leq k,\\
\label{eq:undec2:V0bits2} 
  {\it bit}^{V_0}_{k + 1} &\sqsubseteq \neg \exists V_0^-,\\
\label{eq:undec2:V0bits3} 
  {\it bit}^{V_0}_{k + 1 + i} &\sqsubseteq  \exists V_0^-,\qquad \text{for } 1 \leq i \leq 2k+1,
\end{align}
together with 
\begin{align}
\label{eq:undec2:V0orV1}
 A &\sqsubseteq \Rbox (\exists V_1 \sqcup \exists V_0),\\
\label{eq:undec2:V0disjV1}
  \exists V_0 \sqcap  \exists V_1&\sqsubseteq \bot.
\end{align}
Suppose all of them hold in $\I$. 
By~\eqref{eq:undec2:V0orV1}, \eqref{eq:undec2:V0disjV1}, at each
moment after 0, $a^\I$ has either a $V_0$- or a
$V_1$-successor. By~\eqref{eq:undec2:AV1}, \eqref{eq:undec2:V0disjV1}
and the observations above, $a^\I$ cannot have a $V_0$-successor in
the interval between 0 and $k$. Suppose $y_0$ is a
$V_0$-successor of $a^\I$ at $k+1$ (that this is the case will
be ensured by~\eqref{eq:undec:V1f}). By~\eqref{eq:undec2:AV0},~\eqref{eq:undec2:V0func}, $y_0$ has no $V_0$-predecessors before
0; so, by~\eqref{eq:undec2:V0gen}--\eqref{eq:undec2:V0bits3},
$y_0$ has $V_0$-predecessors at the moments $i$ with $k+1 \leq i \leq
2k+1$ and $2k+3 \leq i \leq 4k+3$ and no $V_0$-predecessors at 
other moments. By~\eqref{eq:undec2:V0func}, all these
$V_0$-predecessors coincide with $a^\I$ (see Fig.~\ref{fig:undec2:H}).
We show now that if
\begin{align}
\label{eq:undec:V1f}
\mathop{\geq 2} V_1 &\sqsubseteq \bot
\end{align}
also holds in $\I$ then $a^\I$ has a $V_0$-successor at $k+1$.
Indeed, suppose $a^\I$ has a $V_1$-successor $z$ at $k+1$. Then, by~\eqref{eq:undec2:AV1}, the choice of $x_0$ and~\eqref{eq:undec:V1f}, $z$ cannot be a $V_1$-successor of $a^\I$ at any moment before that. So, $z$ must belong to the left-hand side concept of~\eqref{eq:undec2:V1gen}, which triggers the following pattern of $V_1$-successors of $a^\I$: $x_0$  at moments $i$ with $0 \leq i \leq k$, $z$ at $i$ with $k+1 \leq i \leq 2k+1$, $x_0$ at $2k+2$ and $z$ at $3k+3$ (see Fig.~\ref{fig:undec2:H2}). This leaves only the moments $i$, for $2k+3 \leq i \leq 3k+2$, without any $V_0$- or $V_1$-successors. But in this case $a^\I$ cannot have any $V_0$- or $V_1$-successor at $2k+3$. Indeed, such a $V_0$-successor $z'$ would have no $V_0$-predecessor at any moment before $2k+3$, and so, by~\eqref{eq:undec2:V0func}--\eqref{eq:undec2:V0bits3}, would remain a $V_0$-successor of $a^\I$ for $k+1$ consecutive moments, which is impossible with only $k$ available slots; by a similar argument and~\eqref{eq:undec:V1f}, $a^\I$ has no $V_1$-successor at $2k+3$.
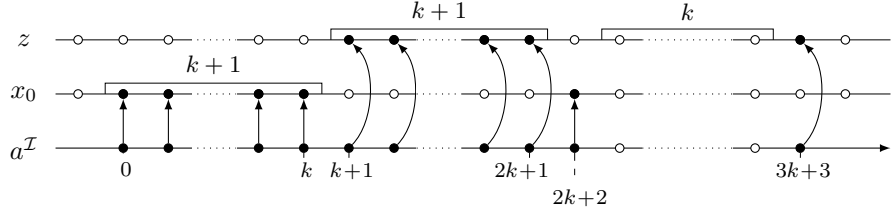
\begin{figure}[t]
\centering
\begin{tikzpicture}[yscale=0.8, point/.style={circle,draw,inner sep=0.4mm}, >=latex, looseness=0.7,yscale=0.9,xscale=1.2]
\draw[->] (-0.75,0) -- (8.5,0);
\node at (-1.1,0) {$a^\I$};
\draw[ultra thin] (-0.75,1) -- (8.5,1);
\node at (-1.1,1) {$x_0$};
\draw[ultra thin] (-0.75,2) -- (8.5,2);
\node at (-1.1,2) {$z$};
\node[point,fill=black,label=below:{\footnotesize $0$}] (a0) at (0,0) {};
\draw[ultra thin] (a0) -- ++(0,-0.1);
\node[point,fill=black] (a05) at (0.5,0) {};
\node[point,fill=black] (a15) at (1.5,0) {};
\node[point,fill=black,label=below:{\footnotesize $k$}] (a1) at (2,0) {};
\draw[ultra thin] (a1) -- ++(0,-0.18);
\node[point,fill=black,label=below:{\footnotesize $k\!+\!1$}] (a2) at (2.5,0) {};
\draw[ultra thin] (a2) -- ++(0,-0.18);
\node[point,fill=black] (a25) at (3,0) {};
\node[point,fill=black] (a35) at (4,0) {};
\node[point,fill=black,label=below:{\footnotesize $2k\!+\!1$\hspace*{1em}}] (a3) at (4.5,0) {};
\draw[ultra thin] (a3) -- ++(0,-0.18);
\node[point,fill=black,label=below:{\footnotesize\rule{0pt}{16pt} $2k\!+\!2$}] (a4) at (5,0) {};
\draw[ultra thin, dashed] (a4) -- ++(0,-0.5);
\node[point,fill=black,label=below:{\footnotesize $3k\!+\!3$}] (a5) at (7.5,0) {};
\draw[ultra thin] (a5) -- ++(0,-0.18);
\foreach \x in {-0.5,0,0.5,1,1.5,2,2.5,3,3.5,4,4.5,5,5.5,6,6.5,7,7.5,8} 
{
\node[point,fill=white]  at (\x,1) {};
\node[point,fill=white] at (\x,2) {};
};
\node[point,fill=white] at (5.5,0) {};
\node[point,fill=white] at (7,0) {};

\node[point,fill=black] (x0) at (0,1) {};
\node[point,fill=black] (x05) at (0.5,1) {};
\node[point,fill=black] (x15) at (1.5,1) {};
\node[point,fill=black] (x1) at (2,1) {};
\node[point,fill=black] (y2) at (2.5,2) {};
\node[point,fill=black] (y25) at (3,2) {};
\node[point,fill=black] (y35) at (4,2) {};
\node[point,fill=black] (y3) at (4.5,2) {};
\node[point,fill=black] (x4) at (5,1) {};
\node[point,fill=black] (y5) at (7.5,2) {};
\draw[->] (a0) -- (x0);
\draw[->] (a05) -- (x05);
\draw[->] (a15) -- (x15);
\draw[->] (a1) -- (x1);
\fill[white] (0.75,-0.25) rectangle +(0.5,2.5);
\foreach \i in {0,1,2} {
	\draw[dotted] (0.75,\i) -- +(0.5,0);
}
\draw[->,bend right] (a2) to (y2);
\draw[->,bend right] (a25) to (y25);
\draw[->,bend right] (a35) to (y35);
\draw[->,bend right] (a3) to (y3);
\fill[white] (3.25,-0.25) rectangle +(0.5,2.5);
\foreach \i in {0,1,2} {
	\draw[dotted] (3.25,\i) -- +(0.5,0);
}
\draw[->] (a4) -- (x4);
\draw[->,bend right] (a5) to (y5);
\fill[white] (5.75,-0.25) rectangle +(1,2.5);
\foreach \i in {0,1,2} {
	\draw[dotted] (5.75,\i) -- +(1,0);
}
\draw[ultra thin] ($(x0) + (-0.2,0)$) -- ($(x0) + (-0.2,0.2)$) -- node[above, midway] {\small $k+1$} ($(x1) + (0.2,0.2)$)  -- ($(x1) + (0.2,0)$);
\draw[ultra thin] ($(y2) + (-0.2,0)$) -- ($(y2) + (-0.2,0.2)$) -- node[above, midway] {\small $k+1$} ($(y3) + (0.2,0.2)$)  -- ($(y3) + (0.2,0)$);
\draw[ultra thin] ($(y5) + (-2.2,0)$) -- ($(y5) + (-2.2,0.2)$) -- node[above, midway] {\small $k$} ($(y5) + (-0.3,0.2)$)  -- ($(y5) + (-0.3,0)$);
\end{tikzpicture}
\caption{A gap of $k$ moments on the horizontal axis: both $x_0$ and $z$ are $V_1$-successors of $a^\I$.}\label{fig:undec2:H2}
\end{figure}
Next, if in addition 
\begin{align}
\label{eq:undec:V0f}
\mathop{\geq 2} V_0 &\sqsubseteq \bot
\end{align}
holds in $\I$, then $a^\I$ has a $V_1$-successor, $x_1$, at $4k+4$. 
Indeed, using~\eqref{eq:undec:V0f} and an argument similar to the one above, one can show that if $a^\I$ has a $V_0$-successor $z$ at $4k+4$ then $z$ is different from $y_1$ and $z$ cannot have $V_0$-predecessors before $4k+4$. But then the pattern of $V_0$-successors required by~\eqref{eq:undec2:V0gen}--\eqref{eq:undec2:V0bits3} would make it impossible for $a^\I$ to have any $V_0$- or $V_1$-successor at $6k+6$, where $z$ has no $V_0$-predecessor. 

Thus, we find ourselves in the same situation as at the very beginning of the construction, but with $x_1$ in place of $x_0$. By repeating the same argument again and again, we obtain domain elements $x_0,x_1,\dots$ and $y_0,y_1,\dots$ of the interpretation $\I$ which are, respectively, $V_1$- and $V_0$-successors of $a^\I$ at the moments of time indicated in Fig.~\ref{fig:undec2:V0V1} by black points and intervals. 

\smallskip

\noindent\textbf{Step 2.} We are now in a position to encode the $\N\times\N$-tiling problem. Let us regard each $T \in \TT$ as a fresh concept name satisfying the disjointness concept inclusions
\begin{align}
T \sqcap T' \sqsubseteq \bot,\qquad\text{ for }T\ne T'.
\end{align}
Consider the following concept inclusions:
\begin{align}
\label{eq:undec-fp-AR} 
A &\sqsubseteq \exists R \sqcap  \Lbox \neg \exists R,\\
\label{eq:undec-fp-row-begin} 
\exists R^- \sqcap \Lbox\neg \exists R^- &\sqsubseteq  \Ldiamond^{=2k+1} {\it row\text{-}start},\\
\label{eq:undec-fp-row-end} 
{\it row\text{-}start} &\sqsubseteq \exists S \sqcap \Lbox \neg \exists S,\\
\label{eq:undec-fp-tile-placed}
\exists S^- \sqcap \Lbox\neg  \exists S^- &\sqsubseteq \bigsqcup_{T \in \TT} T.
\end{align}
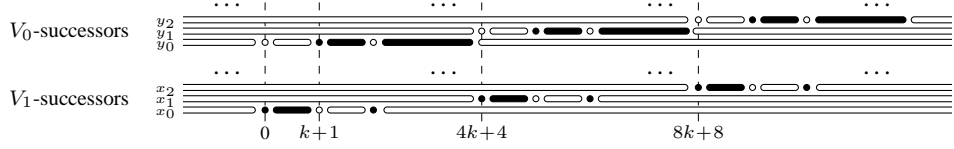
\begin{figure}[t]
\centering
\begin{tikzpicture}[>=latex,xscale=1.2,
	block/.style={rectangle,draw=black,rounded corners=0.4mm,inner ysep=0.4mm,inner xsep=2.5mm},
	point/.style={rectangle,draw=black,rounded corners=0.4mm,inner sep=0.4mm}]
\draw[ultra thin,dashed] (0.9,-2.25) -- +(0,1.6);
\node at (0.9,-2.4) {\scriptsize 0};
\draw[ultra thin,dashed] (1.5,-2.25) -- +(0,1.6);
\node at (1.5,-2.4) {\scriptsize $k\!+\!1$};
\draw[ultra thin,dashed] (3.3,-2.25) -- +(0,1.6);
\node at (3.3,-2.4) {\scriptsize $4k\!+\!4$};
\draw[ultra thin,dashed] (5.7,-2.25) -- +(0,1.6);
\node at (5.7,-2.4) {\scriptsize $8k\!+\!8$};
%
\foreach \x/\y/\n in {1.2/-1.2/0,3.6/-1.05/1,6/-0.9/2} {
\node[point, fill=black] (v0\n) at (\x-0.3,\y-0.9) {};
\node[block, fill=black] (v0r\n) at (\x,\y-0.9) {};
\node[point, fill=white] at (\x+0.3,\y-0.9) {};
\node[block, fill=white] at (\x+0.6,\y-0.9) {};
\node[point, fill=black] at (\x+0.9,\y-0.9) {};
%
\node[point, fill=white] at (\x-0.3,\y) {};
\node[block, fill=white] at (\x,\y) {};
\node[point, fill=black] at (\x+0.3,\y) {};
\node[block, fill=black] at (\x+0.6,\y) {};
\node[point, fill=white] at (\x+0.9,\y) {};
\node[block, inner xsep=6mm,fill=black] at (\x+1.5,\y) {};
}
\begin{scope}
\clip (0,-2.2) rectangle +(8.5,1.6);
\node[block, fill=white, inner xsep=9.5mm] at (0,-2.1) {};
\node[block, fill=white, inner xsep=38.5mm] at (0,-1.95) {};
\node[block, fill=white, inner xsep=67mm] at (0,-1.8) {};
\node[block, fill=white, inner xsep=9.5mm] at (0,-1.2) {};
\node[block, fill=white, inner xsep=38.5mm] at (0,-1.05) {};
\node[block, fill=white, inner xsep=67mm] at (0,-0.9) {};
\node[block, fill=white, inner xsep=67mm] at (7.8,-2.1) {};
\node[block, fill=white, inner xsep=38.5mm] at (7.8,-1.95) {};
\node[block, fill=white, inner xsep=9.5mm] at (7.8,-1.8) {};
\node[block, fill=white, inner xsep=67mm] at (8.85,-1.2) {};
\node[block, fill=white, inner xsep=38.5mm] at (8.85,-1.05) {};
\node[block, fill=white, inner xsep=9.5mm] at (8.85,-0.9) {};
\end{scope}
\foreach \x in {0,2.4,4.8,7.2} {
\node at (\x+0.5,-1.6) {\dots};
\node at (\x+0.5,-0.7) {\dots};
}
\node at (-1.3,-1.95) {\footnotesize $V_1$-successors};
\node at (-1.3,-1.05) {\footnotesize $V_0$-successors};
\node at (-0.2,-2.15) {\tiny $x_0$};
\node at (-0.2,-2) {\tiny $x_1$};
\node at (-0.2,-1.85) {\tiny $x_2$};
\node at (-0.2,-1.25) {\tiny $y_0$};
\node at (-0.2,-1.1) {\tiny $y_1$};
\node at (-0.2,-0.95) {\tiny $y_2$};
\end{tikzpicture} 
\caption{$V_0$- and $V_1$-successors of $a$ in a model of $\K_\TT$.}\label{fig:undec2:V0V1}
\end{figure}%
Intuitively, \eqref{eq:undec-fp-AR} says that $a$ has an $R$-successor, say $z_0$, at the moment 0, and no $R$-successors before 0. Then, by~\eqref{eq:undec:R2}, $z_0$ has no $R$-predecessors before 0. Axioms~\eqref{eq:undec-fp-row-begin}--\eqref{eq:undec-fp-tile-placed} make sure that $z_0$ has an $S$-successor, $w$, which is an instance of $T$ at  $-(2k+1)$, for some tile $T$. In this case, we say that $T$ \emph{is placed on} $z_0$ (rather than on $w$). Tiles will also be placed on domain elements having $S$-successors with a specific pattern of concepts $\exists S^-$ given by the following concept inclusions:
\begin{align}
\label{eq:undec-fp-role-funct}
\mathop{\geq 2} \SVdiamond S^- & \sqsubseteq \bot,\\
\label{eq:undec2:Sfunc}
\mathop{\geq 2} S & \sqsubseteq \bot,\\
   T & \sqsubseteq \bigsqcap_{1 \leq i \leq 6k+4} \hspace*{-1em}\Rdiamond^{=i}\,{\it bit}^T_i 
   \ \ \sqcap \ \  \Rbox^{6k+5} \neg\exists S^-, \label{eq:undec-fp-tile}\hspace*{-8em}\\
  \label{eq:undec-fp-role-pos-begin} 
 &&  {\it bit}^T_i &\sqsubseteq \exists S^-,\quad &\text{for } 1 \leq i < k,\\ 
  {\it bit}^T_k &\sqsubseteq \neg \exists S^-, &
  {\it bit}^T_{k+i} &\sqsubseteq \begin{cases}\neg\exists S^-,& \text{if }i= \textit{left}(T),\\ 
  \exists S^-,&\text{otherwise},\end{cases} \quad&\text{for }1 \leq i \leq k, \\
  {\it bit}^T_{2k+1} &\sqsubseteq \neg \exists S^-, &
  {\it bit}^T_{2k+1+i} &\sqsubseteq \begin{cases}\neg\exists S^-,& \text{if }i= \textit{down}(T),\\ \exists S^-,&\text{otherwise},\end{cases} \quad&\text{for }1 \leq i \leq k, \\
  {\it bit}^T_{3k+2} &\sqsubseteq \neg \exists S^-, &
{\it bit}^T_{3k+2+i} &\sqsubseteq \begin{cases}\neg\exists S^-,& \text{if }i= \textit{up}(T),\\ \exists S^-,&\text{otherwise},\end{cases} \quad&\text{for }1 \leq i \leq k,\\
  {\it bit}^T_{4k+3} &\sqsubseteq \exists S^-, &  
  {\it bit}^T_{4k+3+i} &\sqsubseteq \neg \exists S^-,\quad&\text{for }1 \leq i \leq k, \\
  {\it bit}^T_{5k+4} &\sqsubseteq \exists S^-, &  
{\it bit}^T_{5k+4+i} &\sqsubseteq \begin{cases}\exists S^-,& \text{if }i= \textit{right}(T),\\ \neg\exists S^-,&\text{otherwise},\end{cases} \quad&\text{for }1 \leq i \leq k.
  \label{eq:undec-fp-role-pos-end}
\end{align}
Suppose a domain element $w$ is an instance of $T$ at some moment $t$, for some $T \in \TT$. Then $w$ will be an instance of $\exists S^-$ at the moments $t,\dots,t+k-1$. We think of this time interval on $w$ (and, as before, on $z_0$) as the \emph{plug}, or the \emph{P-section}. After the plug we have a one-instant \emph{gap} (where $w$ is an instance of $\neg \exists S^-$). The gap is followed by a sequence of $k$ moments of time that represent $\textit{left}(T)$ in the sense that only at the $i$th moment of the sequence, where $i = \textit{left}(T)$, $w$ does \emph{not} have an $S$-predecessor. Then we again have a one-instant gap, followed by a sequence of $k$-moments representing $\textit{down}(T)$ (in the same sense), another one-instant gap and a sequence representing $\textit{up}(T)$ (see Fig.~\ref{fig:undec2:tile}).
At the next moment, $t + 4k +3$, $w$ will be an instance of $\exists S^-$; then we have $k$ gaps (i.e.,  $\neg \exists S^-$), called the \emph{socket}, or the \emph{S-section}. After the socket, at $t + 5k +4$,  $w$ is again an instance of $\exists S^-$, and then we have a sequence of $k$ moments representing `inverted' $\textit{right}(T)$: the $i$th moment of this sequence has an $S$-predecessor iff $i = \textit{right}(T)$. 
We note that, by~\eqref{eq:undec-fp-role-funct}, the pattern of $\exists S^-$ on $w$ in Fig.~\ref{fig:undec2:tile} is reflected by the pattern of $\exists S$ on the $S$-predecessor $z_0$ of $w$ at $t$, which (partly) justifies our terminology when we say that \emph{tile $T$ is placed on $z_0$} (rather than on $w$).

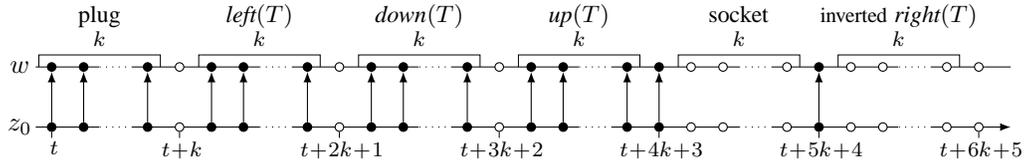
\begin{figure}[b]
\centering
\begin{tikzpicture}[xscale=0.85, point/.style={circle,draw,inner sep=0.4mm}, >=latex, looseness=0.7,yscale=0.8]
\draw[->] (-0.25,0) -- (15,0);
\node at (-0.5,0) {$z_0$};
\draw[ultra thin] (-0.25,1) -- (15,1);
\node at (-0.5,1) {$w$};
\node[point,fill=black,label=below:{\footnotesize $t$}] (a0) at (0,0) {};
\draw[ultra thin] (a0) -- ++(0,-0.18);
\node[point,fill=black] (a05) at (0.5,0) {};
\node[point,fill=black] (a15) at (1.5,0) {};
\node[point,fill=white,label=below:{\footnotesize $t\! +\! k$}] (a1) at (2,0) {};
\draw[ultra thin] (a1) -- ++(0,-0.18);
\node[point,fill=black] (a2) at (2.5,0) {};
\node[point,fill=black] (a25) at (3,0) {};
\node[point,fill=black] (a35) at (4,0) {};
\node[point,fill=white,label=below:{\footnotesize $t \!+\! 2k\!+\!1$}] (a3) at (4.5,0) {};
\draw[ultra thin] (a3) -- ++(0,-0.18);
\node[point,fill=black] (a4) at (5,0) {};
\node[point,fill=black] (a45) at (5.5,0) {};
\node[point,fill=black] (a55) at (6.5,0) {};
\node[point,fill=white,label=below:{\footnotesize $t \!+\! 3k\!+\!2$}] (a5) at (7,0) {};
\draw[ultra thin] (a5) -- ++(0,-0.18);
\node[point,fill=black] (a6) at (7.5,0) {};
\node[point,fill=black] (a65) at (8,0) {};
\node[point,fill=black] (a75) at (9,0) {};
\node[point,fill=black,label=below:{\footnotesize $t \!+\! 4k\!+\!3$}] (a7) at (9.5,0) {};
\draw[ultra thin] (a7) -- ++(0,-0.18);
\node[point,fill=white] (a8) at (10,0) {};
\node[point,fill=white] (a85) at (10.5,0) {};
\node[point,fill=white] (a95) at (11.5,0) {};
\node[point,fill=black,label=below:{\footnotesize $t\! +\! 5k\!+\!4$}] (a9) at (12,0) {};
\draw[ultra thin] (a9) -- ++(0,-0.18);
\node[point,fill=white] (a10) at (12.5,0) {};
\node[point,fill=white] (a105) at (13,0) {};
\node[point,fill=white] (a115) at (14,0) {};
\node[point,fill=white,label=below:{\footnotesize $t\! +\! 6k\!+\!5$}] (a11) at (14.5,0) {};
\draw[ultra thin] (a11) -- ++(0,-0.18);

\node[point,fill=black] (x0) at (0,1) {};
\node[point,fill=black] (x05) at (0.5,1) {};
\node[point,fill=black] (x15) at (1.5,1) {};
\node[point,fill=white] (x1) at (2,1) {};
\node[point,fill=black] (x2) at (2.5,1) {};
\node[point,fill=black] (x25) at (3,1) {};
\node[point,fill=black] (x35) at (4,1) {};
\node[point,fill=white] (x3) at (4.5,1) {};
\node[point,fill=black] (x4) at (5,1) {};
\node[point,fill=black] (x45) at (5.5,1) {};
\node[point,fill=black] (x55) at (6.5,1) {};
\node[point,fill=white] (x5) at (7,1) {};
\node[point,fill=black] (x6) at (7.5,1) {};
\node[point,fill=black] (x65) at (8,1) {};
\node[point,fill=black] (x75) at (9,1) {};
\node[point,fill=black] (x7) at (9.5,1) {};
\node[point,fill=white] (x8) at (10,1) {};
\node[point,fill=white] (x85) at (10.5,1) {};
\node[point,fill=white] (x95) at (11.5,1) {};
\node[point,fill=black] (x9) at (12,1) {};
\node[point,fill=white] (x10) at (12.5,1) {};
\node[point,fill=white] (x105) at (13,1) {};
\node[point,fill=white] (x115) at (14,1) {};
\node[point,fill=white] (x11) at (14.5,1) {};
\draw[->] (a0) -- (x0);
\draw[->] (a05) -- (x05);
\draw[->] (a15) -- (x15);
\fill[white] (0.75,-0.25) rectangle +(0.5,2.5);
\foreach \i in {0,1} {
	\draw[dotted] (0.75,\i) -- +(0.5,0);
}
\draw[->] (a2) -- (x2);
\draw[->] (a25) -- (x25);
\draw[->] (a35) -- (x35);
\fill[white] (3.25,-0.25) rectangle +(0.5,2.5);
\foreach \i in {0,1} {
	\draw[dotted] (3.25,\i) -- +(0.5,0);
}
\draw[->] (a4) -- (x4);
\draw[->] (a45) -- (x45);
\draw[->] (a55) -- (x55);
\fill[white] (5.75,-0.25) rectangle +(0.5,2.5);
\foreach \i in {0,1} {
	\draw[dotted] (5.75,\i) -- +(0.5,0);
}
\draw[->] (a6) -- (x6);
\draw[->] (a65) -- (x65);
\draw[->] (a75) -- (x75);
\fill[white] (8.25,-0.25) rectangle +(0.5,2.5);
\foreach \i in {0,1} {
	\draw[dotted] (8.25,\i) -- +(0.5,0);
}
\draw[->] (a7) -- (x7);
\draw[->] (a9) -- (x9);
\fill[white] (10.75,-0.25) rectangle +(0.5,2.5);
\foreach \i in {0,1} {
	\draw[dotted] (10.75,\i) -- +(0.5,0);
}
\fill[white] (13.25,-0.25) rectangle +(0.5,2.5);
\foreach \i in {0,1} {
	\draw[dotted] (13.25,\i) -- +(0.5,0);
}
\draw[ultra thin] ($(x0) + (-0.2,0)$) -- ($(x0) + (-0.2,0.2)$) -- node[above, midway] {\parbox{5em}{\small\centering plug\\\footnotesize $k$}} ($(x1) + (-0.3,0.2)$)  -- ($(x1) + (-0.3,0)$);
\draw[ultra thin] ($(x2) + (-0.2,0)$) -- ($(x2) + (-0.2,0.2)$) -- node[above, midway] {\parbox{5em}{\small\centering $\textit{left}(T)$\\\footnotesize $k$}} ($(x3) + (-0.3,0.2)$)  -- ($(x3) + (-0.3,0)$);
\draw[ultra thin] ($(x4) + (-0.2,0)$) -- ($(x4) + (-0.2,0.2)$) -- node[above, midway] {\parbox{5em}{\small\centering $\textit{down}(T)$\\\footnotesize $k$}} ($(x5) + (-0.3,0.2)$)  -- ($(x5) + (-0.3,0)$);
\draw[ultra thin] ($(x6) + (-0.2,0)$) -- ($(x6) + (-0.2,0.2)$) -- node[above, midway] {\parbox{5em}{\small\centering $\textit{up}(T)$\\\footnotesize $k$}} ($(x7) + (-0.3,0.2)$)  -- ($(x7) + (-0.3,0)$);
\draw[ultra thin] ($(x8) + (-0.2,0)$) -- ($(x8) + (-0.2,0.2)$) -- node[above, midway] {\parbox{5em}{\small\centering socket\\\footnotesize $k$}} ($(x9) + (-0.3,0.2)$)  -- ($(x9) + (-0.3,0)$);
\draw[ultra thin] ($(x10) + (-0.2,0)$) -- ($(x10) + (-0.2,0.2)$) -- node[above, midway] {\parbox{6em}{\small\centering {\footnotesize inverted} $\textit{right}(T)$\\\footnotesize $k$}} ($(x11) + (-0.3,0.2)$)  -- ($(x11) + (-0.3,0)$);
\end{tikzpicture}
\caption{Representing a tile using an $S$-successor.}\label{fig:undec2:tile}
\end{figure}

Thus, if the concept inclusions above hold, a tile---denote it by $T_{00}$---is placed on $z_0$ at the moment $-(2k+1)$, or, equivalently, $T_{00}$ is placed on an $S$-successor $w$ of $z_0$. The following concept inclusions will ensure then that the tiling is extended properly along both axes:
\begin{align}
\label{eq:undec2:grid-gen}
 \exists R^- \sqcap \Lbox
  \neg \exists R^- &\sqsubseteq \Rbox (\exists S \sqcup \exists R \sqcup \exists R^-),\\
 \exists R^- \sqcap \Lbox
  \neg \exists R^- &\sqsubseteq \Lbox \neg\exists R,\\
  \label{eq:undec2:V0disjR}
  \exists V_0 \sqcap\exists R &\sqsubseteq \bot,\\
  \exists V_0 \sqcap\exists R^- &\sqsubseteq \bot,\\
   \label{eq:undec-fp-role-disj1} 
    \exists S \sqcap  \exists R & \sqsubseteq \bot,\\
   \label{eq:undec-fp-role-disj2} 
    \exists S \sqcap \exists R^- & \sqsubseteq \bot,\\
    \label{eq:undec-fp-role-disj3} 
    \exists R \sqcap \exists R^- & \sqsubseteq \bot.
\end{align}
Indeed, consider the elements $z_0$ and $w$ with the tile $T_{00}$ placed on $z_0$ at $-(2k+1)$. Then $w$ has gaps (i.e., no incoming $S$-arrows) at moments 0, 
$\textit{down}(T_{00})$, $k+1$, $k+1+\textit{up}(T_{00})$, $k$ gaps from $2k+3$ to $3k +
2$ and $k-1$ gaps from $3k+4$ to $4k+3$ (one of the positions is not a gap because of the inverted
representation of $\textit{right}(T_{00})$). By~\eqref{eq:undec2:grid-gen}, each of those positions on $z_0$ must be filled either by an outgoing $S$-arrow, or by an incoming $R$-arrow, or by an outgoing $R$-arrow. Consider now what happens in these positions (see Fig.~\ref{fig:undec2:model}).

\begin{figure}[t]
\centering
\begin{tikzpicture}[>=latex,xscale=1.2,
	block/.style={rectangle,draw=black,rounded corners=0.4mm,inner ysep=0.4mm,inner xsep=2.5mm},
	point/.style={rectangle,draw=black,rounded corners=0.4mm,inner sep=0.4mm}]
\draw[ultra thin,dashed] (-0.3,-2) -- +(0,5);
\node at (-0.3,-2.2) {\scriptsize $-(2k\!+\!1)$};
\draw[ultra thin,dashed] (0.9,-2) -- +(0,5);
\node at (0.9,-2.2) {\scriptsize 0};
\draw[ultra thin,dashed] (1.5,-2) -- +(0,5);
\node at (1.5,-2.2) {\scriptsize $k\!+\!1$};
\draw[ultra thin,dashed] (3.3,-2) -- +(0,5);
\node at (3.3,-2.2) {\scriptsize $4k\!+\!4$};
\draw[densely dotted] (-0.5,-1.8) -- +(10,0);
\node at (-0.7,-1.8) {\footnotesize $a$};
\draw[densely dotted] (-0.5,-0.1) -- +(10,0);
\node at (-0.7,-0.1) {\footnotesize $z_0$};
\draw[densely dotted] (-0.5,1.1) -- +(10,0);
\node at (-0.7,1.1) {\footnotesize $z_1$};
\draw[densely dotted] (-0.5,2.3) -- +(10,0);
\node at (-0.7,2.3) {\footnotesize $z_2$};
\foreach \x/\y/\n in {0/0/00,2.4/0.15/10,4.8/0.3/20,%
                              0.6/1.2/01,3/1.35/11,5.4/1.5/21,%
                              1.2/2.4/02,3.6/2.55/12,6/2.7/22} {	
\node[block, fill=black, label=above:{\scriptsize P}] at (\x,\y) {};
\node[point, fill=white] at (\x+0.3,\y) {};
\node[block, fill=black!60, label=above:{\scriptsize L}] at (\x+0.6,\y) {};
\node[point, fill=white]  at (\x+0.6,\y) {};
\node[point, fill=white] (bd\n) at (\x+0.9,\y) {};
\node[block, fill=black!60, label=above:{\scriptsize D}] (d\n) at (\x+1.2,\y) {};
\node[point, fill=white]  at (\x+1.2,\y) {};
\node[point, fill=white] (ad\n) at (\x+1.5,\y) {};
\node[block, fill=black!60, label=below:{\scriptsize U}] (u\n) at (\x+1.8,\y) {};
\node[point, fill=white] at (\x+1.8,\y) {};
\node[point, fill=black] at (\x+2.1,\y) {};
\node[block, fill=white, label=below:{\scriptsize S}] at (\x+2.4,\y) {};
\node[point, fill=black] at (\x+2.7,\y) {};
\node[block, fill=white, label=below:{\scriptsize R}] at (\x+3,\y) {};
\node[point, fill=black]  at (\x+3,0\y) {};
\node[point, fill=white] (a\n) at (\x+3.3,0\y) {};
}
\begin{scope}[thick]
\draw[->] (u00) -- (d01);
\draw[->] (u10) -- (d11);
\draw[->] (u20) -- (d21);
\draw[->] (u01) -- (d02);
\draw[->] (u11) -- (d12);
\draw[->] (u21) -- (d22);
\draw[->] (ad00) -- (bd01);
\draw[->] (ad10) -- (a01);
\draw[->] (ad20) -- (a11);
\draw[->] (ad01) -- (bd02);
\draw[->] (ad11) -- (a02);
\draw[->] (ad21) -- (a12);
\end{scope}
\foreach \x/\y/\n in {1.2/-1.2/0,3.6/-1.05/1,6/-0.9/2} {
\node[point, fill=black] (v0\n) at (\x-0.3,\y) {};
\node[block, fill=black] (v0r\n) at (\x,\y) {};
\node[point, fill=white] at (\x+0.3,\y) {};
\node[block, fill=white] at (\x+0.6,\y) {};
\node[point, fill=black] at (\x+0.9,\y) {};
\node[block, inner xsep=6mm,fill=white] at (\x+1.5,\y) {};
\node[point, fill=white] at (\x-0.3,\y-0.5) {};
\node[block, fill=white] at (\x,\y-0.5) {};
\node[point, fill=black] at (\x+0.3,\y-0.5) {};
\node[block, fill=black] at (\x+0.6,\y-0.5) {};
\node[point, fill=white] at (\x+0.9,\y-0.5) {};
\node[block, inner xsep=6mm,fill=black] at (\x+1.5,\y-0.5) {};
}
\begin{scope}[thick]
\draw[->]  (v00) -- (bd00);
\draw[->]  (v01) -- (a00);
\draw[->]  (v02) -- (a10);
\draw[->]  (v0r0) -- (d00);
\draw[->]  (v0r1) -- (d10);
\draw[->]  (v0r2) -- (d20);
\end{scope}
\node at (0.6,-1.1) {\scriptsize $V_1$};
\node at (0.6,-1.6) {\scriptsize $V_0$};
\end{tikzpicture}
\caption{The structure of a model of $\K_\TT$.}\label{fig:undec2:model}
\end{figure}
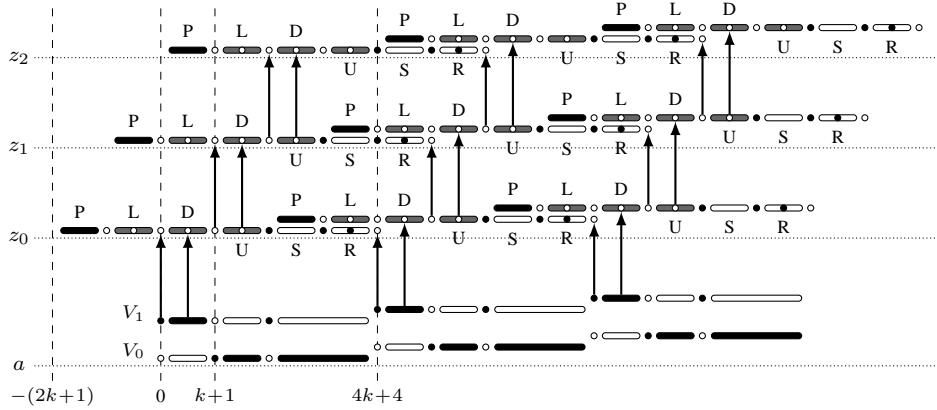

\begin{enumerate}
\item We know that there is an incoming $R$-arrow at 0 (i.e., $z_0$ is an instance of $\exists R^-$), and so, by~\eqref{eq:undec-fp-role-disj2} and \eqref{eq:undec-fp-role-disj3}, $z_0$ cannot be an instance of $\exists S$ and $\exists R$ at 0. 
\item The position at $\textit{down}(T_{00})$ is filled by an incoming $R$-arrow using the following concept inclusions (by~\eqref{eq:undec-fp-role-disj2}, the incoming $R$-arrow can only appear at $\textit{down}(T_{00})$):
\begin{align}
\label{eq:undec-fp-start-begin}
A & \sqsubseteq \bigsqcup_{1 \leq i \leq k} \Rdiamond^{=i}\, {\it init\text{-}bot},\\
 \label{eq:undec-fp-start-end} 
 {\it init\text{-}bot} &\sqsubseteq \exists R.
\end{align} 
\item The position at $k+1$ cannot be filled by an outgoing $S$-arrow  because that would trigger a new tile sequence, which would require $k$ $S$-arrows of the P-section, which is impossible due to~\eqref{eq:undec2:Sfunc}. Next, as we observed above, $a^\I$ belongs to $\exists V_0$ at all moments $i$ with $k+1 \leq i \leq 2k+1$, and so, by~\eqref{eq:undec:R2} and~\eqref{eq:undec2:V0disjR}, $z_0$ cannot have an incoming $R$-arrow at moment $k+1$. Thus, the position at $k+1$ must be filled by an outgoing $R$-arrow. 
Thus, there is an $R$-successor $z_1$ of $z_0$, which, by~\eqref{eq:undec:R2}, implies that $z_1$ has no incoming $R$-arrows before $k+1$. Then, by~\eqref{eq:undec-fp-row-begin}--\eqref{eq:undec-fp-role-pos-end}, there will be a tile placed on $z_1$ at $-k = (k+1)-(2k+1)$. 
\item Similarly, the position at $k+1+\textit{up}(T)$ must be filled by an outgoing $R$-arrow, which ensures that the $\textit{down}$-colour of the tile placed on $z_1$ matches the $\textit{up}$-colour of the tile on $z_0$.
\item The $k$ positions of the S-section from $2k+3$ to $3k + 2$ cannot be filled by an incoming $R$-arrow. On the other hand, the tile placed on $z_1$ has its $\textit{up}$-colour encoded in this range, and so an outgoing $R$-arrow cannot fill \emph{all} these gaps either (as $k > 1$). So, $z_0$ has another $S$-successor $x_1$ in at least one of the moments of the $S$-section. By~\eqref{eq:undec-fp-row-end},~\eqref{eq:undec-fp-role-funct},  $x_1$ does not belong to $\exists S^-$ before $-(2k+1)$. By~\eqref{eq:undec-fp-tile-placed}, a tile is placed on $x_1$ between $-(2k+1)$ and $3k+2$, but, by~\eqref{eq:undec2:Sfunc} and because the tile requires $x_1$ to be the $S$-successor for $k$ consecutive moments of the P-section, it is only possible at $2k+2$. Moreover, since the $\textit{left}$- and $\textit{right}$-sections of these tile sequences overlap on $z_0$, by~\eqref{eq:undec2:Sfunc}, the adjacent colours of these two tiles match. This ensures that the $k-1$ gaps of the inverted representation of the $\textit{right}$-colour of the first tile are also filled.
\end{enumerate}
Let $\K_\TT$ be the KB containing all the concept inclusions above and $\A$. 
If $\I$ is a model of $\K_\TT$ then the process described above generates a sequence $z_0,z_1,\dots$ of domain elements such that each $z_i$ has a tile placed on it at every $4k+4$ moments of time; moreover, the $R$-arrows form a proper $\N\times\N$-grid and the adjacent colours of the tiles match. We only note that the gaps at positions in the $\textit{down}$-section do not need a special treatment after the very first tile $T_{00}$ at $(0,0)$ because, for each $z_i$, the sequence of tiles on $z_{i+1}$ will have their $\textit{left}$- and $\textit{right}$-sections, with no gap to be filled by an incoming $R$-arrow; thus, the only available choice for tiles on $z_i$ is $\exists R$. 

We have proved that if $\K_\TT$ is satisfiable then $\TT$ tiles $\N\times\N$. The converse implication is shown using the satisfying interpretation illustrated in Fig.~\ref{fig:undec2:model}.
\end{proof}

\subsection{Undirected Temporal Operators: Decidability and NP-completeness}

If we disallow the `previous time,' `next time,' `always in the past' and `always in the future' operators in the language of concept inclusions and replace them with `always' ($\SVbox$) then reasoning in the resulting logic \TurDLbn{} becomes decidable and  \NP-complete. 

Obviously, the problem is \NP-hard (because of the underlying DL). However, rather surprisingly, the interaction of temporalised roles and number restrictions is yet another source of nondeterminism, which is exhibited already by very simple TBoxes with concept inclusions in the \textit{core} fragment. The following example illustrates this point and gives a glimpse of the difficulties we shall face in the proof of the \NP{} upper bound by means of the quasimodel technique: unlike other quasimodel proofs~\cite{GKWZ03}, where only types of domain elements need to be guessed, here we also have to guess relations between ABox individuals at all relevant moments of time.

\begin{example}\label{ex:mod-role}
Let $\T  =   \{\,A \sqsubseteq \mathop{\geq 5} R, \ \mathop{\geq 7} \SVdiamond R \sqsubseteq \bot \,\}$  and 
\begin{equation*}
 \A =  \{\, \Rnext A(a),\  R(a,b_1), \ R(a,b_2),\  R(a,b_3), \ \Rnext R(a,b_1) \,\}.
\end{equation*}
The second concept inclusion of the TBox implies that, in any every model $\I$ of $(\T,\A)$, $a$ cannot have more than 6 $(\SVdiamond R)^\I$-successors in total; thus, it can only have $b_1$, $b_2$, $b_3$ and up to 3 $R^\I$-successors from outside the ABox. At moment 1, however, $a$ must have at least 5 $R^\I$-successors, including $b_1$. Thus, one of the $(\SVdiamond R)^\I$-successors in the ABox has to be re-used:  we have either $\I\models \Rnext R(a,b_2)$ or $\I\models \Rnext R(a,b_3)$.

Consider now 
$\T =  \{\, \mathop{\geq 6} \SVdiamond R \sqsubseteq \bot, \ \top\sqsubseteq \mathop{\geq 4} \SVbox R \,\}$ and $\A = \{\, R(a,b_1), R(a,b_2) \,\}$. Then, in every model $\I$ of $(\T,\A)$, either
$\I\models\SVbox R(a,b_1)$ or $\I\models\SVbox
R(a,b_2)$. 
\end{example}

\begin{theorem}\label{thm:R:NP}
The satisfiability problem for \TurDLbn{} KBs is \NP-complete.
\end{theorem}
\begin{proof}
Let $\K=(\T,\A)$ be a \TurDLbn{} KB. In what follows, given an interpretation $\I$,  we write $(\mathop{\geq q}\SVbox R)^\I$ and $(\mathop{\geq q}\SVdiamond R)^\I$ instead of $(\mathop{\geq q}\SVbox R)^{\I(n)}$ and $(\mathop{\geq q}\SVdiamond R)^{\I(n)}$, for $n\in \Z$, because temporalised roles are time-invariant.
As before, $\ob$ denotes the set of all object names occurring in $\A$ (we assume $|\ob| \geq 1$) and $\role$ the set of role names in $\K$ and their inverses. Let $\QT\subseteq \N$ be the set (cf.\ p.~\pageref{def:QK}) comprised of 1 and all $q$ such that one of $\mathop{\geq q} \SVbox R$, $\mathop{\geq q} R$ or $\mathop{\geq q} \SVdiamond R$ occurs in $\T$ and let $\QA$ be the set of all natural numbers from $0$ to $|\ob|$.  Let $q_\K = \max (\QT\cup\QA) + 1$. 
First, we show that it is enough to consider interpretations with the number of $\SVbox R$-successors bounded by $q_\K -1$ (see Appendix~\ref{app:NP} for a proof):
\begin{lemma}\label{lem:finite:box}
Every satisfiable  \TurDLbn{} KB $\K$ can be satisfied in an interpretation $\I$ with $(\mathop{\geq q_\K} \SVbox R)^{\I} = \emptyset$, for each $R\in\role$.
\end{lemma}

Next, we define the notion of quasimodel. Let $Q \supseteq \QT\cup\QA$ be a set of natural numbers with $\max Q = q_\K - 1$.
We assume that the usual order on the natural numbers in $\Qs = Q \cup \{\omega \}$ is extended to $\omega$, which is the greatest element: $0 < 1 < \dots < q_\K  - 1 < \omega$. Let $\Sigma$ consist of  the following concepts and their negations:  subconcepts of concepts occurring in $\K$  and
$\mathop{\geq q} \SVbox R$, $\mathop{\geq q} R$ and $\mathop{\geq q} \SVdiamond R$, for all $R\in\role$ and $q\in \Qs$.
A \emph{$\Sigma$-type} $\type$ is a maximal consistent subset  of  $\Sigma$:
\begin{akrzlist}
\item[t$_1$] $C\in \type$ iff $\neg C\notin \type$, for each
  $C\in\Sigma$,
\item[t$_2$] $C_1 \sqcap C_2\in\type$ iff $C_1,C_2\in \type$, for each
  $C_1\sqcap C_2\in\Sigma$,
\item[t$_3$] if $\SVbox C\in\type$ then $C\in\type$, for each $\SVbox C\in\Sigma$,
\item[t$_4$] if $\mathop{\geq q} R \in \type$ then
  $\mathop{\geq q'} R\in\type$, for each $\mathop{\geq q'} R\in\Sigma$ with $q > q'$ (similarly for $\SVbox R$ and $\SVdiamond R$),
\item[t$_5$] $\mathop{\geq 0} \SVbox R\in \type$ but $\mathop{\geq \omega} \SVbox R\notin \type$, for each role $R$,
\item[t$_6$]  if $\mathop{\geq q} \SVbox R\in\type$ then $\mathop{\geq q}
  R\in\type$ and if $\mathop{\geq q} R\in\type$ then $\mathop{\geq q} \SVdiamond R\in\type$, for each role $R$,
\item[t$_7$] if $\mathop{\geq q}
  \SVdiamond R\in \type$ then $\mathop{\geq q} \SVbox R\in\type$, for  each \emph{rigid} role $R$. 
\end{akrzlist}
Denote by $Z_\A$\label{za} the set of all integers $k$ such that at least one of $\nxt^k
A(a)$, $\nxt^k \neg A(a)$, $\nxt^k S(a,b)$ or $\nxt^k \neg
S(a,b)$ occurs in
$\A$, and let $Z \supseteq Z_\A$ be a finite set of
integers. 

By a \emph{$(Z,\Sigma)$-run}  (or simply run, if $Z$ and $\Sigma$ are clear from the context) we mean a function $r$ from $Z$ to the set of $\Sigma$-types. Concepts of the form $\SVbox C$, $\mathop{\geq q} \SVbox R$, $\mathop{\geq q} \SVdiamond R$ and their negations are called \emph{rigid}. A run $r$ is said to be \emph{coherent} if the following holds for each rigid concept $D$ in $\Sigma$:
\begin{akrzlist}
\item[r$_1$] if $D\in r(k_0)$, for some $k_0\in Z$, then $D \in r(k)$ for all $k \in Z$.
\end{akrzlist}
In the following, the runs are assumed to be coherent and so, for rigid concepts $D$, we can write $D\in r$ in place of $D\in r(k)$, for some (all) $k\in Z$. The \emph{required $R$-rank of $r$ at  $k\in Z$} and the \emph{required} $\SVbox R$- and $\SVdiamond R$-\emph{ranks of $r$} are defined by taking
\begin{align*}
\Rr{R}{k}{r} & =   \max \bigl\{ q\in \Qs \mid \mathop{\geq q} R \in r(k)\bigr\}, & 
\begin{array}{rl}
\Rb{R}{r} & =   \max \bigl\{ q\in \Qs \mid \mathop{\geq q} \SVbox R \in r \bigr\}, \\ 
\Rd{R}{r} & =   \max \bigl\{ q\in \Qs \mid \mathop{\geq q} \SVdiamond R \in r \bigr\}. 
\end{array}
\end{align*}
By the definition of $\Sigma$-types, $\Rb{R}{r} \leq \Rr{R}{k}{r} \leq \Rd{R}{r}$. Moreover, if $R$ is rigid then  
$\Rb{R}{r} = \Rd{R}{r} < \omega$. For flexible roles, however, the inequalities may be strict. 
A run $r$ is  \emph{saturated} if the following hold:
\begin{akrzlist}\itemsep=2pt
\item[r$_2$] for every flexible role $R\in\role$, if $\Rb{R}{r} < \Rd{R}{r}$ then 
\begin{itemize}
\item[--] there is $k_0\in Z$ with $ \Rb{R}{r} < \Rr{R}{k_0}{r}$, and  
\item[--] if additionally $\Rd{R}{r} < \omega$, then there is $k_1\in Z$ with $\Rr{R}{k_1}{r} < \Rd{R}{r}$;
\end{itemize}
\item[r$_3$] for every $\SVbox C\notin r$, there is $k_0\in Z$ with $C\notin r(k_0)$. 
\end{akrzlist}
Finally, we call $\extA$ a \emph{consistent $Z$-extension of $\A$} if $\extA$ extends  $\A$ with assertions of the form $\nxt^k S(a,b)$, for $k\in Z$ and $a,b\in\ob$, 
such that $\nxt^k S(a,b)\notin\extA$ for all $\nxt^k \neg S(a,b)\in\A$. 
Example~\ref{ex:mod-role} shows that, given an ABox, we have first to guess such an extension to describe a quasimodel. More precisely, we have to count the number of $R$-successors \emph{among the ABox individuals} in $\extA$. To this end, define the following sets, for $a\in\ob$, $R\in\role$ and $k \in Z$:
\begin{align*}
\extA^{R,k}_{a} & =   \{ b \mid \nxt^k R(a,b)\in\extA \}, && 
\begin{array}{rl}\extA^{\Box R}_{a} & =   \{ b \mid \nxt^n R(a,b)\in\extA, \text{ for all } n\in Z \}, \\ 
\extA^{\Diamond R}_{a} & =   \{ b \mid \nxt^n R(a,b)\in\extA, \text{ for some } n\in Z \},\end{array}
\end{align*}
where, as on p.~\pageref{def:invABox}, we assume that $\extA$ contains $\nxt^n
S^-(b,a)$ whenever it contains $\nxt^n S(a,b)$. 
We say that a $(Z,\Sigma)$-run $r$ is \emph{$a$-faithful for $\extA$} if
\begin{akrzlist}\itemsep=2pt
\item[r$_4$] $A\in r(k)$, for all $\nxt^k A(a)\in\extA$,  and $\neg A\in r(k)$, for all $\nxt^k \neg A(a)\in \extA$;
\item[r$_5$] $0 \ \ \leq \ \ \Rb{R}{r} - \Nb{\extA}{R}{a} \ \ \leq  \ \ \Rr{R}{k}{r} - \Nr{\extA}{R}{k}{a}  \ \ \leq  \ \ \Rd{R}{r} - \Nd{\extA}{R}{a}$, for all $R\in\role$  and  $k\in Z$;\footnote{We assume that $\omega -n = \omega$, for any natural number $n$.}
\item[r$_6$] for all $R\in\role$, if $\Rb{R}{r}
- \Nb{\extA}{R}{a}  \ \ <  \ \ \Rd{R}{r} - \Nd{\extA}{R}{a}$ then
\begin{itemize}
\item[--] there is $k_0\in Z$ with $\Rb{R}{r}  - \Nb{\extA}{R}{a}  \ \ < \ \ \Rr{R}{k_0}{r} - \Nr{\extA}{R}{k_0}{a}$, and 
\item[--] if additionally $\Rd{R}{r} < \omega$, then there is $k_1\in Z$ with $\Rr{R}{k_1}{r} - \Nr{\extA}{R}{k_1}{a}<   \Rd{R}{r} - \Nd{\extA}{R}{a}$.
\end{itemize}
\end{akrzlist}
Condition~\textbf{(r$_5$)} says that the number of $R$-successors in the ABox extension $\extA$ does not exceed the number of required $R$-successors: in view of  $|\extA^{R,k}_{a}|\in \QA\subseteq Q$, we have $\mathop{\geq q} R\in r(k)$ for $q = \Nr{\extA}{R}{k}{a}$ (and similarly for $\SVbox R$ and $\SVdiamond R$). Condition~\textbf{(r$_5$)}  also guarantees that the number of required $R$-successors that are \emph{not ABox individuals} is consistent for $\SVbox R$, $R$ and $\SVdiamond R$. In particular, since $\Nb{\extA}{R}{a} \leq \Nd{\extA}{R}{a}$, it follows that $\Rb{R}{r} = \Rd{R}{r}$ implies $\Nb{\extA}{R}{a} = \Nd{\extA}{R}{a}$, and so $\nxt^n R(a,b) \in \extA$, for all $n\in Z$, whenever  $\nxt^{\smash{k}} R(a,b) \in \extA$ for some $k\in Z$. Finally,  \textbf{(r$_6$)} is an adaptation of the notion of \emph{saturated runs} for the case of ABox individuals.

A \emph{quasimodel $\mathfrak{Q}$ for $\K$} is
a quadruple $(Q,Z,\mathfrak{R},\extA)$, where $Q$ and $Z$ are finite sets of integers  extending $\QT\cup\QA$ and $Z_\A$, respectively, $\mathfrak{R}$ is a set
of coherent and saturated $(Z,\Sigma)$-runs (for $\Sigma$ defined on the basis of $\Qs$)
and $\extA$ is a consistent $Z$-extension of $\A$ satisfying the
following conditions:
\begin{akrzlist}\itemsep=2pt
\item[Q$_1$] for all $r\in\mathfrak{R}$, $k\in Z$ and $C_1\sqsubseteq C_2 \in \T$, if $C_1\in r(k)$ then $C_2\in r(k)$;
\item[Q$_2$] for all $a\in\ob$, there is a run $r_a\in\mathfrak{R}$  that is $a$-faithful for $\extA$;
\item[Q$_3$] for all $R\in\role$, if there is $r\in\mathfrak{R}$ with $\Rb{R}{r} \geq 1$ then there is $r'\in\mathfrak{R}$ with $\Rb{\,\inv{R}}{r'} \geq 1$;
\item[Q$_4$]   for all $R\in\role$, if there is  $r\in\mathfrak{R}$ with $\Rb{R}{r} < \Rd{R}{r}$ then there exists $r'\in\mathfrak{R}$ with $\Rb{\,\inv{R}}{r'} < \Rd{\,\inv{R}}{r'}$.
\end{akrzlist}
Condition~\textbf{(Q$_1$)} ensures that all runs are consistent with the concept inclusions in $\T$ and~\textbf{(Q$_2$)} that there are runs for all ABox individuals; \textbf{(Q$_3$)} guarantees that a $\SVbox R$-successor can be found whenever required and \textbf{(Q$_4$)} provides $R$- (and thus $\SVdiamond R$-) successors whenever required. The following lemma states that the notion of quasimodel is adequate for checking satisfiability of \TurDLbn{} KBs:
\begin{lemma}\label{lemma:NP:quasimodel}  
A \TurDLbn{} KB $\K$ is satisfiable if and only if there is a quasimodel $\mathfrak{Q}$ for $\K$ such that the size of $\mathfrak{Q}$ is polynomial in the size of \K.
\end{lemma}

\begin{proof}
$(\Rightarrow)$ Let $\I$ be a model of
  $\K$. By Lemma~\ref{lem:finite:box}, we may assume that, for each $R\in\role$, the
  number of $\SVbox R$-successors of any element in $\I$ does not exceed
  $q_\K - 1$.
We construct a polynomial-size quasimodel $\mathfrak{Q} = (Q,Z,\mathfrak{R},\extA)$ for $\K$.
First, we select a set $D$ of elements of $\Delta^\I$ that serve as prototypes for runs in $\mathfrak{R}$: each $u \in D$ will give rise to a run $r_u$ in $\mathfrak{R}$ (after the set $Z$ of time instants has been fixed). Set $D_0 = \{ a^\I \mid a\in \ob \}$ and then proceed by induction: if $D_m$ has already been defined then we construct $D_{m+1}$ by extending $D_m$ as follows:
\begin{akrzitemize}
\item[--] if $D_m\cap (\exists \SVbox R)^{\I} \ne \emptyset$ but $D_m\cap (\exists \SVbox \inv{R})^{\I} = \emptyset$ then we add some $u\in (\exists \SVbox \inv{R})^\I$;
\item[--] if there is $q$ with $D_m\cap \bigl((\mathop{\geq q} \SVdiamond R)^\I\setminus (\mathop{\geq q} \SVbox R)^\I\bigr) \ne\emptyset$  but there is no $q'$ with $D_m\cap \bigl((\mathop{\geq q'} \SVdiamond \inv{R})^\I\setminus (\mathop{\geq q'} \SVbox \inv{R})^\I\bigr) \ne \emptyset$ then add $u\in (\mathop{\geq q''} \SVdiamond \inv{R})^\I\setminus (\mathop{\geq q''} \SVbox \inv{R})^\I$, for some $q''$ (recall that, by Lemma~\ref{lem:finite:box}, we assume that  $q$, $q'$ and $q''$ do not exceed $q_\K$).
\end{akrzitemize}
When neither rule is applicable to $D_m$, stop and set $D = D_m$. Clearly, $|D| \leq |\ob| + 2|\role|$.  

For each $u\in D$, let
\begin{align*}
\rho^{\Box R}_u & =  \max \bigl\{ q < q_\K \mid u \in
  (\mathop{\geq q} \SVbox R)^\I \bigr\},  \\
\rho^{\Diamond R}_u & = \begin{cases}\omega, & \text{ if }u \in
  (\mathop{\geq q_{\K}} \SVdiamond R)^\I,\\ 
  \max \bigl\{ q < q_\K \mid u \in
  (\mathop{\geq q} \SVdiamond R)^\I \bigr\}, & \text{ otherwise}.\end{cases}
\end{align*}
We now choose time instants to be included in the runs $\mathfrak{R}$. Let $Z$ extend $Z_\A$ with the following: 
\begin{akrzlist}\itemsep=2pt
\item[$Z_0$] for any $u\in D$ and $\SVbox C\in\Sigma$ such that $u\notin(\SVbox C)^{\I}$, 
we add some $n\in \Z$ with $u\notin C^{\I(n)}$; 
\item[$Z_1$] for any $u\in D$ and $R\in\role$ such that $\rho^{\Box R}_u < \rho^{\Diamond R}_u$, we add 
\begin{itemize}    
\item[--]  some $n_0\in\Z$ with  $u\in (\mathop{\geq  (\rho^{\Box R}_u + 1)} R)^{\I(n_0)}$ and
\item[--] if additionally $\rho^{\Diamond R}_u < \omega$, some $n_1\in\Z$ with $u\notin(\mathop{\geq \rho^{\Diamond R}_u} R)^{\I(n_1)}$; 
\end{itemize}

\item[$Z_2$] for any $a,b\in\ob$ and $R\in\role$ such that $(a^\I,b^\I)\in(\SVdiamond R)^\I$, we add
\begin{itemize}
\item[--]  some $n_0\in\Z$ with $(a^\I,b^\I)\in R^{\I(n_0)}$ and 
\item[--] if $(a^\I,b^\I)\notin(\SVbox R)^\I$, some $n_1\in\Z$ with $(a^\I,b^\I)\notin R^{\I(n_1)}$;
\end{itemize}

\item[$Z_3$] for any $a\in \ob$ and $R\in\role$ such that $\rho^{\Box R}_{a^\I} - |\Irb{\I}{R}{a}| < \rho^{\Diamond R}_{a^\I} - |\Ird{\I}{R}{a}|$, we add 
\begin{itemize}
\item[--]  some $n_0\in\Z$ with $a^\I\in (\mathop{\geq (q_0 +1)} R)^{\I(n_0)}$, for $q_0 = \rho^{\Box R}_{a^\I} + (|\Ir{\I}{R}{n_0}{a}| -|\Irb{\I}{R}{a}|)$, and
\item[--] if $\rho^{\Diamond R}_{a^\I} < \omega$, some $n_1\in\Z$ with
$a^\I\notin (\mathop{\geq q_1} R)^{\I(n_1)}$, for $q_1 = \rho^{\Diamond R}_{a^\I} - (|\Ird{\I}{R}{a}| - |\Ir{\I}{R}{n_1}{a}|)$, 
\end{itemize}
where $\Ir{\I}{R}{k}{a} = \{ b\in\ob \mid (a^\I,b^\I)\in R^{\I(k)}\}$ and $\Irb{\I}{R}{a}$ and $\Ird{\I}{R}{a}$ are defined similarly.
\end{akrzlist}
Clearly, $|Z_0| \leq |D|\cdot |\K|$, $|Z_1| \leq 2 |D|\cdot|\role|$, $|Z_2| \leq 2 |\ob|^2\cdot|\role|$ and $|Z_3| \leq 2\cdot |\ob|\cdot|\role|$. Thus, $|Z| = O(|\K|^3)$.
The time instants in $Z_0$, $Z_1$ and $Z_2$ exist because $\I\models\K$. We now show that $n_0$ required in $Z_3$ also exists. Suppose, on the contrary, that $a^\I\notin (\mathop{\geq (q_0 + 1)} R)^{\I(n)}$, with $q_0$ as above, for all $n\in\Z$.  Then $a^\I$ has at most $(\rho^{\Box R}_{a^\I} + (|\Ir{\I}{R}{n}{a}| - |\Irb{\I}{R}{a}|))$-many $R$-successors, whence the number of non-ABox $R$-successors of $a^\I$ does not exceed $\rho^{\Box R}_{a^\I} - |\Irb{\I}{R}{a}|$. So, at all instants $n\in\Z$, every $R$-successor of $a^\I$ is either in $\ob$ or is in fact a $\SVbox R$-successor, contrary to $\rho^{\Box R}_{a^\I} - |\Irb{\I}{R}{a}| < \rho^{\Diamond R}_{a^\I} - |\Ird{\I}{R}{a}|$. Using a similar argument, one can show that $n_1$ required in $Z_3$  exists as well.
Having fixed $Z$, we define a consistent $Z$-extension $\extA$ of $\A$ by taking
\begin{equation*}
\extA \ \ =\ \ \A \ \ \cup \ \ \{ \ \nxt^k S(a,b) \ \ \mid \ \ (a^\I,b^\I) \in
 S^{\I(k)} \text{ and } k \in Z \ \}.
\end{equation*}
Let $Q$ be the set comprising $\QT$, $\QA$ and, for any $u\in D$ and $R\in\role$, the integers from 
\begin{equation*}
 \rho^{\Box R}_u, \ \ \rho^{\Diamond R}_u \ \ \text{ and } \ \ \max\bigl\{q < q_\K \mid u \in (\mathop{\geq q} R)^{\I(k)}\bigr\}, \text{ for } k \in Z.
\end{equation*}
By definition, $\max Q  = q_\K  - 1$ and  $|Q| \leq |\QT| + |\QA|  + |D| \cdot |\role| \cdot (2 + |Z|)$.
Let $\mathfrak{R}$ be the set of $(Z,\Sigma)$-runs $r_u$, for $u\in D$, defined by taking, for each $k\in Z$,
\begin{akrzitemize}
\item[--] $\mathop{\geq \omega} R \in r_u(k)$ iff $u \in (\mathop{\geq
    q_\K} R)^{\I(k)}$, and  $\mathop{\geq \omega} \SVdiamond R\in r_u(k)$ iff $u \in
  (\mathop{\geq q_\K} \SVdiamond R)^{\I}$,
\item[--] $C \in r_u(k)$ iff $u \in C^{\I(k)}$, for all other
  concepts $C\in\Sigma$.
\end{akrzitemize}
Since $\I\models\K$ and $\I$ is as in Lemma~\ref{lem:finite:box}, each $r_u(k)$ is a $\Sigma$-type.
Each $r_u\in\mathfrak{R}$ is a coherent and saturated $(Z,\Sigma)$-run: \textbf{(r$_1$)} holds because $\I\models\K$; \textbf{(r$_3$)} and~\textbf{(r$_2$)} are due to $Z_0\subseteq Z$ and $Z_1\subseteq Z$, respectively. Since $\I\models\extA$, each run $r_{a^\I}$ is $a$-faithful for $\extA$. Indeed,~\textbf{(r$_4$)} is due to $Z_\A\subseteq Z$. To show~\textbf{(r$_5$)} and~\textbf{(r$_6$)}, observe that, by definition, $\Nr{\extA}{R}{k}{a} = |\Ir{\I}{R}{k}{a}|$ and, since $Z_2\subseteq Z$, we also have $\Nb{\extA}{R}{a} = |\Irb{\I}{R}{a}|$ and $\Nd{\extA}{R}{a} = |\Ird{\I}{R}{a}|$; moreover,  $\Rb{R}{r_{a^\I}}, \Rr{R}{k}{r_{a^\I}}, \Rd{R}{r_{a^\I}} \in\Qs$ and, for each $q< q_\K$, we have $a^\I \in (\mathop{\geq q} R)^{\I(k)}$ iff $\Rr{R}{k}{r_{a^\I}} \geq q$ (and similarly for $\SVbox R$ and $\SVdiamond R$). Then~\textbf{(r$_5$)} follows from the choice of $\extA$ and~\textbf{(r$_6$)} from $Z_3\subseteq Z$.
We claim that $\mathfrak{Q} = (Q,Z,\mathfrak{R},\extA)$ is a quasimodel for $\K$:  \textbf{(Q$_1$)} holds by definition and  \textbf{(Q$_2$)}--\textbf{(Q$_4$)} follow from  the choice of $D$.
Finally,  as $|\extA| \leq |\A| + |Z|\cdot|\ob|^2\cdot |\role|$ and
$|\mathfrak{R}| \leq |\ob| + 2 |\role|$, the quasimodel is of polynomial size.

\medskip

$(\Leftarrow)$ Let $\mathfrak{Q} = (Q,Z,\mathfrak{R},\extA)$
be a quasimodel for $\K$.  We construct an interpretation $\I$ satisfying
$\K$, which is based on some domain $\Delta^\I$ that will be
defined inductively as the union
\begin{equation*}
\Delta^\I = \bigcup\nolimits_{m\geq 0} \Delta_m,\qquad\text{where } \Delta_{m}\subseteq\Delta_{m+1}, \text{ for all } m \geq 0. 
\end{equation*}
Each set $\Delta_{m+1}$ ($m \geq 0$) is constructed by adding to $\Delta_m$ new domain elements that are copied from the runs in $\mathfrak{R}$; similarly to the proof of Theorem~\ref{lem:qtli-equisat} in Appendix~\ref{app:proof:qtli-equisat}, the function $\cp\colon \Delta^\I\to\mathfrak{R}$ keeps track of this process. In contrast to the proof of Theorem~\ref{lem:qtli-equisat}, however, the runs are defined on a finite set, $Z$, and so we need to multiply (and rearrange) the time instants of $Z$ when creating elements of $\Delta^\I$ from runs in $\mathfrak{R}$. To this end, for each $u\in\Delta^\I$, we define a function $\nu_u\colon\Z\to Z$ that maps each time instant $n\in\Z$ of $u\in\Delta^\I$ to its `origin' $\nu_u(n)\in Z$ on the run $\cp(u)$. Since the constructed interpretation $\I$ may contain infinite sequences of domain elements related by roles, we will need to ensure that each $\Sigma$-type of the run appears infinitely often along $\Z$ (note, however, that the actual order of time instants is important only for $Z_\A$, the instants of the ABox).

The interpretation of role names in $\I$ is constructed inductively along with the construction of the domain:  $S^{\I(n)} = \bigcup_{m\geq 0} S^{n,m}$, where $S^{n,m} \subseteq \Delta_m \times
\Delta_m$,  for $m \geq 0$. Given $m \geq 0$ and $u \in \Delta_m$,
we define the \emph{actual $S$-rank at moment $n\in\Z$} and the \emph{actual $\SVbox S$- and $\SVdiamond S$-ranks on step $m$}:
\begin{align*}
  \Tr{S}{n}{u}{m} & = \sharp \{u' \in \Delta_m \mid (u,u') \in S^{n,m}\}, &&
  \begin{array}{rl}
 \Tb{S}{u}{m} & = \textstyle\sharp \{u' \in \Delta_m \mid (u,u') \in S^{k,m}, \text{ for all } k\in \Z\}, \\
\Td{S}{u}{m} & = \textstyle\sharp \{u' \in \Delta_m \mid (u,u') \in S^{k,m}, \text{ for some } k\in\Z\}.
\end{array}
\end{align*}
The actual  $S^-$-, $\SVbox S^-$- and $\SVdiamond S^-$-ranks are defined similarly, with $(u,u')$ replaced by $(u',u)$. Let
\begin{equation*}
\Db{R}{u}{m} =  \Rb{R}{\cp(u)} - \Tb{R}{u}{m},\qquad \Dr{R}{n}{u}{m} =  \Rr{R}{\nu_u(n)}{\cp(u)}- \Tr{R}{n}{u}{m}\quad\text{ and }\quad \Dd{R}{u}{m} =  \Rd{R}{\cp(u)} - \Td{R}{u}{m}.
\end{equation*}
The inductive construction of the domain and sets $S^{n,m}$  will ensure that, for each $m \geq 0$, the following  holds for all $u\in\Delta_m\setminus\Delta_{m-1}$ (for convenience, we assume $\Delta_{-1}= \emptyset$):
\begin{akrzlist}
\item[fn] $\Td{R}{u}{m} < \omega$, for all  $R\in\role$;
\item[rn] 
$0 \ \ \leq \ \ \Db{R}{u}{m} \  \ \leq \ \ \Dr{R}{n}{u}{m} \ \ \leq \ \  \Dd{R}{u}{m}$, for all  $R\in\role$ and all $n\in\Z$;
\item[df] for all  $R\in\role$, if $\Db{R}{u}{m} \  \ < \ \  \Dd{R}{u}{m}$ then 
\begin{itemize}
\item[--] $\Db{R}{u}{m}  <  \Dr{R}{n}{u}{m} $, for infinitely many $n\in\Z$, and 
\item[--] if additionally $\Dd{R}{u}{m} < \omega$, then $\Dr{R}{n}{u}{m} <   \Dd{R}{u}{m}$, for infinitely many $n\in\Z$.
\end{itemize}
\end{akrzlist}
Note that, by~\textbf{(fn)}, $\Dd{R}{u}{m}$ and the $\Dr{R}{n}{u}{m}$ are well-defined and $\Dd{R}{u}{m} = \omega$ is just in case
$\Rd{R}{\cp(u)}=\omega$. 

For the basis of induction  ($m=0$), set $\Delta_0 = \ob$ and $a^\I=a$, for each $a \in\ob$.  By~\textbf{(Q$_2$)}, for each $a\in\Delta_0$,  there is a run $r_a\in\mathfrak{R}$ that is $a$-faithful for $\extA$. So, set $\cp(a) = r_a$ and take $\nu_a = \nu$, for some fixed function $\nu \colon\Z\to Z$ such that $\nu(k)=k$ and $\nu^{-1}(k)$ is infinite, for each $k\in Z$. 
For every role name $S$, let 
\begin{equation}\label{eq:NP:relation:basis}
  S^{n,0}= \bigl\{(a,b)\in\Delta_0\times\Delta_0 \mid \nxt^{\nu(n)} S(a,b) \in \extA \bigr\},\qquad\text{ for }n \in\Z.
\end{equation}
By definition, $\Tb{R}{a}{0} = \Nb{\extA}{R}{a}$, $ \Tr{R}{n}{a}{0} =
\Nr{\extA}{R}{\nu(n)}{a}$, for all $n\in\Z$, and $\Td{R}{a}{0} =
\Nd{\extA}{R}{a}$.  For each $a\in\Delta_0$,  \textbf{(fn)} is by construction, \textbf{(rn)} is immediate from~\textbf{(r$_5$)} and~\textbf{(df)} follows from~\textbf{(r$_6$)} and the fact
that $\nu^{-1}(k)$ is infinite, for each $k\in Z$.

Assuming that $\Delta_m$ and the $S^{n,m}$ have been defined and~\textbf{(fn)}, \textbf{(rn)} and \textbf{(df)} hold for some $m \geq 0$, we construct $\Delta_{m+1}$ and the $S^{n,m+1}$ and show that the   properties also hold for $m+1$.  By~\textbf{(rn)}, for all $u\in\Delta_m$ and $R\in\role$, we have $\Db{R}{u}{m} \geq 0$, $\Dr{R}{n}{u}{m}\geq 0$,  for all $n\in\Z$, and $\Dd{R}{u}{m}\geq 0$.
If these inequalities are actually equalities then we are done. However, in general this is not the case as there may be `defective' elements whose actual rank is smaller than the required rank. Consider the following four sets of defects in $S^{n,m}$, for $R = S$ and $R = S^-$:
\begin{equation*}
  \LambdaB{R}{m} ~=~ \bigl\{ u \in \Delta_m\setminus \Delta_{m-1} \mid 0 < \Db{R}{u}{m} \bigr\}\quad\text{ and }\quad
  \LambdaD{R}{m} ~=~ \bigl\{ u \in \Delta_m\setminus \Delta_{m-1} \mid \Db{R}{u}{m} < \Dd{R}{u}{m} \bigr\}.
\end{equation*}
The purpose of $\LambdaB{R}{m}$ is to identify 
elements $u\in \Delta_m\setminus \Delta_{m-1}$ that should have
$\Rb{R}{\cp(u)}$-many distinct $\SVbox R$-arrows
(according to $\mathfrak Q$), but some arrows are still
missing (only $ \Tb{R}{u}{m}$ arrows exist in $\Delta_m$). The purpose of
$\LambdaD{R}{m}$ is to identify elements $u$ that should have
$\Rd{R}{\cp(u)}$-many distinct $\SVdiamond
R$-arrows (according to $\mathfrak Q$), but some arrows
are still missing---only $ \Td{R}{u}{m}$ arrows exist in $\Delta_m$ and $\Tb{R}{u}{m}$ of those are in fact $\SVbox R$-arrows. Although $\SVbox R$-arrows are also $\SVdiamond R$-arrows,
their defects are repaired using a separate rule; and defects of $R$-arrows are dealt with as part of repairing defects of $\SVdiamond R$-arrows.  The following rules  extend $\Delta_m$ to $\Delta_{m+1}$ and each $S^{n,m}$ to $S^{n,m+1}$:
\begin{akrzlist}\itemsep=4pt
\item[$\LambdaB{S}{m}$] \ If $\Db{S}{u}{m}> 0$ then $\Rb{S}{\cp(u)} \geq 1$. By~\textbf{(Q$_3$)}, there is $r' \in \mathfrak{R}$ such that $\Rb{S^-}{r'}\geq 1$. We add $q = \Db{S}{u}{m}$ copies $v_1, \dots, v_q$ of the run $r'$ to $\Delta_{m+1}$ and  set $\cp(v_i) = r'$,  add $(u,v_i)$ to $S^{n,m+1}$, for all $n\in\Z$, and let $\nu_{v_i}\colon\Z\to Z$ be such that $\nu_{v_i}^{-1}(k)$ is infinite, for each $k\in Z$.

\item[$\LambdaD{S}{m}$] Let $K$ be $\bigl\{ i\mid 0 < i \leq \Dd{S}{u}{m} -\Db{S}{u}{m} \bigr\}$ if $\Dd{S}{u}{m}  < \omega$ and $\bigl\{ i\mid 0 < i \leq q_\K + 1 \bigr\}$ otherwise. By assumption,  $K\ne \emptyset$. We attach $|K|$
fresh $\SVdiamond S$-successors to $u$ so that the required
$\SVbox S$-, $S$- and $\SVdiamond S$-ranks coincide with the respective actual
ranks at step $m+1$. By~\textbf{(rn)} and~\textbf{(df)}, there exists a
function $\gamma\colon \Z \to 2^K$ such that,  for each $i\in K$, there are infinitely many $n_0\in \Z$ with $i \notin \gamma(n_0)$ and infinitely many $n_1 \in \Z$ with $i \in \gamma(n_1)$, and for all $n\in\Z$,
\begin{equation*}
|\gamma(n)|= \begin{cases}\Dr{S}{n}{u}{m}  - \Db{S}{u}{m}, & \text{ if }\Dr{S}{n}{u}{m}< \omega,\\
q_\K, & \text{ otherwise.}\end{cases}
\end{equation*}
By assumption, we have $\Rb{R}{\cp(u)} - \Tb{R}{u}{m} < \Rd{R}{\cp(u)} - \Td{R}{u}{m}$; by definition, $\Tb{R}{u}{m} \leq \Td{R}{u}{m}$ and, by~\textbf{(fn)}, $\Td{R}{u}{m}< \omega$, whence  $\Rb{S}{\cp(u)} < \Rd{S}{\cp(u)}$. Therefore, by~\textbf{(Q$_4$)}, there exists $r'\in\mathfrak{R}$ with $\Rb{S^-}{r'} < \Rd{S^-}{r'}$.
We add $|K|$ \emph{fresh copies} $v_1, \dots ,v_{|K|}$ of $r'$ to $\Delta_{m+1}$ and, for each $i\in K$, set $\cp(v_i) = r'$ and, for every $n\in\Z$, add $(u,v_i)$ to $S^{n,m+1}$ iff $i \in \gamma(n)$.  Let
\begin{equation*}
Z^{\Box S^-} \!\!= \bigl\{ k\in Z \mid \Rb{S^-}{r'} = \Rr{S^-\!\!}{\ k}{r'}  \bigr\}, \quad
Z^{\Diamond S^-} \!\!= \begin{cases}
\bigl\{ k\in Z \mid   \Rr{S^-\!\!}{\ k}{r'} = \Rd{S^-}{r'} \bigr\}, & \text{ if } \Rd{S^-}{r'} < \omega,\\
\emptyset, & \text{ otherwise.} 
\end{cases}
\end{equation*}
For each $v_i$, we take a function $\nu_{v_i}\colon \Z\to Z$ such that each $\nu_{v_i}^{-1}(k)$ is infinite, for $k\in Z$, and
\begin{akrzitemize}
\item if $k\in Z^{\Box S^-}$ then $i\notin\gamma(n)$, for each $n\in \nu^{-1}_{v_i}(k)$;
\item if $k\in Z\setminus (Z^{\Box S^-}\cup  Z^{\Diamond S^-})$ then $i\in\gamma(n)$ for  infinitely many $n\in \nu^{-1}_{v_i}(k)$ and $i\notin\gamma(n)$ for infinitely many $n\in \nu^{-1}_{v_i}(k)$;
\item if $k\in Z^{\Diamond S^-}$ then $i\in\gamma(n)$, for each $n\in \nu^{-1}_{v_i}(k)$
\end{akrzitemize}
(see Fig.~\ref{fig:quasimodel}). Intuitively, if $k$ is such that not
every $S$-predecessor is required to be a $\SVbox S$-predecessor then
there should be infinitely many copies of  $k$ with
$(u,v_i)\in S^{n,m+1}$; symmetrically, if $k$ is such that not
every $\SVdiamond S$-predecessor is required to be an $S$-predecessor,
there should be infinitely many copies of $k$ with
$(u,v_i)\notin S^{n,m+1}$.

\item[$\LambdaB{S^-}{m}$] and \textbf{($\LambdaD{S^-}{m}$)} are the mirror images of \textbf{($\LambdaB{S}{m}$)} and \textbf{($\LambdaD{S}{m}$)}, respectively. 
\end{akrzlist}
By construction, the rules guarantee that, for any $m \geq 0$ and $u\in\Delta_m$,
\begin{equation}\label{eq:QM:fx}
0 = \Db{R}{u}{m+1} \  =  \ \Dr{R}{n}{u}{m+1} \ = \ \Dd{R}{u}{m+1}, \text{ for all } R\in\role \text{  and all } n\in\Z.
\end{equation}

\begin{figure}[t]
\centering%
\begin{tikzpicture}[>=latex,point/.style={rectangle,draw=black,thick,inner sep=2.5pt,minimum size=1.4mm},lpoint/.style={circle,draw=black,thick,inner sep=1.5pt,minimum size=1.4mm},xscale=0.85,yscale=0.82]
\scriptsize
\foreach \x in {0,...,5} {
\draw[ultra thin,gray!50] (2*\x,-2.5) -- ++(0,5.8);
\node at (2*\x,-2.6) {\scriptsize $k_{\x}$};
}
\draw[thick,gray] (-1,0) -- (11,0);
\node[lpoint,fill=gray!70] (a0) at (0,0) {\textcolor{white}{\bf 2}};
\node[lpoint,fill=gray!70] (a1) at (2,0) {\textcolor{white}{\bf 2}};
\node[lpoint,fill=black] (a2) at (4,0) {\textcolor{white}{\bf 3}};
\node[lpoint,fill=gray!70] (a3) at (6,0) {\textcolor{white}{\bf 2}};
\node[lpoint,fill=black] (a4) at (8,0) {\textcolor{white}{\bf 3}};
\node[lpoint,fill=gray!70] (a5) at (10,0) {\textcolor{white}{\bf 2}};
\node at (-2.8,0) {\footnotesize $\Rb{S}{\cp(u)} \!\!= 1$,\  \ $\Rd{S}{\cp(u)} \!\!= 3$};
\node at (-1.6,1) {\footnotesize $v_2$};
\node at (-1.6,1.5) {\footnotesize $v_3$};
\node at (-1.6,-1) {\footnotesize $v_1$};
\draw[thick,gray] (-1,-1) -- (11,-1);
\foreach \x in {0,1,2,3,4,5} {
	\node[point,fill=white] (b\x) at (2*\x,-1) {\textcolor{black}{\bf 1}};
	\draw[->,thick] (a\x) to (b\x);
}
%
\draw[thick,gray] (-1,1) -- (11,1);
\draw[thick,gray] (-1,1.5) -- (11,1.5);
\foreach \x in {0,3} {
	\node[point,fill=white] (c\x) at (2*\x,1.5) {\bf 0};
}
\foreach \x in {3} {
	\node[point,fill=gray] (c\x) at (2*\x,1.5) {\textcolor{white}{\bf 1}};
}
\foreach \x in {1,2} {
	\node[point,fill=gray] (c\x) at (2*\x,1.5) {\textcolor{white}{\bf 1}};
	\draw[->,thick,out=30,in=-30,looseness=1.5] (a\x) to (c\x); 
}
\foreach \x in {4,5} {
	\node[point,fill=black] (c\x) at (2*\x,1.5) {\textcolor{white}{\bf 2}};
	\draw[->,thick,out=30,in=-30,looseness=1.5] (a\x) to (c\x); 
}
\foreach \x in {1} {
	\node[point,fill=white] (d\x) at (2*\x,1) {\bf 0};
}
\foreach \x in {5} {
	\node[point,fill=gray] (d\x) at (2*\x,1) {\textcolor{white}{\bf 1}};
}
\foreach \x/\c in {2,3,4} {
	\node[point,fill=gray] (d\x) at (2*\x,1) {\textcolor{white}{\bf 1}};
	\draw[->,thick,out=30,in=-30,looseness=1.5] (a\x) to (d\x); 
}
\foreach \x/\c in {0} {
	\node[point,fill=black] (d\x) at (2*\x,1) {\textcolor{white}{\bf 2}};
	\draw[->,thick,out=30,in=-30,looseness=1.5] (a\x) to (d\x); 
}
\node at (-3.3,2.5) {\footnotesize $\Rb{S^-}{r''} \!\! = 0$, \  $\Rd{S^-}{r''} \!\!= 2$};
\draw[rounded corners=1mm,gray,fill=gray!10] (-1.5,2.2) rectangle +(5,0.6);
\draw[ultra thin,gray] (-1.5,2.5) -- ++(5,0);
\node[point,fill=black] (p2) at (3,2.5) {\textcolor{white}{2}};
\node[point,fill=gray] (p1) at (1,2.5) {\textcolor{white}{1}};
\node[point,fill=white] (p0) at (-1,2.5) {\textcolor{black}{0}};
\foreach \n in {c0} { 
	\draw[dashed,gray] (p0.east) to (\n.north);
}
\foreach \n in {c1,c2,c3} { 
	\draw[dashed,gray] (p1.east) to (\n.north);
}
\foreach \n in {c4,c5} { 
	\draw[dashed,gray] (p2.east) to (\n.north);
}
\node at (-0.8,3.2) {\small $Z^{\Box S^-}$};
\node at (3.2,3.2) {\small $Z^{\Diamond S^-}$};
%
\node at (-3.3,-2) {\footnotesize $\Rb{S^-}{r'} \!\! = 1$, \  $\Rd{S^-}{r'} \!\!= 1$};
\draw[rounded corners=1mm,gray,fill=gray!10] (-1.5,-2.3) rectangle +(5,0.6);
\draw[ultra thin,gray] (-1.5,-2) -- ++(5,0);
\node[point,fill=white] (q2) at (3,-2) {\textcolor{black}{1}};
\node[point,fill=white] (q1) at (1,-2) {\textcolor{black}{1}};
\node[point,fill=white] (q0) at (-1,-2) {\textcolor{black}{1}};
\foreach \n in {b0,b1} { 
	\draw[dashed,gray] (q0.east) to (\n.south);
}
\foreach \n in {b2,b3} { 
	\draw[dashed,gray] (q1.east) to (\n.south);
}
\foreach \n in {b4,b5} { 
	\draw[dashed,gray] (q2.east) to (\n.south);
}
\end{tikzpicture}
\caption{Repairing defects of $u$: the required $S$-rank, $\smash{\Rr{S}{n}{\cp(u)}}$, of $u$ is specified inside the circular  nodes. Rule $(\Lambda^m_{\Box S})$ uses run $r'$ to create $v_1$; $(\Lambda^m_{\Diamond S})$ uses  copies of run $r''$ to create $v_2$ and $v_3$. The required $S^-$-rank, $\smash{\Rr{S^-\!\!}{\ n}{\cp(v_i)}}$, of the created points is specified inside the square nodes. Note that,   at instant $k_3$, the element $v_3$ requires an $S$-predecessor different from $u$.}\label{fig:quasimodel}
\end{figure}
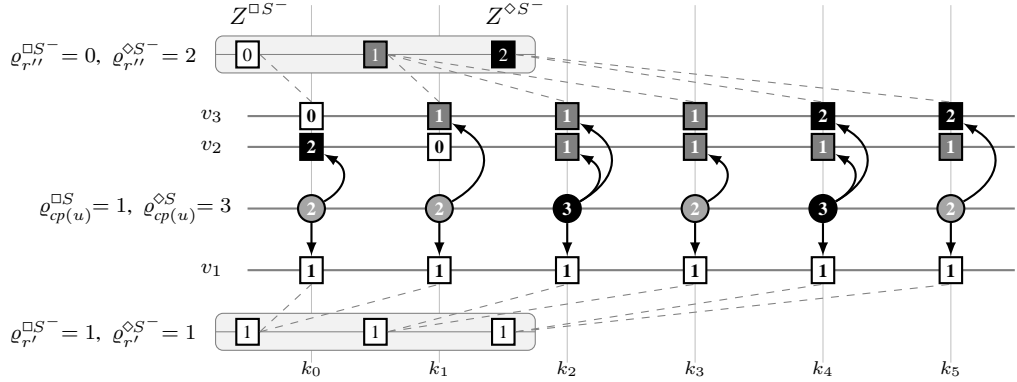

We now show that~\textbf{(fn)}, \textbf{(rn)} and \textbf{(df)} hold
for each $v\in\Delta_{m+1}\setminus\Delta_m$. Indeed, \textbf{(fn)}
holds because $\Td{R}{v}{m+1} \le 1$. In the case of \textbf{($\LambdaB{S}m$)}, property~\textbf{(rn)} follows from
\begin{align}
\label{eq:NP:witnessS-}
& 1 = \Tb{S^-}{v}{m+1} =\Tr{S^-\!\!}{\ n}{v}{m+1}=\Td{S^-}{v}{m+1}  \leq \Rb{S^-}{\cp(v)}\leq \Rr{S^-\!\!}{\ \nu_v(n)}{\cp(v)} \leq \Rd{S^-}{\cp(v)},\\
\label{eq:NP:witnessR}
& 0 = \Tb{R}{v}{m+1} =\Tr{R}{n}{v}{m+1}= \Td{R}{v}{m+1}  \leq \Rb{R}{\cp(v)}\leq \Rr{R}{\,\nu_v(n)}{\cp(v)}  \leq \Rd{R}{\cp(v)},\quad \text{ for all }R \ne S^-.
\end{align}
Then~\textbf{(df)} is by the definition of $\nu_{v}$. The case of~\textbf{($\LambdaB{S^-}m$)} is similar. For the case of~\textbf{($\LambdaD{S}{m}$)}, we observe  that, for each $R\ne S^-$, we have~\eqref{eq:NP:witnessR}, and so \textbf{(rn)} and \textbf{(df)} follow as above. Let us consider $S^-$.
By~\textbf{(r$_2$)}, both $Z\setminus Z^{\Diamond S^-}$ and
$Z\setminus Z^{\Box S^-}$ are non-empty. It follows that
$\Tb{S^-}{v}{m+1} = 0$ and $\Td{S^-}{v}{m+1} = 1$. By definition, we
also have~\textbf{(rn)}. To show~\textbf{(df)}, suppose $ \Rb{S^-}{\cp(v)} < \Rd{S^-}{\cp(v)} - 1$. Clearly,
$Z^{\Box S^-}\cap Z^{\Diamond S^-} = \emptyset$.  If there is $k \in
Z^{\Diamond S^-}$ then $\Rr{S^-\!\!}{\ k}{\cp(v)} = \Rd{S^-}{\cp(v)}$
and $\Tr{S^-\!\!}{\ n}{v}{m+1} = 1$, for all (infinitely many)
$n\in\nu^{-1}_v(k)$, whence the first item of~\textbf{(df)} holds;
otherwise, there is $k\in Z\setminus (Z^{\Box S^-}\cup Z^{\Diamond
  S^-})$ and therefore, $\Rb{S^-}{\cp(v)} < \Rr{S^-\!\!}{\ k}{\cp(v)}$ and
$\Tr{S^-\!\!}{\ n}{v}{m+1} = 0$, for infinitely many $n\in
\nu^{-1}_v(k)$, whence the first item
of~\textbf{(df)} holds.  The second item of~\textbf{(df)} is obtained by a
symmetric argument.

The definition of $\I$ is completed by taking $A^{\I(n)} = \bigl\{ u\in\Delta^\I \mid A\in r(\nu_u(n)), \ r = \cp(u)\bigr\}$, for each concept name $A$.
Observe that  $\I\models\extA$ because each $\nu_a$,  $a\in\ob$, coincides with the fixed $\nu$.
Next, we show by induction on the construction of concepts $C$ in $\K$ that
\begin{equation*}
C\in r(\nu_u(n)) \text{ with } r = \cp(u)\qquad\text{iff}\qquad u\in C^{\I(n)},\qquad\text{ for all } n\in\Z\text{ and } u\in \Delta^\I. 
\end{equation*}
The basis of induction is by definition for $C = \bot$ and $C = A_i$; for $C= \mathop{\geq q} R$ it follows from~\eqref{eq:QM:fx} and the fact that arrows to $u\in\Delta_{m}\setminus\Delta_{m-1}$ can be added only at steps $m$ and $m+1$ as part of the defect repair process. The induction step for $C= \neg C_1$ and $C = C_1\sqcap C_2$ follows from the induction hypothesis by~\textbf{(t$_1$)} and~\textbf{(t$_2$)}, respectively.
The induction step for $C = \SVbox C_1$ follows from the
induction hypothesis by~\textbf{(t$_1$)}, \textbf{(t$_3$)}, \textbf{(r$_1$)} and~\textbf{(r$_3$)}. Thus, by~\textbf{(Q$_1$)}, $\I\models\T$.

It remains to show $\I\models\A$. By the definition of $\extA$ and~$\I$, if $\nxt^k A(a) \in \A$ then $\I\models \nxt^k A(a)$ and if $\nxt^k \neg A(a) \in \A$ then $\I\models \nxt^k \neg A(a)$. If $\nxt^k S(a,b) \in \A$ then, by~\eqref{eq:NP:relation:basis}, $(a^\I,b^\I)\in S^{k,0}$, whence $\I\models \nxt^k S(a,b)$. If $\nxt^k \neg S(a,b) \in \A$ then $(a^\I,b^\I)\notin \extA$, whence $(a^\I,b^\I)\notin S^{k,0}$ by~\eqref{eq:NP:relation:basis}, and so $\I\models \nxt^k \neg S(a,b)$  as no new arrows can be added between ABox individuals.
\end{proof}

We are now in a position to establish the \NP{} membership of the
satisfiability problem for \TurDLbn{} KBs. To check whether a KB $\K = (\T,\A)$ is satisfiable, it is
enough to guess a structure $\mathfrak{Q} =
(Q,Z,\mathfrak{R},\extA)$ consisting of a set $\mathfrak{R}$ of runs  and an extension $\extA$ of the ABox
$\A$, both of which are of polynomial size in
$|\K|$, and check that $\mathfrak{Q}$ is a quasimodel for
$\K$. \NP-hardness follows from the complexity
of \bDLb{}. This completes the proof of Theorem~\ref{thm:R:NP}.\qed
\end{proof}
%

\section{Conclusions}
\label{sec:concl}

Logics interpreted over two- (or more) dimensional Cartesian products are notorious for their bad
computational properties, which is well-documented in the modal logic
literature (see~\cite{GKWZ03,Kurucz07} and references therein). For
example, satisfiability of bimodal formulas over Cartesian products of
transitive Kripke frames is undecidable~\cite{GabelaiaKWZ05}; by
dropping the requirement of transitivity we gain decidability, but not
elementary~\cite{GollerJL12}; if one dimension is a linear-time line then the complexity can only become worse~\cite{GKWZ03}.

The principal achievement of this article is the construction of
temporal description logics that (\emph{i}) are interpreted over 2D
Cartesian products, (\emph{ii}) are capable of capturing standard
temporal conceptual modelling constraints, and (\emph{iii}) in many
cases are of reasonable computational complexity.
Although TDLs \TdrDLbn{} and \TxDLbn, capturing lifespan cardinalities together with qualitative or quantitative evolution, turned out to be undecidable (as well as TDLs with unrestricted role inclusions), the complexity of the remaining ten logics ranges between  \NLogSpace{} and \PSpace{}.
We established these positive results by reductions to various clausal
fragments of propositional temporal logic (the complexity analysis of
which could be of interest on its own).  We have conducted initial
experiments, using two off-the-shelf temporal reasoning tools,
NuSmv~\cite{CCGPRST:02} and TeMP~\cite{HustadtKRV04}, which
showed feasibility of automated reasoning over TCMs with both
timestamping and evolution constraints but without sub-relations
(\TdxDLbn). Many efficiency issues are yet to be resolved
but the first results are encouraging.

The most interesting TDLs not considered in this article are probably
\TdxrDLcn{} and \TdxrDLkn. We conjecture that both of them are
decidable. We also believe that the former can be used as a variant of temporal RDFS (cf.~\cite{GutierrezHV05}). 

Although the results in this article establish tight complexity bounds for TDLs, they can only be used to obtain upper complexity bounds for the corresponding fragments of TCMs; the lower  bounds are mostly left for future work~\cite{AKRZ:ER10}. 

The original \DL{} family~\cite{CDLLR07} was designed with the primary aim of ontology-based data access (OBDA) by means of first-order query rewriting. In fact, OBDA has already reached a mature stage and  become a prominent direction in the development of the next generation of information systems and the Semantic Web; see~\cite{DBLP:conf/rweb/PolleresHDU13,DBLP:conf/rweb/KontchakovRZ13} for recent surveys and references therein. In particular, W3C has introduced a special profile, OWL~2~QL, of the Web Ontology Language OWL~2 that is suitable for OBDA and based on the \DL{} family. An interesting problem, both theoretically and practically, is to investigate how far this approach can be developed in the temporal case and what temporal ontology languages can support first-order query rewriting; see recent~\cite{Gutierrez-BasultoK12,DBLP:journals/ws/Motik12,ArtaleKWZ13,DBLP:conf/cade/BaaderBL13,DBLP:conf/frocos/BorgwardtLT13} for some initial results.

\appendixhead{URLend}


\bibliographystyle{acmsmall}
\bibliography{biblio}

\begin{thebibliography}{}

\bibitem[\protect\citeauthoryear{Apostol}{Apostol}{1976}]{Apostol76}
{\sc Apostol, T.} 1976.
\newblock {\em Introduction to Analytic Number Theory}.
\newblock Springer.

\bibitem[\protect\citeauthoryear{Artale, Calvanese, Kontchakov, Ryzhikov, and
  Zakharyaschev}{Artale et~al\mbox{.}}{2007a}]{ACKRZ:er07}
{\sc Artale, A.}, {\sc Calvanese, D.}, {\sc Kontchakov, R.}, {\sc Ryzhikov,
  V.}, {\sc and} {\sc Zakharyaschev, M.} 2007a.
\newblock Reasoning over extended {ER} models.
\newblock In {\em Proc.\ of the 26th Int.\ Conf.\ on Conceptual Modeling
  (ER'07)}. Lecture Notes in Computer Science Series, vol. 4801. Springer,
  277--292.

\bibitem[\protect\citeauthoryear{Artale, Calvanese, Kontchakov, and
  Zakharyaschev}{Artale et~al\mbox{.}}{2007b}]{ACKZ:aaai07}
{\sc Artale, A.}, {\sc Calvanese, D.}, {\sc Kontchakov, R.}, {\sc and} {\sc
  Zakharyaschev, M.} 2007b.
\newblock {DL-Lite} in the light of first-order logic.
\newblock In {\em Proc.\ of the 22nd Nat.\ Conf.\ on Artificial Intelligence
  (AAAI~2007)}. AAAI Press, 361--366.

\bibitem[\protect\citeauthoryear{Artale, Calvanese, Kontchakov, and
  Zakharyaschev}{Artale et~al\mbox{.}}{2009a}]{ACKZ:jair09}
{\sc Artale, A.}, {\sc Calvanese, D.}, {\sc Kontchakov, R.}, {\sc and} {\sc
  Zakharyaschev, M.} 2009a.
\newblock The {DL-Lite} family and relations.
\newblock {\em J.\ Artif.\ Intell.\ Res.\/}~{\em 36}, 1--69.

\bibitem[\protect\citeauthoryear{Artale and Franconi}{Artale and
  Franconi}{1998}]{ArFr98}
{\sc Artale, A.} {\sc and} {\sc Franconi, E.} 1998.
\newblock A temporal description logic for reasoning about actions and plans.
\newblock {\em J.\ Artif.\ Intell.\ Res.\/}~{\em 9}, 463--506.

\bibitem[\protect\citeauthoryear{Artale and Franconi}{Artale and
  Franconi}{1999}]{Ar:Fr:er-99}
{\sc Artale, A.} {\sc and} {\sc Franconi, E.} 1999.
\newblock Temporal {ER} modeling with description logics.
\newblock In {\em Proc.\ of the 18th Int.\ Conf.\ on Conceptual Modeling
  (ER'99)}. Lecture Notes in Computer Science Series, vol. 1728. Springer,
  81--95.

\bibitem[\protect\citeauthoryear{Artale and Franconi}{Artale and
  Franconi}{2001}]{ArFr01}
{\sc Artale, A.} {\sc and} {\sc Franconi, E.} 2001.
\newblock A survey of temporal extensions of description logics.
\newblock {\em Annals Math.\ and Artif.\ Intell.\/}~{\em 30,\/}~1--4, 171--210.

\bibitem[\protect\citeauthoryear{Artale and Franconi}{Artale and
  Franconi}{2005}]{AF05}
{\sc Artale, A.} {\sc and} {\sc Franconi, E.} 2005.
\newblock Temporal description logics.
\newblock In {\em Handbook of Temporal Reasoning in Artificial Intelligence}.
  Foundations of Artificial Intelligence. Elsevier, 375--388.

\bibitem[\protect\citeauthoryear{Artale and Franconi}{Artale and
  Franconi}{2009}]{artale:franconi:john09}
{\sc Artale, A.} {\sc and} {\sc Franconi, E.} 2009.
\newblock Foundations of temporal conceptual data models.
\newblock In {\em Conceptual Modeling: Foundations and Applications}. Lecture
  Notes in Computer Science Series, vol. 5600. Springer, 10--35.

\bibitem[\protect\citeauthoryear{Artale, Franconi, and Mandreoli}{Artale
  et~al\mbox{.}}{2003}]{AFM:lead03}
{\sc Artale, A.}, {\sc Franconi, E.}, {\sc and} {\sc Mandreoli, F.} 2003.
\newblock Description logics for modelling dynamic information.
\newblock In {\em Logics for Emerging Applications of Databases}. Springer,
  239--275.

\bibitem[\protect\citeauthoryear{Artale, Franconi, Wolter, and
  Zakharyaschev}{Artale et~al\mbox{.}}{2002}]{AFWZ:02}
{\sc Artale, A.}, {\sc Franconi, E.}, {\sc Wolter, F.}, {\sc and} {\sc
  Zakharyaschev, M.} 2002.
\newblock A temporal description logic for reasoning about conceptual schemas
  and queries.
\newblock In {\em Proc.\ of the 8th Joint European Conf.\ on Logics in
  Artificial Intelligence (JELIA-02)}. Lecture Notes in Artificial Intelligence
  Series, vol. 2424. Springer, 98--110.

\bibitem[\protect\citeauthoryear{Artale, Kontchakov, Lutz, Wolter, and
  Zakharyaschev}{Artale et~al\mbox{.}}{2007c}]{AKLWZ:time07}
{\sc Artale, A.}, {\sc Kontchakov, R.}, {\sc Lutz, C.}, {\sc Wolter, F.}, {\sc
  and} {\sc Zakharyaschev, M.} 2007c.
\newblock Temporalising tractable description logics.
\newblock In {\em Proc.\ of the 14th Int.\ Symposium on Temporal Representation
  and Reasoning (TIME07)}. IEEE Computer Society, 11--22.

\bibitem[\protect\citeauthoryear{Artale, Kontchakov, Ryzhikov, and
  Zakharyaschev}{Artale et~al\mbox{.}}{2009b}]{AKRZ:09}
{\sc Artale, A.}, {\sc Kontchakov, R.}, {\sc Ryzhikov, V.}, {\sc and} {\sc
  Zakharyaschev, M.} 2009b.
\newblock {DL-Lite} with temporalised concepts, rigid axioms and roles.
\newblock In {\em Proc.\ of the 7th Int.\ Symposium on Frontiers of Combining
  Systems (FroCoS-09)}. Lecture Notes in Computer Science Series, vol. 5749.
  Springer, 133--148.

\bibitem[\protect\citeauthoryear{Artale, Kontchakov, Ryzhikov, and
  Zakharyaschev}{Artale et~al\mbox{.}}{2010}]{AKRZ:ER10}
{\sc Artale, A.}, {\sc Kontchakov, R.}, {\sc Ryzhikov, V.}, {\sc and} {\sc
  Zakharyaschev, M.} 2010.
\newblock Complexity of reasoning over temporal data models.
\newblock In {\em Proc.\ of the 29th Int.\ Conf.\ on Conceptual Modeling
  (ER'10)}. Lecture Notes in Computer Science Series, vol. 4801. Springer,
  277--292.

\bibitem[\protect\citeauthoryear{Artale, Kontchakov, Wolter, and
  Zakharyaschev}{Artale et~al\mbox{.}}{2013}]{ArtaleKWZ13}
{\sc Artale, A.}, {\sc Kontchakov, R.}, {\sc Wolter, F.}, {\sc and} {\sc
  Zakharyaschev, M.} 2013.
\newblock Temporal description logic for ontology-based data access.
\newblock In {\em Proc.\ of the 23rd Int.\ Joint Conf.\ on Artificial
  Intelligence (IJCAI 2013)}. AAAI Press, 711--717.

\bibitem[\protect\citeauthoryear{Artale, Lutz, and Toman}{Artale
  et~al\mbox{.}}{2007d}]{ALT:ijcai07}
{\sc Artale, A.}, {\sc Lutz, C.}, {\sc and} {\sc Toman, D.} 2007d.
\newblock A description logic of change.
\newblock In {\em Proc.\ of the 20th Int.\ Joint Conf.\ on Artificial
  Intelligence (IJCAI-07)}. 218--223.

\bibitem[\protect\citeauthoryear{Artale, Parent, and Spaccapietra}{Artale
  et~al\mbox{.}}{2007e}]{APS:amai07}
{\sc Artale, A.}, {\sc Parent, C.}, {\sc and} {\sc Spaccapietra, S.} 2007e.
\newblock Evolving objects in temporal information systems.
\newblock {\em Annals Math.\ and Artif.\ Intell.\/}~{\em 50,\/}~1--2, 5--38.

\bibitem[\protect\citeauthoryear{Baader, Borgwardt, and Lippmann}{Baader
  et~al\mbox{.}}{2013}]{DBLP:conf/cade/BaaderBL13}
{\sc Baader, F.}, {\sc Borgwardt, S.}, {\sc and} {\sc Lippmann, M.} 2013.
\newblock Temporalizing ontology-based data access.
\newblock In {\em Proc.\ of the 24th Int.\ Conf.\ on Automated Deduction
  (CADE-24)}. Lecture Notes in Computer Science Series, vol. 7898. Springer,
  330--344.

\bibitem[\protect\citeauthoryear{Baader, Calvanese, McGuinness, Nardi, and
  Patel-Schneider}{Baader et~al\mbox{.}}{2003}]{BCMNP03}
{\sc Baader, F.}, {\sc Calvanese, D.}, {\sc McGuinness, D.}, {\sc Nardi, D.},
  {\sc and} {\sc Patel-Schneider, P.~F.}, Eds. 2003.
\newblock {\em The Description Logic Handbook: {T}heory, Implementation and
  Applications}.
\newblock Cambridge University Press.
\newblock (2nd edition, 2007).

\bibitem[\protect\citeauthoryear{Baader, Ghilardi, and Lutz}{Baader
  et~al\mbox{.}}{2008}]{BaaGhiLu-KR08}
{\sc Baader, F.}, {\sc Ghilardi, S.}, {\sc and} {\sc Lutz, C.} 2008.
\newblock {LTL} over description logic axioms.
\newblock In {\em Proc.\ of the 11th Int.\ Conf.\ on the Principles of
  Knowledge Representation and Reasoning (KR~2008)}. AAAI Press, 684--694.

\bibitem[\protect\citeauthoryear{Baader, Ghilardi, and Lutz}{Baader
  et~al\mbox{.}}{2012}]{BaaderGL12}
{\sc Baader, F.}, {\sc Ghilardi, S.}, {\sc and} {\sc Lutz, C.} 2012.
\newblock {LTL} over description logic axioms.
\newblock {\em {ACM} Trans.\ Computational Logic\/}~{\em 13,\/}~3.

\bibitem[\protect\citeauthoryear{Bauland, Schneider, Schnoor, Schnoor, and
  Vollmer}{Bauland et~al\mbox{.}}{2009}]{DBLP:journals/corr/abs-0812-4848}
{\sc Bauland, M.}, {\sc Schneider, T.}, {\sc Schnoor, H.}, {\sc Schnoor, I.},
  {\sc and} {\sc Vollmer, H.} 2009.
\newblock The complexity of generalized satisfiability for linear temporal
  logic.
\newblock {\em Logical Methods in Computer Science\/}~{\em 5,\/}~1.

\bibitem[\protect\citeauthoryear{Berardi, Calvanese, and De~Giacomo}{Berardi
  et~al\mbox{.}}{2005}]{BeCD05-AIJ-2005}
{\sc Berardi, D.}, {\sc Calvanese, D.}, {\sc and} {\sc De~Giacomo, G.} 2005.
\newblock Reasoning on {UML} class diagrams.
\newblock {\em Artif.\ Intell.\/}~{\em 168,\/}~1--2, 70--118.

\bibitem[\protect\citeauthoryear{Bettini}{Bettini}{1997}]{Bett97}
{\sc Bettini, C.} 1997.
\newblock Time dependent concepts: Representation and reasoning using temporal
  description logics.
\newblock {\em Data \& Knowledge Eng.\/}~{\em 22,\/}~1, 1--38.

\bibitem[\protect\citeauthoryear{B{\"o}rger, Gr{\"a}del, and
  Gurevich}{B{\"o}rger et~al\mbox{.}}{1997}]{BoGG97}
{\sc B{\"o}rger, E.}, {\sc Gr{\"a}del, E.}, {\sc and} {\sc Gurevich, Y.} 1997.
\newblock {\em The Classical Decision Problem}.
\newblock Perspectives in Mathematical Logic. Springer.

\bibitem[\protect\citeauthoryear{Borgida and Brachman}{Borgida and
  Brachman}{2003}]{BoBr03}
{\sc Borgida, A.} {\sc and} {\sc Brachman, R.~J.} 2003.
\newblock Conceptual modeling with description logics.
\newblock See \citeN{BCMNP03}, Chapter~10, 349--372.
\newblock (2nd edition, 2007).

\bibitem[\protect\citeauthoryear{Borgwardt, Lippmann, and Thost}{Borgwardt
  et~al\mbox{.}}{2013}]{DBLP:conf/frocos/BorgwardtLT13}
{\sc Borgwardt, S.}, {\sc Lippmann, M.}, {\sc and} {\sc Thost, V.} 2013.
\newblock Temporal query answering in the description logic {DL-Lite}.
\newblock In {\em Proc.\ of the 9th Int.\ Symposium on Frontiers of Combining
  Systems (FroCoS 2013)}. Lecture Notes in Computer Science Series, vol. 8152.
  Springer, 165--180.

\bibitem[\protect\citeauthoryear{Calvanese, De~Giacomo, Lembo, Lenzerini, and
  Rosati}{Calvanese et~al\mbox{.}}{2005}]{CDLLR05}
{\sc Calvanese, D.}, {\sc De~Giacomo, G.}, {\sc Lembo, D.}, {\sc Lenzerini,
  M.}, {\sc and} {\sc Rosati, R.} 2005.
\newblock {\sl DL-Lite}: {T}ractable description logics for ontologies.
\newblock In {\em Proc.\ of the 20th Nat.\ Conf.\ on Artificial Intelligence
  (AAAI~2005)}. AAAI Press, 602--607.

\bibitem[\protect\citeauthoryear{Calvanese, De~Giacomo, Lembo, Lenzerini, and
  Rosati}{Calvanese et~al\mbox{.}}{2007}]{CDLLR07}
{\sc Calvanese, D.}, {\sc De~Giacomo, G.}, {\sc Lembo, D.}, {\sc Lenzerini,
  M.}, {\sc and} {\sc Rosati, R.} 2007.
\newblock Tractable reasoning and efficient query answering in description
  logics: The {{\textit{DL-Lite}}} family.
\newblock {\em J.\ Aut.\ Reasoning\/}~{\em 39,\/}~3, 385--429.

\bibitem[\protect\citeauthoryear{Calvanese, De~Giacomo, Lenzerini, and
  Nardi}{Calvanese et~al\mbox{.}}{2001}]{CDLN-handbook-AR-2001}
{\sc Calvanese, D.}, {\sc De~Giacomo, G.}, {\sc Lenzerini, M.}, {\sc and} {\sc
  Nardi, D.} 2001.
\newblock Reasoning in expressive description logics.
\newblock In {\em Handbook of Automated Reasoning}. Vol.~II. Elsevier Science
  Publishers, 1581--1634.

\bibitem[\protect\citeauthoryear{Calvanese, Lenzerini, and Nardi}{Calvanese
  et~al\mbox{.}}{1999}]{calvanese-et-al:jair-99}
{\sc Calvanese, D.}, {\sc Lenzerini, M.}, {\sc and} {\sc Nardi, D.} 1999.
\newblock Unifying class-based representation formalisms.
\newblock {\em J.\ Artif.\ Intell.\ Res.\/}~{\em 11}, 199--240.

\bibitem[\protect\citeauthoryear{Chen and Lin}{Chen and Lin}{1993}]{ChenLin93}
{\sc Chen, C.-C.} {\sc and} {\sc Lin, I.-P.} 1993.
\newblock The computational complexity of satisfiability of temporal {Horn}
  formulas in propositional linear-time temporal logic.
\newblock {\em Information Processing Letters\/}~{\em 45,\/}~3, 131--136.

\bibitem[\protect\citeauthoryear{Chen}{Chen}{1976}]{Chen76}
{\sc Chen, P. P.-S.} 1976.
\newblock The {E}ntity-{R}elationship model---toward a unified view of data.
\newblock {\em {ACM} Trans.\ Database Syst.\/}~{\em 1}, 9--36.

\bibitem[\protect\citeauthoryear{Chomicki, Toman, and B{\"o}hlen}{Chomicki
  et~al\mbox{.}}{2001}]{ChomickiTB:01}
{\sc Chomicki, J.}, {\sc Toman, D.}, {\sc and} {\sc B{\"o}hlen, M.~H.} 2001.
\newblock Querying {ATSQL} databases with temporal logic.
\newblock {\em {ACM} Trans.\ Database Syst.\/}~{\em 26,\/}~2, 145--178.

\bibitem[\protect\citeauthoryear{Chrobak}{Chrobak}{1986}]{chrobak-ufa}
{\sc Chrobak, M.} 1986.
\newblock Finite automata and unary languages.
\newblock {\em Theor.\ Comput.\ Sci.\/}~{\em 47,\/}~2, 149--158.

\bibitem[\protect\citeauthoryear{Cimatti, Clarke, Giunchiglia, Giunchiglia,
  Pistore, Roveri, Sebastiani, and Tacchella}{Cimatti
  et~al\mbox{.}}{2002}]{CCGPRST:02}
{\sc Cimatti, A.}, {\sc Clarke, E.~M.}, {\sc Giunchiglia, E.}, {\sc
  Giunchiglia, F.}, {\sc Pistore, M.}, {\sc Roveri, M.}, {\sc Sebastiani, R.},
  {\sc and} {\sc Tacchella, A.} 2002.
\newblock {NuSMV}~2: An opensource tool for symbolic model checking.
\newblock In {\em Proc.\ of the 14th Int.\ Conf.\ on Computer Aided
  Verification (CAV'02)}. Lecture Notes in Computer Science Series, vol. 2404.
  Springer, 359--364.

\bibitem[\protect\citeauthoryear{Combi, Degani, and Jensen}{Combi
  et~al\mbox{.}}{2008}]{combi:et:al:er-08}
{\sc Combi, C.}, {\sc Degani, S.}, {\sc and} {\sc Jensen, C.~S.} 2008.
\newblock Capturing temporal constraints in temporal {ER} models.
\newblock In {\em Proc.\ of the 27th Int.\ Conf.\ on Conceptual Modeling
  (ER'08)}. Lecture Notes in Computer Science Series, vol. 5231. Springer,
  397--411.

\bibitem[\protect\citeauthoryear{Degtyarev, Fisher, and Konev}{Degtyarev
  et~al\mbox{.}}{2006}]{DegtyarevFK06}
{\sc Degtyarev, A.}, {\sc Fisher, M.}, {\sc and} {\sc Konev, B.} 2006.
\newblock Monodic temporal resolution.
\newblock {\em {ACM} Trans.\ Computational Logic\/}~{\em 7,\/}~1, 108--150.

\bibitem[\protect\citeauthoryear{Demri and Schnoebelen}{Demri and
  Schnoebelen}{2002}]{DBLP:journals/iandc/DemriS02}
{\sc Demri, S.} {\sc and} {\sc Schnoebelen, P.} 2002.
\newblock The complexity of propositional linear temporal logics in simple
  cases.
\newblock {\em Information and Compuation\/}~{\em 174,\/}~1, 84--103.

\bibitem[\protect\citeauthoryear{Dixon, Fisher, and Konev}{Dixon
  et~al\mbox{.}}{2007}]{DBLP:conf/ijcai/DixonFK07}
{\sc Dixon, C.}, {\sc Fisher, M.}, {\sc and} {\sc Konev, B.} 2007.
\newblock Tractable temporal reasoning.
\newblock In {\em Proc.\ of the 20th Int.\ Joint Conf.\ on Artificial
  Intelligence (IJCAI 07)}. 318--323.

\bibitem[\protect\citeauthoryear{Dolby, Fokoue, Kalyanpur, Ma, Schonberg,
  Srinivas, and Sun}{Dolby et~al\mbox{.}}{2008}]{DFKM*08}
{\sc Dolby, J.}, {\sc Fokoue, A.}, {\sc Kalyanpur, A.}, {\sc Ma, L.}, {\sc
  Schonberg, E.}, {\sc Srinivas, K.}, {\sc and} {\sc Sun, X.} 2008.
\newblock Scalable grounded conjunctive query evaluation over large and
  expressive knowledge bases.
\newblock In {\em Proc.\ of the 7th Int.\ Semantic Web Conf.\ (ISWC~2008)}.
  Lecture Notes in Computer Science Series, vol. 5318. Springer, 403--418.

\bibitem[\protect\citeauthoryear{Elmasri and Navathe}{Elmasri and
  Navathe}{2007}]{ElNa07}
{\sc Elmasri, R.~A.} {\sc and} {\sc Navathe, S.~B.} 2007.
\newblock {\em Fundamentals of Database Systems\/} 5th Ed.
\newblock Addison Wesley Publ.\ Co.

\bibitem[\protect\citeauthoryear{Finger and McBrien}{Finger and
  McBrien}{2000}]{finger:mcbrien:2000}
{\sc Finger, M.} {\sc and} {\sc McBrien, P.} 2000.
\newblock Temporal conceptual-level databases.
\newblock In {\em Temporal Logics -- Mathematical Foundations and Computational
  Aspects}. Oxford University Press, 409--435.

\bibitem[\protect\citeauthoryear{Fisher}{Fisher}{1991}]{DBLP:conf/ijcai/Fisher91}
{\sc Fisher, M.} 1991.
\newblock A resolution method for temporal logic.
\newblock In {\em Proc.\ of the 12th Int.\ Joint Conf.\ on Artificial
  Intelligence (IJCAI~91)}. Morgan Kaufmann, 99--104.

\bibitem[\protect\citeauthoryear{Fisher, Dixon, and Peim}{Fisher
  et~al\mbox{.}}{2001}]{DBLP:journals/tocl/FisherDP01}
{\sc Fisher, M.}, {\sc Dixon, C.}, {\sc and} {\sc Peim, M.} 2001.
\newblock Clausal temporal resolution.
\newblock {\em {ACM} Trans.\ Computational Logic\/}~{\em 2,\/}~1, 12--56.

\bibitem[\protect\citeauthoryear{Gabbay, Finger, and Reynolds}{Gabbay
  et~al\mbox{.}}{2000}]{GabbayFingerReynolds2000a}
{\sc Gabbay, D.}, {\sc Finger, M.}, {\sc and} {\sc Reynolds, M.} 2000.
\newblock {\em Temporal Logic: Mathematical Foundations and Computational
  Aspects}. Vol.~2.
\newblock Oxford University Press.

\bibitem[\protect\citeauthoryear{Gabbay, Hodkinson, and Reynolds}{Gabbay
  et~al\mbox{.}}{1994}]{Gabbayetal94}
{\sc Gabbay, D.}, {\sc Hodkinson, I.}, {\sc and} {\sc Reynolds, M.} 1994.
\newblock {\em Temporal Logic: Mathematical Foundations and Computational
  Aspects}. Vol.~1.
\newblock Oxford University Press.

\bibitem[\protect\citeauthoryear{Gabbay, Kurucz, Wolter, and
  Zakharyaschev}{Gabbay et~al\mbox{.}}{2003}]{GKWZ03}
{\sc Gabbay, D.}, {\sc Kurucz, A.}, {\sc Wolter, F.}, {\sc and} {\sc
  Zakharyaschev, M.} 2003.
\newblock {\em Many-dimensional modal logics: theory and applications}.
\newblock Studies in Logic. Elsevier.

\bibitem[\protect\citeauthoryear{Gabelaia, Kurucz, Wolter, and
  Zakharyaschev}{Gabelaia et~al\mbox{.}}{2005}]{GabelaiaKWZ05}
{\sc Gabelaia, D.}, {\sc Kurucz, A.}, {\sc Wolter, F.}, {\sc and} {\sc
  Zakharyaschev, M.} 2005.
\newblock Products of 'transitive' modal logics.
\newblock {\em Journal of Symbolic Logic\/}~{\em 70,\/}~3, 993--1021.

\bibitem[\protect\citeauthoryear{G{\"o}ller, Jung, and Lohrey}{G{\"o}ller
  et~al\mbox{.}}{2012}]{GollerJL12}
{\sc G{\"o}ller, S.}, {\sc Jung, J.~C.}, {\sc and} {\sc Lohrey, M.} 2012.
\newblock The complexity of decomposing modal and first-order theories.
\newblock In {\em Proc.\ of the 27th Annual IEEE Symposium on Logic in Computer
  Science (LICS 2012)}. IEEE, 325--334.

\bibitem[\protect\citeauthoryear{Gregersen and Jensen}{Gregersen and
  Jensen}{1998}]{gregersen:jensen:tr98}
{\sc Gregersen, H.} {\sc and} {\sc Jensen, J.} 1998.
\newblock Conceptual modeling of time-varying information.
\newblock Tech. Rep. Time{C}enter TR-35, Aalborg University, Denmark.

\bibitem[\protect\citeauthoryear{Gregersen and Jensen}{Gregersen and
  Jensen}{1999}]{gregersen:jensen:tkde-99}
{\sc Gregersen, H.} {\sc and} {\sc Jensen, J.} 1999.
\newblock Temporal {E}ntity-{R}elationship models---{A} survey.
\newblock {\em {IEEE} Trans.\ Knowledge and Data Eng.\/}~{\em 11,\/}~3,
  464--497.

\bibitem[\protect\citeauthoryear{Guensel}{Guensel}{2005}]{Guensel05}
{\sc Guensel, C.} 2005.
\newblock A tableaux-based reasoner for temporalised description logics.
\newblock Ph.D. thesis, University of Liverpool.

\bibitem[\protect\citeauthoryear{Guti{\'e}rrez, Hurtado, and
  Vaisman}{Guti{\'e}rrez et~al\mbox{.}}{2005}]{GutierrezHV05}
{\sc Guti{\'e}rrez, C.}, {\sc Hurtado, C.~A.}, {\sc and} {\sc Vaisman, A.~A.}
  2005.
\newblock Temporal {RDF}.
\newblock In {\em Proc.\ of the 2nd European Semantic Web Conf.\ (ESWC 2005)}.
  Lecture Notes in Computer Science Series, vol. 3532. Springer, 93--107.

\bibitem[\protect\citeauthoryear{Guti{\'e}rrez-Basulto and
  Klarman}{Guti{\'e}rrez-Basulto and Klarman}{2012}]{Gutierrez-BasultoK12}
{\sc Guti{\'e}rrez-Basulto, V.} {\sc and} {\sc Klarman, S.} 2012.
\newblock Towards a unifying approach to representing and querying temporal
  data in description logics.
\newblock In {\em Proc.\ of the 6th Int.\ Conf.\ on Web Reasoning and Rule
  Systems (RR 2012)}. Lecture Notes in Computer Science Series, vol. 7497.
  Springer, 90--105.

\bibitem[\protect\citeauthoryear{Hall and Gupta}{Hall and
  Gupta}{1991}]{gupta:hall:icde-91}
{\sc Hall, G.} {\sc and} {\sc Gupta, R.} 1991.
\newblock Modeling transition.
\newblock In {\em Proc.\ of the 7th Int.\ Conf.\ on Data Engineering
  (ICDE'91)}. IEEE Computer Society, 540--549.

\bibitem[\protect\citeauthoryear{Halpern and Reif}{Halpern and
  Reif}{1981}]{HalpernR81}
{\sc Halpern, J.~Y.} {\sc and} {\sc Reif, J.~H.} 1981.
\newblock The propositional dynamic logic of deterministic, well-structured
  programs (extended abstract).
\newblock In {\em Proc.\ of the 22nd Annual Symposium on Foundations of
  Computer Science (FOCS'81)}. IEEE Computer Society, 322--334.

\bibitem[\protect\citeauthoryear{Halpern and Shoham}{Halpern and
  Shoham}{1991}]{HaSh91}
{\sc Halpern, J.~Y.} {\sc and} {\sc Shoham, Y.} 1991.
\newblock A propositional modal logic of time intervals.
\newblock {\em J.\ ACM\/}~{\em 38,\/}~4, 935--962.

\bibitem[\protect\citeauthoryear{Heymans, Ma, Anicic, Ma, Steinmetz, Pan, Mei,
  Fokoue, Kalyanpur, Kershenbaum, Schonberg, Srinivas, Feier, Hench, Wetzstein,
  and Keller}{Heymans et~al\mbox{.}}{2008}]{HMAM*08}
{\sc Heymans, S.}, {\sc Ma, L.}, {\sc Anicic, D.}, {\sc Ma, Z.}, {\sc
  Steinmetz, N.}, {\sc Pan, Y.}, {\sc Mei, J.}, {\sc Fokoue, A.}, {\sc
  Kalyanpur, A.}, {\sc Kershenbaum, A.}, {\sc Schonberg, E.}, {\sc Srinivas,
  K.}, {\sc Feier, C.}, {\sc Hench, G.}, {\sc Wetzstein, B.}, {\sc and} {\sc
  Keller, U.} 2008.
\newblock Ontology reasoning with large data repositories.
\newblock In {\em Ontology Management, Semantic Web, Semantic Web Services, and
  Business Applications}. Vol.~7. Springer, 89--128.

\bibitem[\protect\citeauthoryear{Hodkinson, Wolter, and
  Zakharyaschev}{Hodkinson et~al\mbox{.}}{2000}]{hodk:wolter:mz:dec-tl:99}
{\sc Hodkinson, I.}, {\sc Wolter, F.}, {\sc and} {\sc Zakharyaschev, M.} 2000.
\newblock Decidable fragments of first-order temporal logics.
\newblock {\em Annals of Pure and Applied Logic\/}~{\em 106}, 85--134.

\bibitem[\protect\citeauthoryear{Hustadt, Konev, Riazanov, and
  Voronkov}{Hustadt et~al\mbox{.}}{2004}]{HustadtKRV04}
{\sc Hustadt, U.}, {\sc Konev, B.}, {\sc Riazanov, A.}, {\sc and} {\sc
  Voronkov, A.} 2004.
\newblock {TeMP}: A temporal monodic prover.
\newblock In {\em Proc.\ of the 2nd Int.\ Joint Conf.\ on Automated Reasoning
  (IJCAR 2004)}. Lecture Notes in Computer Science Series, vol. 3097. Springer,
  326--330.

\bibitem[\protect\citeauthoryear{Jensen and Snodgrass}{Jensen and
  Snodgrass}{1999}]{jensen:snodgrass:tkde-99}
{\sc Jensen, C.~S.} {\sc and} {\sc Snodgrass, R.~T.} 1999.
\newblock Temporal data management.
\newblock {\em {IEEE} Trans.\ Knowledge and Data Eng.\/}~{\em 111,\/}~1,
  36--44.

\bibitem[\protect\citeauthoryear{Kontchakov, Lutz, Wolter, and
  Zakharyaschev}{Kontchakov et~al\mbox{.}}{2004}]{KontchakovLWZ04}
{\sc Kontchakov, R.}, {\sc Lutz, C.}, {\sc Wolter, F.}, {\sc and} {\sc
  Zakharyaschev, M.} 2004.
\newblock Temporalising tableaux.
\newblock {\em Studia Logica\/}~{\em 76,\/}~1, 91--134.

\bibitem[\protect\citeauthoryear{Kontchakov, Rodriguez-Muro, and
  Zakharyaschev}{Kontchakov
  et~al\mbox{.}}{2013}]{DBLP:conf/rweb/KontchakovRZ13}
{\sc Kontchakov, R.}, {\sc Rodriguez-Muro, M.}, {\sc and} {\sc Zakharyaschev,
  M.} 2013.
\newblock Ontology-based data access with databases: A short course.
\newblock In {\em Reasoning Web. The 9th Int.\ Summer School on Semantic
  Technologies for Intelligent Data Access}. Lecture Notes in Computer Science
  Series, vol. 8067. Springer, 194--229.

\bibitem[\protect\citeauthoryear{Kurucz}{Kurucz}{2007}]{Kurucz07}
{\sc Kurucz, A.} 2007.
\newblock Combining modal logics.
\newblock In {\em Handbook of Modal Logic}, {P.~Blackburn}, {J.~{van B}enthem},
  {and} {F.~Wolter}, Eds. Studies in Logic and Practical Reasoning Series,
  vol.~3. Elsevier, 869--924.

\bibitem[\protect\citeauthoryear{Lichtenstein, Pnueli, and Zuck}{Lichtenstein
  et~al\mbox{.}}{1985}]{DBLP:conf/lop/LichtensteinPZ85}
{\sc Lichtenstein, O.}, {\sc Pnueli, A.}, {\sc and} {\sc Zuck, L.~D.} 1985.
\newblock The glory of the past.
\newblock In {\em Proc.\ of the Conf.\ on Logic of Programs}. Lecture Notes in
  Computer Science Series, vol. 193. Springer, 196--218.

\bibitem[\protect\citeauthoryear{Ludwig and Hustadt}{Ludwig and
  Hustadt}{2010}]{LudwigH10}
{\sc Ludwig, M.} {\sc and} {\sc Hustadt, U.} 2010.
\newblock Implementing a fair monodic temporal logic prover.
\newblock {\em AI Communications\/}~{\em 23,\/}~2--3, 69--96.

\bibitem[\protect\citeauthoryear{Lutz, Sturm, Wolter, and Zakharyaschev}{Lutz
  et~al\mbox{.}}{2002}]{LutzSWZ02}
{\sc Lutz, C.}, {\sc Sturm, H.}, {\sc Wolter, F.}, {\sc and} {\sc
  Zakharyaschev, M.} 2002.
\newblock A tableau decision algorithm for modalized {ALC} with constant
  domains.
\newblock {\em Studia Logica\/}~{\em 72,\/}~2, 199--232.

\bibitem[\protect\citeauthoryear{Lutz, Wolter, and Zakharyaschev}{Lutz
  et~al\mbox{.}}{2008}]{LuWoZa-TIME-08}
{\sc Lutz, C.}, {\sc Wolter, F.}, {\sc and} {\sc Zakharyaschev, M.} 2008.
\newblock Temporal description logics: A survey.
\newblock In {\em Proc.\ of the 15th Int.\ Symposium on Temporal Representation
  and Reasoning (TIME 08)}. IEEE Computer Society, 3--14.

\bibitem[\protect\citeauthoryear{Markey}{Markey}{2004}]{Markey04}
{\sc Markey, N.} 2004.
\newblock Past is for free: on the complexity of verifying linear temporal
  properties with past.
\newblock {\em Acta Informatica\/}~{\em 40,\/}~6--7, 431--458.

\bibitem[\protect\citeauthoryear{McBrien, Seltveit, and Wangler}{McBrien
  et~al\mbox{.}}{1992}]{mcbrien:et:al:cismod-92}
{\sc McBrien, P.}, {\sc Seltveit, A.}, {\sc and} {\sc Wangler, B.} 1992.
\newblock An {E}ntity-{R}elationship model extended to describe historical
  information.
\newblock In {\em Proc.\ of the Int.\ Conf.\ on Information Systems and
  Management of Data {(CISMOD'92)}}. Bangalore, India, 244--260.

\bibitem[\protect\citeauthoryear{Mendelzon, Milo, and Waller}{Mendelzon
  et~al\mbox{.}}{1994}]{Mendelzon:94}
{\sc Mendelzon, A.~O.}, {\sc Milo, T.}, {\sc and} {\sc Waller, E.} 1994.
\newblock Object migration.
\newblock In {\em Proc. of the 13th ACM SIGACT-SIGMOD-SIGART Symposium on
  Principles of Database Systems (PODS 94)}. ACM, 232--242.

\bibitem[\protect\citeauthoryear{Motik}{Motik}{2012}]{DBLP:journals/ws/Motik12}
{\sc Motik, B.} 2012.
\newblock Representing and querying validity time in {RDF} and {OWL}: A
  logic-based approach.
\newblock {\em J.\ Web Semantics\/}~{\em 12}, 3--21.

\bibitem[\protect\citeauthoryear{Ono and Nakamura}{Ono and
  Nakamura}{1980}]{OnoNakamura80}
{\sc Ono, H.} {\sc and} {\sc Nakamura, A.} 1980.
\newblock On the size of refutation {K}ripke models for some linear modal and
  tense logics.
\newblock {\em Studia Logica\/}~{\em 39}, 325--333.

\bibitem[\protect\citeauthoryear{Papadimitriou}{Papadimitriou}{1994}]{1994-papadimitriou}
{\sc Papadimitriou, C.~M.} 1994.
\newblock {\em Computational complexity}.
\newblock Addison-Wesley, Reading, Massachusetts.

\bibitem[\protect\citeauthoryear{Parent, Spaccapietra, and Zimanyi}{Parent
  et~al\mbox{.}}{2006}]{mads-book:06}
{\sc Parent, C.}, {\sc Spaccapietra, S.}, {\sc and} {\sc Zimanyi, E.} 2006.
\newblock {\em Conceptual Modeling for Traditional and Spatio-Temporal
  Applications---{T}he {MADS} Approach}.
\newblock Springer.

\bibitem[\protect\citeauthoryear{Plaisted}{Plaisted}{1986}]{Plaisted86}
{\sc Plaisted, D.} 1986.
\newblock A decision procedure for combinations of propositional temporal logic
  and other specialized theories.
\newblock {\em J.\ Aut.\ Reasoning\/}~{\em 2}, 171--190.

\bibitem[\protect\citeauthoryear{Poggi, Lembo, Calvanese, De~Giacomo,
  Lenzerini, and Rosati}{Poggi et~al\mbox{.}}{2008}]{PLCD*08}
{\sc Poggi, A.}, {\sc Lembo, D.}, {\sc Calvanese, D.}, {\sc De~Giacomo, G.},
  {\sc Lenzerini, M.}, {\sc and} {\sc Rosati, R.} 2008.
\newblock Linking data to ontologies.
\newblock {\em J.\ Data Semantics\/}~{\em X}, 133--173.

\bibitem[\protect\citeauthoryear{Polleres, Hogan, Delbru, and Umbrich}{Polleres
  et~al\mbox{.}}{2013}]{DBLP:conf/rweb/PolleresHDU13}
{\sc Polleres, A.}, {\sc Hogan, A.}, {\sc Delbru, R.}, {\sc and} {\sc Umbrich,
  J.} 2013.
\newblock {RDFS} and {OWL} reasoning for {L}inked {D}ata.
\newblock In {\em Reasoning Web. The 9th Int.\ Summer School on Semantic
  Technologies for Intelligent Data Access}. Lecture Notes in Computer Science
  Series, vol. 8067. Springer, 91--149.

\bibitem[\protect\citeauthoryear{Rabinovich}{Rabinovich}{2010}]{Rabi:10}
{\sc Rabinovich, A.} 2010.
\newblock Temporal logics over linear time domains are in {PSPACE}.
\newblock In {\em Proc. of the 4th Int.\ Workshop on Reachability Problems}.
  Lecture Notes in Computer Science Series, vol. 6227. Springer, 29--50.

\bibitem[\protect\citeauthoryear{Reynolds}{Reynolds}{2010}]{Rey:10}
{\sc Reynolds, M.} 2010.
\newblock The complexity of decision problems for linear temporal logics.
\newblock {\em Journal of Studies in Logic\/}~{\em 3,\/}~1, 19--50.

\bibitem[\protect\citeauthoryear{Schild}{Schild}{1993}]{Schild93}
{\sc Schild, K.} 1993.
\newblock Combining terminological logics with tense logic.
\newblock In {\em Proc.\ of the 6th Portuguese Conf.\ on Artificial
  Intelligence (EPIA'93)}. Springer, 105--120.

\bibitem[\protect\citeauthoryear{Schmiedel}{Schmiedel}{1990}]{schmiedel:90}
{\sc Schmiedel, A.} 1990.
\newblock A temporal terminological logic.
\newblock In {\em Proc.\ of the 8th National Conf.\ on Artificial Intelligence
  (AAAI'90)}. AAAI Press / The MIT Press, 640--645.

\bibitem[\protect\citeauthoryear{Sistla and Clarke}{Sistla and
  Clarke}{1982}]{SistlaClarke82}
{\sc Sistla, A.~P.} {\sc and} {\sc Clarke, E.~M.} 1982.
\newblock The complexity of propositional linear temporal logics.
\newblock In {\em Proc.\ of the 14th Annual ACM Symposium on Theory of
  Computing (STOC'82)}. ACM, 159--168.

\bibitem[\protect\citeauthoryear{Stockmeyer and Meyer}{Stockmeyer and
  Meyer}{1973}]{StockmeyerM73}
{\sc Stockmeyer, L.~J.} {\sc and} {\sc Meyer, A.~R.} 1973.
\newblock Word problems requiring exponential time: Preliminary report.
\newblock In {\em Proc.\ of the 5th Annual ACM Symposium on Theory of Computing
  (STOC'73)}. ACM, 1--9.

\bibitem[\protect\citeauthoryear{Su}{Su}{1997}]{Su97}
{\sc Su, J.} 1997.
\newblock Dynamic constraints and object migration.
\newblock {\em Theor.\ Comput.\ Sci.\/}~{\em 184,\/}~1-2, 195--236.

\bibitem[\protect\citeauthoryear{Tauzovich}{Tauzovich}{1991}]{tauzovich:er-91}
{\sc Tauzovich, B.} 1991.
\newblock Towards temporal extensions to the entity-relationship model.
\newblock In {\em Proc.\ of the 10th Int.\ Conf.\ on Conceptual Modeling
  (ER'91)}. ER Institute, 163--179.

\bibitem[\protect\citeauthoryear{Theodoulidis, Loucopoulos, and
  Wangler}{Theodoulidis et~al\mbox{.}}{1991}]{theodoulidis:et:al:is-91}
{\sc Theodoulidis, C.}, {\sc Loucopoulos, P.}, {\sc and} {\sc Wangler, B.}
  1991.
\newblock A conceptual modelling formalism for temporal database applications.
\newblock {\em Inf.\ Syst.\/}~{\em 16,\/}~3, 401--416.

\bibitem[\protect\citeauthoryear{To}{To}{2009}]{to-ufa}
{\sc To, A.~W.} 2009.
\newblock Unary finite automata vs. arithmetic progressions.
\newblock {\em Inf.\ Process.\ Lett.\/}~{\em 109,\/}~17, 1010--1014.

\bibitem[\protect\citeauthoryear{Wolter and Zakharyaschev}{Wolter and
  Zakharyaschev}{1999}]{WoZa99b}
{\sc Wolter, F.} {\sc and} {\sc Zakharyaschev, M.} 1999.
\newblock Modal description logics: Modalizing roles.
\newblock {\em Fundamenta Informatic\ae\/}~{\em 39}, 411--438.

\end{thebibliography}

\begin{received}
{September 2012}{October 2013}{February 2014}
\end{received}


\elecappendix

\setcounter{section}{1}

\section{Proof of Theorem~\ref{lem:qtli-equisat}}\label{app:proof:qtli-equisat}

\hspace*{1em}\textsc{Theorem~\ref{lem:qtli-equisat}.} \ {\itshape
 A \TuDLbn{} KB $\K  = (\T,\A)$ is satisfiable iff the \QTLi{} sentence
  $\K^\dagger$ is satisfiable.}

\begin{proof}
($\Leftarrow$) Let $\Mmf$ be a first-order temporal model with
  a \emph{countable} domain $D$ and $\Mmf,0 \models
  \K^\dagger$.
Without loss of generality we may assume that the $a^\Mmf$, for $a\in\ob$,
are all distinct.
We are going to construct a \TuDLbn{} interpretation $\I$
satisfying $\K$ and based on some domain $\Delta^\I$
that will be inductively defined as the union
\begin{equation*}
\Delta^\I = \bigcup\nolimits_{m \geq 0} \Delta_m,\quad\text{ where } \ \ \Delta_0 = \bigl\{ a^\Mmf \mid a\in\ob\bigr\} \subseteq D \ \ \text{ and } \ \ \Delta_m\subseteq \Delta_{m+1}, \text{ for } m \geq 0.
\end{equation*}
The interpretations of object names in $\I$ are given by their
interpretations in $\Mmf$: $a^{\I} = a^\Mmf \in
\Delta_0$. Each set $\Delta_{m+1}$, for $m\geq 0$, is constructed by adding to $\Delta_m$ some new elements that are fresh \emph{copies} of certain elements from $D\setminus \Delta_0$. If such a new element $u$ is a copy of $u'\in D\setminus \Delta_0$ then we write $\cp(u) = u'$, while for $u\in \Delta_0$ we let $\cp(u) = u$.

The interpretation $A^{\I(n)}$ of each concept name $A$ in
$\I$ is defined by taking
\begin{equation}\label{eq:atomic}
A^{\I(n)} = \bigl\{ u\in \Delta^\I \mid
\Mmf,n\models A^*[\cp(u)]\bigr\}.
\end{equation}
The interpretation $S^{\I(n)}$ of each role name $S$ in $\I$ is constructed inductively as the union
\begin{equation*}
S^{\I(n)} = \bigcup\nolimits_{m \geq  0} S^{n,m},\qquad\text{ where } S^{n,m} \subseteq \Delta_m \times
\Delta_m, \text{ for all } m \geq 0.
\end{equation*}
We require the following two definitions to guide our construction. The \emph{required $R$-rank $ \varrho_d^{R,n}$ of $d\in D$ at  moment
  $n$} is
\begin{equation*}
  \varrho_d^{R,n} ~=~ \max\bigl(\{0\} \cup \{q \in Q_\T \mid
  \Mmf,n \models E_qR[d]\}\bigr).
\end{equation*}
By~\eqref{eq:role:saturation}, $\varrho_d^{R,n}$ is a
function and if $\varrho_d^{R,n} = q$ then
$\Mmf,n\models E_{q'}R[d]$ for every $q' \in Q_\T$ with
$q' \leq q$, and $(\Mmf,n)\models \neg E_{q'}R[d]$ for every
$q' \in Q_\T$ with $q' > q$. We also define the \emph{actual
  $R$-rank $\tau_{u,m}^{R,n}$ of $u\in\Delta^\I$ at moment $n$ and step $m$} by
taking
\begin{equation*}
  \tau_{u,m}^{R,n} =  
 \max \bigl( \{0\} \cup \{q \in Q_\T \mid  (u,u_1),\dots,(u,u_q)\in   R^{n,m} \text{ for distinct } u_1,\dots,u_q\in\Delta^\I \} \bigr),
\end{equation*}
where $R^{n,m}$ is $S^{n,m}$ if $R = S$ and $\{ (u',u)\mid (u,u')\in S^{n,m} \}$ if $R = S^-$, for a role name $S$.

For the basis of induction, for each role name $S$,  we set
\begin{equation}\label{eq:model:relation:basis}
S^{n,0} = \bigl\{ (a^\I,b^\I)\in
\Delta_0\times \Delta_0 \mid  S(a,b) \in \A_n^S \bigr\},\qquad\text{ for } n \in\Z
\end{equation}
(note that $S(a,b)\in\A_n^S$ for all $n\in\Z$ if
$\nxt^k S(a,b)\in\A$, for a rigid role name $S$).  It follows
from the definition of $\A^\dagger$ that, for all $R \in \role$ and
$u\in \Delta_0$,
\begin{equation}\label{eq:model:rank:basis}
\tau_{u,0}^{R,n}  ~\leq~ \varrho_{\cp(u)}^{R,n}.
\end{equation}
Suppose that $\Delta_m$ and the $S^{n,m}$ have been defined
for some $m\ge 0$. If, for all roles $R$ and $u\in \Delta_m$, we had $\tau_{u,m}^{R,n} = \varrho_{\cp(u)}^{R,n}$ then the interpretation of roles would have been
constructed. However, in general this is not the case because there
may be some `defects' in the sense that the actual rank of some elements
is smaller than the required rank. Consider the following two sets of
defects in $S^{n,m}$:
\begin{equation*}
  \Lambda_R^{n,m} ~=~ \bigl\{ u \in \Delta_m\setminus \Delta_{m-1} \mid \tau_{u,m}^{R,n} < \varrho_{\cp(u)}^{R,n} \bigr\},\qquad\text{ for } R \in \{S, S^-\}
\end{equation*}
(for convenience, we assume $\Delta_{-1}= \emptyset$).
The purpose of, say, $\Lambda_S^{n,m}$ is to identify those
`defective' elements $u\in \Delta_m\setminus \Delta_{m-1}$ from which precisely
$\varrho_{\cp(u)}^{S,n}$ distinct $S$-arrows should start (according to
$\Mmf$), but some arrows are still missing (only $ \tau_{u,m}^{S,n}$
arrows exist). To `repair' these defects, we extend $\Delta_m$ to
$\Delta_{m+1}$ and $S^{n,m}$ to $S^{n,m+1}$ according to the following
rules:
\begin{akrzlist}\itemsep=10pt
\item[$\Lambda_S^{n,m}$] Let $u\in \Lambda_S^{n,m}$. Denote
  $d =\cp(u)$ and $q = \varrho_{\cp(u)}^{S,n} -\tau_{u,m}^{S,n}$. Then
  $\Mmf,n\models E_{q'}S[d]$ for some $q'\geq q >
  0$. By~\eqref{eq:role:saturation}, we have $\Mmf,n\models
  E_1S[d]$ and, by~\eqref{eq:role:existence:2}, there is
  $d'\in D$ such that $\Mmf,n\models E_1S^-[d']$. In this case
  we take $q$ \emph{fresh} copies $u'_1,\dots,u'_q$
  of $d'$ (so $\cp(u'_i) = d'$), add them to
  $\Delta_{m+1}$ and add the pairs $(u,u'_1),\dots,(u,u'_q)$ to
  $S^{n,m+1}$. If $S$ is rigid we add these pairs to all $S^{k,m+1}$, for $k\in\Z$.
\item[$\Lambda_{S^-}^{n,m}$] Let $u\in
  \Lambda_{S^-}^{n,m}$. Denote $d =\cp(u)$ and $q =
  \varrho_{\cp(u)}^{S^-,n} -\tau_{u,m}^{S^-,n}$. Then \mbox{$\Mmf,n\models
    E_{q'}S^-[d]$} for $q'\geq q >
  0$. By~\eqref{eq:role:saturation}, $\Mmf,n\models E_1S^-[d]$
  and, by~\eqref{eq:role:existence:2}, there is $d'\in D$ with
  $\Mmf,n\models E_1S[d']$. In this case we take $q$
  \emph{fresh} copies $u'_1,\dots,u'_q$ 
  of $d'$, add them to
  $\Delta_{m+1}$ and add the pairs $(u'_1,u),\dots,(u'_q,u)$ to
  $S^{n,m+1}$. If $S$ is rigid we add these pairs to all $S^{k,m+1}$, for $k\in\Z$.
\end{akrzlist}
Now we observe the
following property of the construction: for all $m_0 \geq 0$ and $u\in \Delta_{m_0}\setminus \Delta_{m_0-1}$,
\begin{equation}\label{eq:rm}
\tau_{u,m}^{R,n}  =\begin{cases} 0, & \text{if}\ m < m_0,\\
q, & \text{if}\ m = m_0, \text{ for some }  q \leq \varrho_{\cp(u)}^{R,n},\\
\varrho_{\cp(u)}^{R,n}, & \text{if}\ m > m_0.\end{cases}
\end{equation}
To prove this property, consider all possible cases. If $m < m_0$ then
$u\notin \Delta_m$, i.e., it has not been added to $\Delta_m$ yet,
and so $\tau_{u,m}^{R,n} = 0$.
If $m = m_0 = 0$ then $\tau_{u,m}^{R,n} \leq \varrho_{\cp(u)}^{R,n}$ by~\eqref{eq:model:rank:basis}. 
If $m = m_0 > 0$ then $u$
was added at step $m_0$ to repair a defect of some  $u'\in
\Delta_{m_0-1}$.
This means that either
$(u',u)\in S^{n,m_0}$ and $u'\in\Lambda_S^{n,m_0-1}$, or $(u,u')\in
S^{n,m_0}$ and $u'\in\Lambda_{S^-}^{n,m_0-1}$, for a role name $S$. Consider the first case.
Since \emph{fresh} witnesses $u$ are picked up every time the rule
$(\Lambda_S^{n,m_0-1})$ is applied and those witnesses
satisfy $\Mmf,n\models E_1S^-[\cp(u)]$, we obtain $\tau_{u,m_0}^{S,n} = 0$,
$\tau_{u,m_0}^{S^-,n} = 1$ and $\varrho_{\cp(u)}^{S^-,n} \geq 1$. The
second case is similar. If $m = m_0 + 1$ then all defects of $u$ are
repaired at step $m_0 + 1$ by applying the rules $(\Lambda_S^{n,m_0})$
and $(\Lambda_{S^-}^{n,m_0})$. Therefore, $\tau_{u,m_0}^{R,n} =
\varrho_{\cp(u)}^{R,n}$. If $m > m_0 + 1$ then~\eqref{eq:rm} follows
from the observation that no new arrows involving $u$ can be added
after step $m_0 + 1$.

It follows that, for all $R\in\role$, $q \in
Q_\T$, $n\in\Z$ and $u\in \Delta^\I$,
\begin{equation}\label{eq:model:expansion:exists}
\Mmf,n\models E_qR[\cp(u)] \quad \text{iff}\quad u\in (\mathop{\geq q} R)^{\I(n)}.
\end{equation}
Indeed, if $\Mmf,n\models E_qR[\cp(u)]$ then, by definition,
$\varrho_{\cp(u)}^{R,n}\geq q$. Let $u\in \Delta_{m_0}\setminus \Delta_{m_0 - 1}$.
Then, by~\eqref{eq:rm},
$\tau_{u,m}^{R,n} =\varrho_{\cp(u)}^{R,n}\geq q$, for all $m>m_0$. It follows from the
definition of $\tau_{u,m}^{R,n}$ and $R^{\I(n)}$ that $u\in(\mathop{\geq q}
R)^{\I(n)}$. Conversely, let $u\in(\mathop{\geq q}
R)^{\I(n)}$ and $u\in \Delta_{m_0}\setminus \Delta_{m_0-1}$.
Then, by~\eqref{eq:rm}, we
have $q\leq \tau_{u,m}^{R,n} = \varrho_{\cp(u)}^{R,n}$, for all $m>m_0$. So, by the
definition of $\varrho_{\cp(u)}^{R,n}$ and~\eqref{eq:role:saturation}, we
obtain $\Mmf,n\models E_qR[\cp(u)]$.

Now we show by induction on the construction of concepts $C$ in $\K$ that
\begin{equation*}
\Mmf,n\models C^*[\cp(u)]\qquad\text{iff}\qquad u\in C^{\I(n)},\qquad\text{ for all } n\in\Z\text{ and } u\in \Delta^\I. 
\end{equation*}
The basis of induction is trivial for $C = \bot$ and follows
from~\eqref{eq:atomic} if $C = A_i$ and
\eqref{eq:model:expansion:exists} if $C= \mathop{\geq q} R$. The
induction step for the Booleans ($C= \neg C_1$ and $C = C_1\sqcap
C_2$) and the temporal operators ($C = C_1\U C_2$ and $C= C_1\S C_2$) follows from the
induction hypothesis. Thus, $\I\models\T$.

It only remains to show that $\I\models\A$. If $\nxt^n A(a) \in \A$ then, by the definition of $\A^\dagger$ and~\eqref{eq:atomic}, $\I\models \nxt^n A(a)$. If $\nxt^n \neg A(a) \in \A$ then, analogously, $\I\models \nxt^n \neg A(a)$. If $\nxt^n S(a,b) \in \A$ then, by~\eqref{eq:model:relation:basis},
$(a^\I,b^\I)\in S^{n,0}$, whence, by
the definition of $S^{\I(n)}$,
$\I\models \nxt^n S(a,b)$. If $\nxt^n \neg S(a,b) \in
\A$ then, by~\eqref{eq:model:relation:basis},
$(a^\I,b^\I)\notin S^{n,0}$, and so,
as no new arrows can be added between ABox individuals,
$\I\models \nxt^n \neg S(a,b)$.

\smallskip
$(\Rightarrow)$ is straightforward.
\end{proof}


\section{Proof of Theorem~\ref{thm:pspace:2}}\label{app:PSpace}

\hspace*{1em}\textsc{Theorem~\ref{thm:pspace:2}.} \ {\itshape
The satisfiability problem for the core fragment of \TuDLbn{} KBs is \PSpace-complete.}

\begin{proof}
The proof is by reduction of the halting problem for Turing
 machines with a polynomial tape. We recall that, given a deterministic Turing machine  $M=\langle Q, \Gamma, \#, \Sigma, \delta, q_0, q_f
  \rangle$ and a polynomial $s(n)$, we construct a TBox $\T_M$ containing concept inclusions~\eqref{eq:tm:change:head:R}--\eqref{eq:tm:non-term}, which we list here for the reader's  convenience: 
\begin{align}
\tag{\ref{eq:tm:change:head:R}} H_{iq} & \sqsubseteq \bot \U
  H_{{(i+1)}q'},\quad H_{iq} \sqsubseteq \bot \U S_{ia'}, &&
  \text{if } \delta(q,a) = (q',a',R) \text{ and } i < s(n),\\
  \tag{\ref{eq:tm:change:head:L}} H_{iq} & \sqsubseteq \bot \U
  H_{{(i-1)}q'},\quad H_{iq} \sqsubseteq \bot \U S_{ia'}, && \text{if } \delta(q,a) = (q',a',L) \text{ and }
  i > 1,\\
  \tag{\ref{eq:tm:E}} H_{iq} & \sqsubseteq \bot \U D_i,\\
 \tag{\ref{eq:tm:neq:D}} D_i \sqcap  D_j & \sqsubseteq \bot,&& \text{if } i \neq j,\\
  \tag{\ref{eq:tm:preserve}} S_{ia} & \sqsubseteq S_{ia} \U  D_i,\\
  \tag{\ref{eq:tm:non-term}} H_{iq_f} & \sqsubseteq \bot.
\end{align}
For an input $\vec{a}=a_1\dots a_n$ of length $n$, we take the following ABox $\A_{\vec{a}}$:
\begin{equation*}
  H_{1q_0}(d),
  \qquad S_{ia_i}(d), \ \text{ for } 1 \leq i \leq n,\qquad
  S_{i\#}(d), \ \text{ for } n < i \leq s(n).
\end{equation*}
We show that $(\T_M,\A_{\vec{a}})$ is unsatisfiable
iff $M$ accepts $\vec{a}$. We represent configurations of $M$ as
tuples of the form $\mathfrak{c}=\langle b_1\dots b_{s(n)}, i, q
\rangle$, where $b_1\dots b_{s(n)}$ is the contents of the first
$s(n)$ cells of the tape with $b_j \in \Gamma$, for all $j$, the head
position is $i$, $1 \leq i \leq s(n)$, and $q\in Q$ is the control
state. Let $\I$ be an interpretation for
$\K_{M,\vec{a}}$. We say that $\I$ \emph{encodes
  configuration} $\mathfrak{c}=\langle b_1\dots b_{s(n)}, i, q
\rangle$ at moment $k$ if $d^\I\in H_{iq}^{\I(k)}$
and $d^\I\in S_{jb_j}^{\I(k)}$, for all $1 \leq j
\leq s(n)$. We note here that, in principle, many different
configurations can be encoded at moment $k$ in
$\I$. Nevertheless, any prefix of a model of
$(\T_M,\A_{\vec{a}})$ contains the computation of $M$ on the
given input $\vec{a}$:
\begin{lemma}\label{lemma:computation}
  Let $\mathfrak{c}_0, \dots, \mathfrak{c}_m$ be a sequence of
  configurations representing a partial computation of $M$ on
  $\vec{a}$. Then every model $\I$ of
  $(\T_M,\A_{\vec{a}})$ encodes $\mathfrak{c}_k$ at
  moment $k$, for $0 \leq k \leq m$.
\end{lemma}
\begin{proof}
  The proof is by induction on $k$. For $k=0$, the
  claim follows from $\I\models \A_{\vec{a}}$.
  For the induction step, let $\I$ encode
  $\mathfrak{c}_k=\langle b_1\dots b_i\dots b_{s(n)}, i, q \rangle$ at
  moment $k$, and let $\mathfrak{c}_{k+1}$ be $\langle b_1\dots b_i'
  \dots b_{s(n)}, i', q' \rangle$. Then we have $q \in Q\setminus
  \{q_f\}$. Consider first $\delta(q,b_i)=(q', b_i', L)$, in
  which case $i > 1$ and $i' = i -1$. Since $d^\I\in
  H_{iq}^{\I(k)}$ we have, by~\eqref{eq:tm:E},
  $d^\I\in D_i^{\I(k+1)}$ and,
  by~\eqref{eq:tm:change:head:L}, $d^\I\in
  H_{(i-1)q'}^{\I(k+1)}$ and
  $d^\I\in S_{ib_i'}^{\I(k+1)}$. Consider cell $j$,
  $1 \leq j \leq s(n)$, such that $j \ne i$. By~\eqref{eq:tm:neq:D},
  $d^\I\notin D_j^{\I(k+1)}$, and so, since
  $d^\I\in S_{jb_j}^{\I(k)}$, we obtain,
  by~\eqref{eq:tm:preserve}, $d^\I\in
  S_{jb_j}^{\I(k+1)}$. Hence, $\I$ encodes
  $\mathfrak{c}_{k+1}$ at moment $k+1$.  The case of
  $\delta(q,b_i)=(q', b_i', R)$ is analogous.
\end{proof}

It follows that if $M$ accepts $\vec{a}$ then
$(\T_M,\A_{\vec{a}})$ is unsatisfiable. Indeed, if $M$ accepts
$\vec{a}$ then the computation is a sequence of configurations
$\mathfrak{c}_0, \dots, \mathfrak{c}_m$ such that
$\mathfrak{c}_m=\langle b_1\dots b_{s(n)}, i, q_f \rangle$. Suppose
$(\T_M,\A_{\vec{a}})$ is satisfied in a model
$\I$. By
Lemma~\ref{lemma:computation}, $d^\I\in
H_{iq_f}^{\I(m)}$, which contradicts~\eqref{eq:tm:non-term}.

Conversely, if $M$ rejects $\vec{a}$ then $(\T_M,\A_{\vec{a}})$ is
satisfiable.  Let $\mathfrak{c}_0, \dots, \mathfrak{c}_m$ be a
sequence of configurations representing the rejecting computation of
$M$ on $\vec{a}$, $\mathfrak{c}_k=\langle b_{1,k},\dots, b_{s(n),k},
i_k, q_k \rangle$, for $0 \leq k \leq m$. We define an interpretation
$\I$ with $\Delta^\I= \{ d\}$, $d^\I = d$
and, for every $a\in\Gamma$, $q\in Q$, $1 \leq i \leq s(n)$ and $k
\geq 0$, we set (note that $q_m$ is a rejecting state and so, $\delta(q_f,a)$ is
undefined):
\begin{align*}
H_{iq}^{\I(k)} & = \begin{cases}
\Delta^\I, & \text{if } \ k \leq m, i=i_k \text{ and } q = q_k,\\
\emptyset, & \text{otherwise},
\end{cases} \\
S_{ia}^{\I(k)} & = \begin{cases}
\Delta^\I, & \text{if } \ k \leq m \text{ and } a=b_{i,k},\\
\Delta^\I, & \text{if } \ k = m + 1 \text{ and } a=b_{i,m},\\
\emptyset, & \text{otherwise},
\end{cases} \\
D_i^{\I(k)} & = \begin{cases}
\Delta^\I, & \text{if } \ 0 < k \leq m + 1 \text{ and }i=i_{k-1},\\
\Delta^\I, & \text{if } \ k = m + 2 \\
\emptyset, & \text{otherwise}.
\end{cases}
\end{align*}
It can be easily verified that $\I \models (\T_M,\A_{\vec{a}})$.
\end{proof}


\section{Proof of Theorem~\ref{thm:R:NP}}\label{app:NP}

\hspace*{1em}\textsc{Lemma~\ref{lem:finite:box}.} \
{\itshape
  Let $\K$ be a \TurDLbn{} KB and $q_\K =\max (\QT\cup\QA) + 1$. If $\K$ is satisfiable then it can be
  satisfied in an interpretation $\I$ such that $(\mathop{\geq
    q_\K} \SVbox R)^{\I} = \emptyset$, for each
  $R\in\role$.}

\begin{proof}
  Let $\I\models \K$. Without loss of generality, we
  will assume that the domain $\Delta^\I$ is at most countable.  
Construct a new interpretation $\I^*$ as follows. We take $\Delta^\I\times\N$ as the domain of
$\I^*$ and set $a^{\I^*} = (a^\I,0)$, for all $a\in\ob$. For
each $n\in\Z$, we set
\begin{align*}
A^{\I^*(n)} &= \{ (u,i) \mid u\in A^{\I(n)},\ i\in\N\},&& \text{ for every concept name } A,\\
S^{\I^*(n)} &= \{ ((u,i),(v,i)) \mid (u,v)\in S^{\I(n)},\ i\in\N\},&&\text{ for every role name }S.
\end{align*}
It should be clear that $\I^* \models \K$.

Suppose that
$u\in\Delta^\I$ has at least $q_\K$-many $\SVbox
R$-successors in $\I$ and assume that the pairs
\begin{equation*}
\bigl((u,i),(u_1,i)\bigr) ,\dots, \bigl((u,i),(u_{q_\K - 1},i)\bigr), \bigl((u,i),(u_{q_\K},i)\bigr) ,\dots
\end{equation*}
are all in $(\SVbox R)^{\I^*}$. We can also assume that if $(u_j,0) =
a^{\I^*}$, for some $a\in\ob$, then $j < q_\K$.
We then
rearrange some of the $R$-arrows of the form $((u,i),
(u_j,i'))$, simultaneously \emph{at all moments of time}, in the following
manner. We remove $((u,i), (u_j,i))$ from $(\SVbox
R)^{\I^*}$, for all $j$ and $i$ such that $j\geq q_\K$
or $i > 0$. Note that this operation does not affect the $\SVbox R$-arrows to the ABox individuals. To preserve the extension of concepts of the form $\mathop{\geq q} \SVbox R$, we
then add new $\SVbox R$-arrows of the form $((u,i), (u_j,i'))$, for $i > i' \geq 0$, to
$(\SVbox R)^{\I^*}$ in such a way that the following
conditions are satisfied:
\begin{itemize}
\item[--] for every $(u_j,i')$, there is precisely one $\SVbox R$-arrow of the form $((u,i), (u_j,i'))$,
\item[--] for every $(u,i)$, there are precisely $(q_\K-1)$-many $\SVbox R$-arrows of the form $((u,i), (u_j,i'))$.
\end{itemize}
Such a rearrangement is possible because $\I^*$ contains countably infinitely many copies of $\I$.
We leave it to the reader to check that the resulting interpretation is still a model of $\K$.

The rearrangement process is then repeated for each other $u\in\Delta^\I$ with at least $q_{\K}$-many $\SVbox R$-successors.
\end{proof}

\end{document}